\documentclass[a4paper,11pt]{article}
\usepackage[utf8]{inputenc}
\usepackage[T1]{fontenc}
\usepackage{authblk}
\usepackage{textcomp}
\usepackage{fancyhdr}
\usepackage{fullpage}
\usepackage{lmodern}
\usepackage{amsmath,amssymb,bm,bbm}
\usepackage{algorithm}
\usepackage{algorithmic}
\usepackage{float}
\usepackage{graphicx}
\usepackage{extarrows}
\usepackage{caption}
\usepackage{graphicx, subfig}
\usepackage{titling}
\usepackage{yhmath}
\usepackage{mathrsfs}
\usepackage{cite}
\usepackage{footnote}
\usepackage{amsthm}
\usepackage{color}
\usepackage{slashed}
\usepackage{verbatim}
\usepackage[colorlinks=true,linkcolor=red,citecolor=blue]{hyperref}
\usepackage{amsfonts,euscript}
\usepackage{latexsym, multicol}
\usepackage{epstopdf}
\usepackage{psfrag}
\usepackage{mathrsfs}
\usepackage{xcolor}
\usepackage{comment}
\usepackage{authblk}
\usepackage{indentfirst}
\usepackage{lipsum}
\usepackage{threeparttable}
\usepackage{tikz}
\usetikzlibrary{decorations.pathreplacing}
\DeclareFontFamily{U}{mathx}{\hyphenchar\font45}
\DeclareFontShape{U}{mathx}{m}{n}{
<5> <6> <7> <8> <9> <10>
<10.95> <12> <14.4> <17.28> <20.74> <24.88>
mathx10}{}
\DeclareSymbolFont{mathx}{U}{mathx}{m}{n}
\DeclareFontSubstitution{U}{mathx}{m}{n}
\DeclareMathAccent{\widecheck}{0}{mathx}{"71}

\numberwithin{equation}{section}
\DeclareMathOperator{\dvol}{dvol}
\def\eerr{{\mathcal{E}rr}}

\newcommand{\BBb}{\underline{\mathcal{B}}}
\newcommand{\bea}{\begin{eqnarray}}
\newcommand{\eea}{\end{eqnarray}}
\def\beaa{\begin{eqnarray*}}
\def\eeaa{\end{eqnarray*}}
\def\ba{\begin{array}}
\def\ea{\end{array}}
\def\be#1{\begin{equation}\label{#1}}
\def\eeq{\end{equation}}
\def\Gc{{\widecheck{G}}}
\def\gcm{{}^{(PG)}}
\def\dko{\dkb^{\leq 1}}

\def\dkk{\dkb^{\leq k}}
\def\Hbc{\widecheck{\Hb}}

\def\Mk{\mathfrak{M}}
\def\Bk{\mathfrak{B}}
\def\rS{r^\S}

\def\DDb{\overline{\DD}}
\def\fl{\mathcal{L}}

\def\Mint{{\, ^{(int)}\MM}}
\def\Mtop{{ \, ^{(top)} \MM}}

\def\Bkd{{}^*\Bk}

\def\ks{{k_{small}}}
\def\kl{{k_{large}}}

\def\sc{{}^{({LG})}}

\def\JpS{{J^{(p,\S)}}}
\def\Jpp{{J'}^{(p)}}

\def\Zscr{{\mathscr{Z}}}

\def\SSS{\mathbb{S}}
\def\hk{\mathfrak{h}}

\def\doo{\overset{\circ}{d}}

\def\ddo{\overset{\circ}{d}}
\def\Fb{{\underline{F}}}

\def\dec{\de_{dec}}
\def\dual{{}^*}
\def\Xb{{\underline{X}}}
\def\Mt{\widetilde{\MM}}
\def\vr{\varrho}

\def\RRR{\mathbb{R}}
\def\OO{O^+}

\def\Cb{{\underline{C}}}

\def\atrchb{{}^{(a)}\trchb}
\def\Gag{\Ga_g}
\def\Gab{\Ga_b}

\def\Xh{\widehat{X}}
\def\atrch{{}^{(a)}\trch}
\def\Xhb{\widehat{\underline{X}}}
\def\Deo{\overset{\circ}{\De}}

\def\ovla{\overset{\circ}{\la}}

\def\nabo{\overset{\circ}{\nab}}
\def\divo{\overset{\circ}{\sdiv}\,}
\def\curlo{\overset{\circ}{\curl}\,}

\def\Uscr{{\mathscr{U}}}

\def\Jp{J^{(p)}}

\def\CC{\mathcal{C}}
\def\Wbscrone{{\underline{\mathscr{W}}}_1}
\def\Wbscr{{\underline{\mathscr{W}}}_2}
\def\Jscr{\mathscr{J}}

\def\rhod{\dual\rho}
\def\DD{\mathcal{D}}
\def\mub{{\underline{\mu}}}

\def\trchc{\widecheck{\trch}}

\def\trXc{\widecheck{\tr X}}
\def\trXbc{\widecheck{\tr\Xb}}
\def\Pc{\widecheck{P}}
\def\Xib{\underline{\Xi}}
\def\Xbh{\widehat{\Xb}}
\def\DDov{{\ov{\DD}}}

\def\sk{\mathfrak{s}}

\def\qb{\ov{q}}

\def\D{\mathbf{D}}
\def\DD{\mathcal{D}}
\def\dkb{\slashed{\dk}}
\def\dko{\dkb^{\leq 1}}
\def\TT{\mathcal{T}}

\renewcommand{\c}{\cdot}
\def\Hb{\mathbb{H}}

\def\Hc{\widecheck{H}}

\DeclareMathOperator{\sdiv}{div}

\def\fb{{\underline{f}}}
\def\D{\mathbf{D}}

\def\EE{\mathcal{E}}

\def\th{\theta}

\newcommand{\ov}{\overline}

\def\NNN{\mathbb{N}}

\def\Xib{{\underline{\Xi}}}

\renewcommand{\O}{\mathbf{O}}
\newtheorem{thm}{Theorem}[section]

\newtheorem{prp}[thm]{Proposition}
\newtheorem{cor}[thm]{Corollary}
\newtheorem{lm}[thm]{Lemma}
\newtheorem{df}[thm]{Definition}

\newtheorem{rk}[thm]{Remark}
\newtheorem{remark}[thm]{Remark}

\newtheorem{prop}[thm]{Proposition}
\newtheorem{proposition}[thm]{Proposition}
\newtheorem{lem}[thm]{Lemma}
\newtheorem{lemma}[thm]{Lemma}

\renewcommand{\div}{\sdiv}
\def\AA{\mathcal{A}}
\def\om{\omega}

\def\hot{\widehat{\otimes}}

\newcommand{\ep}{\varepsilon}

\newcommand{\la}{\lambda}
\newcommand{\Up}{\Upsilon}

\newcommand{\ga}{\gamma}

\def\II{\mathscr{I}}
\def\les{\lesssim}
\allowdisplaybreaks[4]
\def\MM{\mathcal{M}}
\DeclareMathOperator{\curl}{curl}

\newcommand{\f}{\frac}

\DeclareMathOperator{\tr}{tr}
\DeclareMathOperator{\grad}{grad}

\def\The{\Theta}
\def\Thb{\underline{\The}}

\newcommand{\vphi}{\varphi}
\renewcommand{\a}{\alpha}
\renewcommand{\b}{\beta}
\newcommand{\Ga}{\Gamma}
\newcommand{\pr}{\partial}
\newcommand{\chib}{{\underline{\chi}}}
\newcommand{\etab}{{\underline{\eta}}}
\newcommand{\xib}{{\underline{\xi}}}
\newcommand{\omb}{{\underline{\omega}}}
\newcommand{\g}{{\bf g}}

\newcommand{\gt}{\widetilde{\g}}
\renewcommand{\aa}{\underline{\a}}
\renewcommand{\S}{{\mathbf{S}}}

\DeclareMathOperator{\lot}{l.o.t.}

\def\HH{\mathcal{H}}
\newcommand{\bb}{\underline{\b}}
\DeclareMathOperator{\err}{Err}

\numberwithin{equation}{section}
\def\M{\MM}

\def\trch{\tr\chi}
\def\trchb{\tr\chib}

\def\La{\Lambda}

\def\Si{\Sigma}

\def\LL{\mathcal{L}}

\def\hch{\widehat{\chi}}
\def\hchb{\widehat{\chib}}

\def\ombc{\widecheck{\omb}}

\def\ze{\zeta}
\def\nab{\nabla}
\def\De{\Delta}
\def\R{\mathbf{R}}

\def\BB{\mathcal{B}}

\def\JJ{\mathcal{J}}
\def\D{{\bf D}}
\def\de{\delta}
\def\si{\sigma}
\def\far{{}^{(far)}\M}

\def\ombc{\widecheck{\omb}}

\def\zec{\widecheck{\ze}}

\def\PP{\mathcal{P}}
\def\QQ{\mathcal{Q}}

\def\dk{\mathfrak{d}}

\def\Ab{{\underline{A}}}
\def\Bb{{\underline{B}}}
\def\Thb{{\underline{\Theta}}}
\def\Hb{\underline{H}}
\def\Xscr{\mathscr X}
\def\Xbscr{\underline{\mathscr X}}
\def\aXscr{{}^{(a)}\mathscr X}
\def\aXbscr{{}^{(a)}\underline{\mathscr X}}

\def\Ybscr{\underline{\mathscr Y}}
\def\Bbscr{\underline{\mathscr B}}
\def\Abscr{\underline{\mathscr A}}
\def\Fscr{\mathscr{F}}
\def\Hscr{\mathscr H}
\def\Mscr{\mathscr M}
\def\Mbscr{\underline{\mathscr M}}
\def\Pscr{\mathscr P}
\def\Qscr{\mathscr Q}

\def\JJ{\mathcal J}

\def\s2{\sqrt{2}}
\def\trXc{\widecheck{\tr X}}
\def\Zc{\widecheck{Z}}
\def\Jk{\mathfrak{J}}
\def\Mext{{}^{(ext)}\MM}

\def\chih{\widehat{\chi}}
\def\chibh{\widehat{\chib}}
\DeclareMathOperator{\KSAF}{KSAF^+}
\title{A canonical foliation on null infinity in perturbations of Kerr}
\author{Sergiu Klainerman, Dawei Shen, Jingbo Wan}
\date{\vspace{-5ex}}
\begin{document}
\maketitle
\vspace{-2ex}
\begin{center}{\it\large Dedicated\footnote{In recognition for their fundamental contributions connected to the subject matter of this paper.} to Demetrios Christodoulou and Roger Penrose.}
\end{center}
\vspace{0cm}
\begin{abstract}
Kerr stability for small angular momentum has been proved by Klainerman-Szeftel, Giorgi-Klainerman-Szeftel and Shen in the series of works \cite{KS:Kerr1,KS:Kerr2,KS:main,GKS,Shen}. Some of the most basic conclusions of the result, concerning various physical quantities on the future null infinity $\II^+$ are derived in Section 3.8 of \cite{KS:main}. Further important conclusions were later derived in \cite{AHS} and \cite{K:Chen}. In this paper, based on the existence and uniqueness results for GCM spheres of \cite{KS:Kerr1,KS:Kerr2}, we establish the existence of a canonical foliation on $\II^+$ for which the null energy, linear momentum, center of mass and angular momentum are well defined and satisfy the expected physical laws of gravitational radiation. The rigid character of this foliation eliminates the usual ambiguities related to these quantities in the physics literature. We also show that under the initial assumption of \cite{KS:main,GKS}, the center of mass of the black hole has a large deformation (recoil) after the perturbation.
\end{abstract}

{\centering\subsubsection*{\small Keywords}}
\noindent Kerr stability, future null infinity, outgoing PG foliation, intrinsic GCM spheres, geodesic foliation, LGCM foliation, regularity up to null infinity, supertranslation ambiguity, Bondi mass loss, center of mass, angular momentum, gravitational wave recoil, black hole kick.
\vspace{0.5cm}
\tableofcontents
\section{Introduction}
The definitions of conserved quantities such as mass and angular momentum have been among the most controversial problems of general relativity, see \cite{Wald,J:G} for a more recent survey. Among these challenges, a major issue is the definition of an angular momentum for a distant observer at null infinity. The difficulty, rooted in the equivalence principle, is due to the absence of a local energy-momentum density for the gravitational field. Indeed, the symmetric energy-momentum of matter $T_{\mu\nu}$\footnote{$T_{\mu\nu}=\frac{\pr L_M}{\pr\g^{\mu\nu}}-\frac 1 2\g_{\mu\nu}L_M$.}, verifying the local conservation laws $\D^\nu T_{\mu\nu}=0$, is obtained by taking the variation of the corresponding action integral $\int L_M\dvol_\g$ of matter. The conserved quantities can then be defined, in causal regions of a spacetime that admits a Killing vectorfield $X$, by integrating the divergence $\D^\mu P_\mu=0$, with $P_\mu=T_{\mu\nu} X^\nu$. In the particular case of a Kerr spacetime this leads to a well defined of local versions of energy and angular momentum for the corresponding matter field. The same procedure applied to the gravitational action $\int\R\dvol_\g$ leads, however, to the Einstein tensor $\R_{\mu\nu}-\frac 1 2\g_{\mu\nu}\R$, that is, the left hand side of the actual Einstein field equations\footnote{We note however the existence of higher-order gravitational energy-momentum objects, such as the Bel-Robinson tensor at the level of  the Riemann curvature tensor, see \cite{Sen} and the references within. The Bel-Robinson tensor was first used to derive bounds for energy-type curvature quantities in \cite{Ch-Kl}. Nonsymmetric coordinate-dependent notions of energy-momentum (pseudo-tensors) can also be defined, see \cite{J:G} and references within}.  

In view of these difficulties, mathematical physicists have given up on a local notion of energy-momentum for general gravitational fields and restrict instead to asymptotically flat metrics where the null energy and angular momentum can be defined as limits of appropriate quantities at spacelike and null infinity. The definitions of the correct limits (ADM quantities) at $i_0$ have been well understood since the seminal \cite{ADM} paper. The first definitions at null infinity, due to \cite{Bondi,Sachs,Pen}, have generated much more controversy due to the ad hoc, non dynamical, definition of null infinity.  Indeed, to be able to define the Bondi mass and angular momentum, one has to make specific asymptotic assumptions about the behavior of the metric at null infinity, which can only be verified by a rigorous, evolutionary, mathematical construction from prescribed initial data. The first such result, for general perturbations of Minkowski spacetime, is due to \cite{Ch-Kl}. In the context of general perturbations of a slowly rotating Kerr spacetime, this was first achieved in the sequence of papers \cite{KS:main,GKS,KS:Kerr1,KS:Kerr2,Shen}.

Prior to such a construction, the definition of the main physical quantities at null infinity, especially the angular momentum\footnote{The definition of angular momentum  for a distant observer at null infinity has been particularly controversial.} remained subject to possible ambiguities due to the presence of supertranslations, an infinite-dimensional subgroup of the Bondi-Metzner-Sachs (BMS) group\footnote{More generally, this is the group of asymptotic symmetry transformations that leave  invariant the boundary conditions  at null infinity, for a general asymptotic flat Einstein vacuum spacetime. In \cite{Bondi,Sachs}, it was shown that the group is independent of the particular gravitational field and that it contains, in addition to the Poincar\'e group, an infinite dimensional group of \emph{supertranslations}.} of transformations that leave invariant the (ad hoc) boundary conditions at null infinity. In fact, as Penrose clarified in \cite{PenroseUnsolved}, the notion of \emph{angular momentum carried away by gravitational radiation} can be altered by supertranslations. Therefore, it is crucial to establish a rigorous definition of angular momentum that is free of supertranslation ambiguity We refer to \cite{J:G} for traditional attempts to define the null angular momentum and to \cite{CWY,CWWY,KS:Kerr2,Rizzi} for recent mathematical approaches to the problem. To address this ambiguity,  which can be traced down to the freedom of choosing specific sections of $\II^+$,  there are two possible approaches: 
\begin{enumerate}
    \item Modify the definitions of angular momentum and center of mass to ensure supertranslation invariance.
    \item Find a canonical foliation on $\II^+$ that anchors these definitions in a rigid fashion.
\end{enumerate}
The first approach was taken by Chen-Keller-Wang-Wang-Yau \cite{CKWWY,CWWY}, whose definition of angular momentum ($J_{CWY}$) is supertranslation invariant between various foliations on $\II^+$ in the Bondi-Sachs gauge. Notably, the expression of $J_{CWY}$ was first discovered by taking the limit of an appropriate notion of quasilocal angular momentum for compact 2D surfaces of a given spacetime; see \cite{CWY,KeWaYa}.

In this paper, we rely instead on the constructive approach introduced in the context of Schwarzschild \cite{KS} and Kerr stability \cite{KS:main}. We recall, see Section \ref{section:Recall1}, that the final Kerr spacetime was derived in \cite{KS:main} as a limit of finite GCM admissible spacetimes. These  come equipped  with  specified GCM   foliations  and  quasilocal  notions of  mass and angular momentum, defined on their essential boundaries $\Si_*$,\footnote{The spacelike GCM hypersurfaces $\Si_*$ were constructed in \cite{Shen} based on the GCM spheres of \cite{KS:Kerr1,KS:Kerr2}. All these were first introduced, albeit in the special case of axially symmetric polarized spacetimes, in \cite{KS}.} see Figure \ref{fig1-introd1} below. We give a rigorous definition of future null infinity $\II^+$ as the limit of a sequence of incoming null cones $\Cb_n$ and show that $\II^+$, defined constructively in this way, inherits a canonical GCM foliation and a canonical definition of the null angular momentum compatible with the foliation. We show that the definition is unique, thus circumventing the problem of supertranslation ambiguity. the null energy, linear momentum, and center of mass on $\II^+$ are defined in a similar way. Moreover, the familiar physical laws for these quantities hold in the canonical GCM foliation. Finally we show that under the initial assumption of \cite{KS,KS:main}, the center of mass of the final black hole has a large deformation (recoil) relative to the reference Kerr spacetime before perturbation.

As can be seen from the text above, the General Covariant Modulation (GCM) procedure plays a fundamental role in our approach, as it addresses three key aspects of the Kerr black hole stability problem. From an analytic point of view, it is an infinite-dimensional modulation procedure designed to handle the full diffeomorphism group. group of Einstein equations, making it possible to treat the stability problem as a Cauchy evolution problem. Geometrically, it provides a rigid, canonical, foliation at null infinity that eliminates the above-mentioned supertranslation ambiguities. From a physical point of view, the GCM procedure determines the center-of-mass frame of the final black hole, providing the best adapted framework for the definition of the main physical quantities at null infinity and the derivation of their evolution laws.
\subsection{Kerr stability for small angular momentum}\label{section:Recall1}
The main result stated in \cite{KS:main} and proved in the sequence of papers   \cite{KS:main, GKS, KS:Kerr1, KS:Kerr2, Shen} can be summarized as follows.
\begin{thm}\label{MainThm-firstversion}
The future globally hyperbolic development of a general asymptotically flat initial data set, sufficiently close (in a suitable topology) to a $Kerr(a_0, m_0)$ initial data set, for sufficiently small $|a_0|/m_0$,  has a complete future null infinity $\II^+$ and converges in its causal past $\JJ^{-1}(\II^{+})$ to another nearby Kerr spacetime $Kerr(a_f, m_f)$ with parameters $(a_f, m_f)$ close to the initial ones $(a_0, m_0)$.
\end{thm} 
\begin{figure}[H]\label{fig0-introd}
\centering
\includegraphics[scale=0.8]{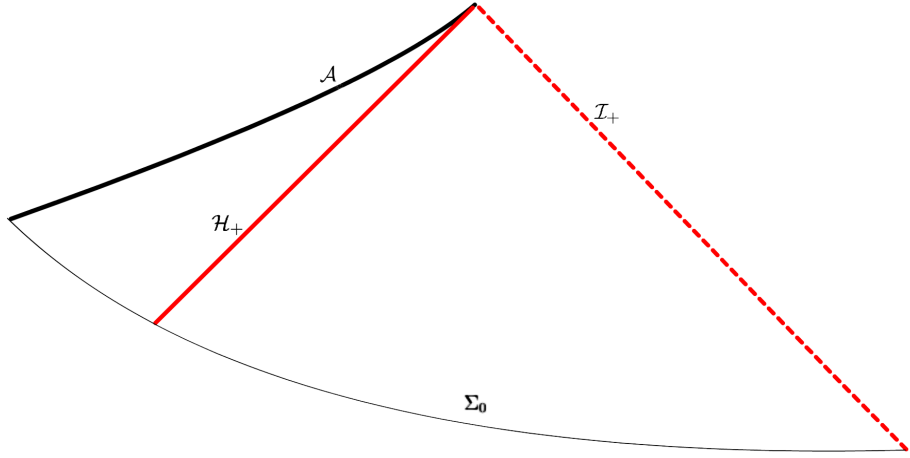}
\caption{\small{The Penrose diagram of the final space-time in Theorem \ref{MainThm-firstversion} with initial hypersurface $\Si_0$, future space-like boundary $\AA$, and $\II^+$ the complete future null infinity. The hypersurface $\HH_+$} is the future event horizon of the final Kerr.}
\end{figure} 
The proof of Theorem \ref{MainThm-firstversion} is based on a limiting argument\footnote{For the reader interested in a more in depth account the slowly rotating Kerr stability result discussed here, including a short description of its history, we refer to \cite{KS:short}. The following paragraph is essentially the same as that in Section 1.3.2 of \cite{KS:short}.} for a continuous family of such spacetimes $\MM=\Mext\cup\Mtop\cup\Mint$, represented graphically in Figure \ref{fig1-introd1}, together with a set of bootstrap assumptions (BA)  for the connection and curvature coefficients, relative to the adapted frames. The main features of these spacetimes are as follows:
\begin{figure}[ht!]
\centering
\includegraphics[scale=1]{kerr_1.pdf}
\caption{\small{The Penrose diagram of a finite GCM admissible space-time $\MM=\Mext\cup\Mtop\cup\Mint$. The future boundary $\Si_*$ initiates at the GCM sphere $S_*$. The past boundary of $\MM$, $\BB_1\cup\BBb_1$, is included in the initial layer $\LL_0$, in which the spacetime is assumed given.}}
\label{fig1-introd1} 
\end{figure} 
\begin{itemize}
\item The capstone of the entire construction is the sphere $S_*$, on the future boundary $\Si_*$ of $\Mext$, which verifies a set of specific extrinsic and intrinsic conditions denoted by the acronym\footnote{Short hand for general covariant modulated.} GCM.
\item The spacelike hypersurface $\Si_*$, initialized at $S_*$, verifies a set of additional GCM conditions.
\item Once $\Si_*$ is specified the whole GCM admissible spacetime $\MM$ is determined by a more conventional construction, based on geometric transport type equations\footnote{More precisely $\Mext$ can be determined  from $\Si_*$ by a specified  outgoing  foliation terminating in the timelike boundary $\TT$, $\Mint$ is determined from $\TT$ by a specified incoming one, and $\Mtop$ is a complement of $\Mext\cup\Mint$ which makes $\MM$ a  causal domain.}.  
\item The construction, which also allows us to specify adapted null frames, is made possible by the covariance properties of the Einstein vacuum equations.
\item The past boundary $\BB_1\cup\BBb_1$ of $\MM$, which is itself to be constructed, is included in the initial layer $\LL_0$ in which the spacetime is assumed to be known\footnote{The passage form  the initial data, specified on the initial spacelike hypersurface $\Si_0$, to the initial layer spacetime $\LL_0$, can be justified by arguments similar to those of \cite{Kl-Ni,knpeeling,Caciotta}, see \cite{ShenKerr}.}, i.e. a small vacuum  perturbation of a Kerr solution.
\end{itemize}
Assuming that a given finite, GCM admissible spacetime $\MM$ saturates the bootstrap assumptions BA we reach a contradiction as follows:
\begin{itemize}
\item First, improve the BA for some of the components of the curvature tensor with respect to the frame. These verify equations (called Teukolsky equations) that decouple, up to terms quadratic in the perturbation, and are treated by wave equation methods.
\item Use the information provided by these curvature coefficients together with the gauge choice on $\MM$, induced by the GCM condition on $\Si_*$, to improve BA for all other Ricci and curvature components.
\item Use these improved estimates to extend $\MM$ to a strictly larger spacetime $\MM'$ and then construct a new GCM sphere $S_*'$, a new boundary $\Si_*'$ that initiates on $S_*'$, and a new GCM admissible spacetime $\MM'$, with $\Si_*'$ as the boundary, strictly larger than $\MM$.
\item The final spacetime $\big(\MM_{\infty},\g_{\infty}\big)$,  $\MM_{\infty}=\Mext\cup\Mint $, derived in the proof of Theorem \ref{MainThm-firstversion}, is equipped with specified PG structures\footnote{See Section 2.4 in \cite{KS:main} for the precise definition of an outgoing PG structure.} in both $\Mext$ and $\Mint$. In particular,  $\Mext$  has an outgoing PG structure $\big\{(e_3,e_4;\HH), r, \th, u)\big\}$, as well as a set of $\ell=1$ basis $\Jp$, for $p=+,0,-$, relative to each, the null connections and the curvature components have specified decay properties. 
\end{itemize}
\begin{remark}\label{remark:final state1}
The final state $\MM_{\infty} $ verifies all the relevant properties described in Section 3.4 in \cite{KS:main} in which the precise version of Theorem \ref{MainThm-firstversion} is given. In this paper, we are only interested in the limits at future null infinity, that is, we are only interested in $\Mext_{\infty}$, the external part of $\MM_{\infty}$. We refer to these as \emph{$\KSAF$ spacetimes},\footnote{When we want to precise the final parameters, we refer to these as $\KSAF(a_f,m_f)$ spacetime with $(a_f,m_f)$ the final parameters in Theorem \ref{MainThm-firstversion}.} (see Definition \ref{dfKSAF} for more explanations). Since there is no danger of confusion, we simply denote them by $\MM$. Thus, a $\KSAF$ spacetime $\MM$ is future null complete and verifies all the properties of $\Mext_\infty$ described in the main theorem\footnote{Note however that all the results  concerning $\KSAF$ spacetimes, including those derived here, hold for the full sub-extremal range  $|a|<m$.} of \cite{KS:main}. In particular $\KSAF$ spacetimes verify all the assumptions made in \cite{KS:Kerr1,KS:Kerr2,Shen} concerning the existence of GCM spheres and GCM hypersurfaces.
\end{remark}
\begin{remark}\label{remark:final state2}
Section 3.8 of \cite{KS:main} contains various conclusions concerning the completeness of $\II^+$, the existence of a future event horizon, the definition of Bondi mass and angular momentum at null infinity, and the derivation of the Bondi mass loss formula. All these follow easily from the decay properties embodied in the main result of \cite{KS:main}, (see Section 3.4.3 in \cite{KS:main}) and the basic equations of motion.        
Using the same ingredients, a memory-type result for the angular momentum was later derived by An-He-Shen \cite{AHS}. The regularity of the horizon was recently proved by Chen-Klainerman \cite{K:Chen}. It is important to note that all these conclusions remain true for the full sub-extremal range $|a|<m$ provided that similar decay estimates hold.\footnote{These, of course can only be established by extending the stability result to the entire sub-extremal case.}
\end{remark}
\subsection{General covariant modulated (GCM) procedure}\label{reviewsection}
The GCM spheres introduced in \cite{KS:Kerr1,KS:Kerr2} play an essential role throughout this paper. For the sake of this introduction, we review the main definitions needed to understand the statements of the results of Section \ref{secfirstmain}. A more comprehensive discussion of these is given in Section \ref{secpre}.
\subsubsection{Null horizontal structures}\label{subsection:horizstructures-intro}
We use in what follows the language of null horizontal structures $e_3, e_4, \HH$ developed in \cite{GKS} and \cite{KS:main} and summarized here in section \ref{sect:horizontal-structures}. In a first approximation, we need the definitions of the Ricci coefficients $\chi,\chib,\ze$, with the first two decomposed into $\hch,\trch,\atrch$ and, respectively, $\chibh,\trchb,\atrchb$, and the curvature components $\b,\rho,\rhod,\bb$. We also need the mass aspect functions
\begin{equation*}
\mu:=-\div\ze-\rho+\f12\hch\cdot\hchb,\qquad\quad\mub:=\div\ze-\rho+\f12\hch\cdot\hchb.
\end{equation*}
Recall that in the canonical outgoing geodesic foliation of Schwarzschild, we have
\begin{align}\label{Introd:GCMspheres1}
\trch=\frac{2}{r},\qquad\quad\trchb=-\frac{2\Up}{r},\quad\qquad \mu=\frac{2m}{r^3},
\end{align}
where we denoted
\begin{align}\label{dfUp}
    \Up:=1-\frac{2m}{r}.
\end{align}
Note that $r$ and $m$ coincide with the area radius and Hawking mass of surfaces  $S$ of fixed $r$, i.e.
\begin{equation*}
r:=\sqrt{\frac{|S|}{4\pi}},\qquad\qquad\frac{2m}{r}:=1+\frac{1}{16\pi}\int_S\trch\trchb.
\end{equation*}
We now introduce the following definition of \emph{outgoing principle geodesic (PG) structure}, which plays an essential role in this paper.
\begin{df}\label{dfoutgoingPGfirst}
An outgoing PG structure consists of a null pair $(e_3, e_4)$ and the induced horizontal structure $\HH$, together with a scalar function r such that:
\begin{enumerate}
    \item $e_4$ is a null outgoing geodesic vectorfield, i.e. $\D_4e_4=0$,
    \item $r$ is an affine parameter, i.e. $e_4(r)=1$,
    \item the gradient of $r$, defined by $\grad r:=\g^{\a\b}\pr_\b r\pr_\a$, is perpendicular to $\HH$.
\end{enumerate}
\end{df}
\begin{rk}
    By abuse of language, we call an outgoing PG structure $\{S(u,r),(e_3,e_4,\HH)\}$ simply an outgoing PG $S(u,r)$--foliation. The asymptotic properties of $\KSAF$ spacetimes, mentioned in Remark \ref{remark:final state1}, are defined relative to a given outgoing PG structure.
\end{rk}
\subsubsection{Deformations of spheres and frame transformations}
We consider a spacetime region $(\M,\g)$ endowed with an outgoing PG $S(u,r)$--foliation. The construction of GCM spheres in \cite{KS:Kerr1,KS:Kerr2} was obtained by deforming a given background sphere $S(u,r)$ by a map $\Psi:S(u,r)\to S'\subset\M$ of the form
\begin{align}
\Psi(u,r,\th^1,\th^2)=\left(u+U(\th^1,\th^2),r+R(\th^1,\th^2),\th^1,\th^2\right)
\end{align}
with $(U,R)$ smooth functions on $S$, vanishing at a fixed point of $S$, and $(\th^1,\th^2)$ spherical coordinates on $S$. Given such a deformation, we identify, at any point on $S'$, two important null frames:
\begin{enumerate}
\item The null frame $(e_3, e_4, e_1, e_2)$ of the background outgoing PG foliation;
\item A new  null frame $(e'_3, e'_4, e'_1, e'_2)$  adapted to $S'$ (i.e. $e_1'$, $e_2'$ tangent to $S'$),  obtained  according to the transformation formulae  \eqref{eq:Generalframetransf-intro}.
\end{enumerate}
In general, two null frames $(e_3, e_4, e_1, e_2)$ and $(e_3', e_4', e_1', e_2')$ are related by a frame transformation as follows:\footnote{See also Lemma \ref{transformation}.}
\begin{align}\label{eq:Generalframetransf-intro}
\begin{split}
e_4'&=\la\left(e_4+f^be_b+\frac 1 4 |f|^2e_3\right),\\
e_a'&= \left(\de_{ab} +\frac{1}{2}\fb_af_b\right)e_b+\frac 1 2\fb_ae_4+\left(\frac 1 2 f_a +\frac{1}{8}|f|^2\fb_a\right)e_3,\\
e_3'&=\la^{-1}\left(\left(1+\frac{1}{2}f\c\fb+\frac{1}{16}|f|^2|\fb|^2\right)e_3 +\left(\fb^b+\frac 1 4 |\fb|^2f^b\right)e_b+\frac 1 4 |\fb|^2 e_4 \right),
\end{split}
\end{align}
where the scalar function $\la$ and the 1-forms $f$ and $\fb$ are called the \emph{transition functions} from $(e_3,e_4,e_1,e_2)$ to $(e_3',e_4',e_1',e_2')$.
\subsubsection{Basis of \texorpdfstring{$\ell=1$}{} modes}
We introduce the following generalization of the $\ell=1$ spherical harmonics of the standard sphere,\footnote{Recall that on the standard sphere $\SSS^2$, in spherical coordinates $(\th,\vphi)$, these are $J^{(0)}=\cos\th$, $J^{(+)}=\sin\th\cos\vphi$ and  $J^{(-)}=\sin\th\sin\vphi$.} which are used to define the GCM conditions.
\begin{df}\label{jpdefintro}
On an $\ep$--almost round sphere $(S, g^S)$, in the sense of
\begin{equation}\label{almostroundintro}
    \left\|K^S - \frac{1}{(r^S)^2} \right\|_{L^\infty(S)} \leq \frac{\ep}{(r^S)^2},
\end{equation}
where $K^S$ denotes the Gauss curvature of $S$ and $r^S$ denotes the area radius of $S$.
\begin{enumerate}
    \item We define an $\ep$--approximate basis of $\ell=1$ modes on $S$ to be a triplet of functions $\Jp$ on $S$ verifying
\begin{align}
\begin{split}\label{def:Jpsphericalharmonicsintro}
(r^2\De+2)\Jp&= O(\ep),\qquad p=0,+,-,\\
\frac{1}{|S|} \int_{S}\Jp J^{(q)}&=\frac{1}{3}\de_{pq}+O(\ep),\qquad p,q=0,+,-,\\
\frac{1}{|S|}\int_{S}\Jp&=O(\ep),\qquad p=0,+,-,
\end{split}
\end{align}
where $\ep>0$ is a sufficiently small constant.
\item Let $(\Phi, u)$ be the unique, up to isometries of $\mathbb{S}^2$, uniformization\footnote{This was called effective uniformization in  \cite{KS:Kerr2}.}, as in Corollary 3.8 of \cite{KS:Kerr2}, that is, a unique diffeomorphism $\Phi:\SSS^2\to S$ and a unique centered conformal factor $u$ s.t.
\[
\Phi^\#(g^S)=(r^S)^2 e^{2u}\ga_0.
\]
We define the \emph{canonical} choice of $\ep$--approximate basis of $\ell=1$ modes\footnote{This was simply called \emph{canonical $\ell=1$ modes} in \cite{KS:Kerr2}.} on $S$ by
\[
J^S := J^{\SSS^2} \circ \Phi^{-1},
\]
where $J^{\SSS^2}$ denotes the standard $\ell=1$ spherical harmonic on round sphere. 
\end{enumerate}
\end{df}
\begin{rk}
On an $\ep$--almost round sphere $(S,g^S)$, the \emph{canonical} choice of $\ep$--approximate basis $J^{(S,p)}$ of $\ell = 1$ modes on $S$ satisfy an addition property
\begin{align*}
\frac{1}{|S|}\int_{S}J^{(S,p)}&=0,\qquad p=0,+,-.
\end{align*}
This is a simple consequence of the centered (balancing) condition of the conformal factor $u$ in the definition.
\end{rk}
Assuming the existence of such a basis $\Jp$, $p\in\{-,0,+\}$, we define, for any scalar function $h$ on $S$,
\begin{align}\label{defl=1intro}
(h)^S_{\ell=1}:=\left\{\frac{1}{|S|}\int_{S}h\Jp,\quad p=-,0,+\right\}.
\end{align}
A scalar function $h$ on $S$ is said to be supported on $\ell\leq 1$ modes, i.e.  $(h)^S_{\ell\ge 2}=0$, if there exist constants $A_0, B_{-}, B_0, B_{+} $ such that
\begin{align}
h=A_0+B_{-}J^{(-)}+B_0J^{(0)}+B_{+}J^{(+)}.
\end{align}
\subsubsection{GCM spheres and incoming geodesic foliation}
The idea behind GCM spheres in perturbations of Kerr, is that of finding 2-surfaces verifying conditions as close as possible to \eqref{Introd:GCMspheres1}. In fact, the GCM spheres are topological spheres $\S$ endowed with a null frame $(e_3^\S,e_4^\S,e_1^\S,e_2^\S)$ adapted to $\S$ (i.e. $e_1^\S,e_2^\S$ tangent to $\S$), relative to which the null expansions $\trch^\S$, $\trchb^\S$ and mass aspect function $\mu^\S$ satisfy:
\begin{align}\label{Introd:GCMspheres2}
\trch^\S-\frac{2}{r^\S}=0,\qquad\left(\trchb^\S+\frac{2\Up^\S}{r^\S}\right)_{\ell\ge 2}=0,\qquad\left(\mu^\S-\frac{2m^\S}{(r^\S)^3} \right)_{\ell\ge 2}=0,
\end{align}
where $r^\S$ and $m^\S$ denote the area radius and Hawking mass of $\S$ and $\Up^\S$ is defined as in \eqref{dfUp}. A GCM sphere is called an intrinsic GCM sphere if in additional to \eqref{Introd:GCMspheres2}, it satisfies:
\begin{align}\label{Introd:GCMspheres3}
    \trchb^\S=-\frac{2\Up^\S}{r^\S},\qquad\qquad(\div^\S\b^\S)_{\ell=1}=0,
\end{align}
w.r.t. the \emph{canonical} choice of $\ep$--approximate basis of $\ell=1$ modes of $\S$. Note that all the quantities we have introduced above, based on the definitions made in Section \ref{subsection:horizstructures-intro}, are well defined on a given sphere. We now recall the definition of an incoming geodesic foliation. To this end, we need to define the transverse Ricci coefficients, see Section \ref{sect:horizontal-structures},
\begin{align*}
\eta_a=\frac 1 2 \g(\D_3e_4,e_a),\qquad\xib_a=\frac 1 2 \g(\D_3 e_3 , e_a),\qquad \omb=\frac{1}{4}\g(\D_3e_3,e_4).
\end{align*}
Recall that the incoming geodesic conditions (see Sections 2.3--2.4 in \cite{KS:main}) are given by:
\begin{align}\label{Geodesiccondtions}
\eta=0,\qquad\quad\xib=0,\qquad\quad \omb=0.
\end{align}
\begin{rk}
Emanating from an intrinsic GCM sphere $\S$, we can construct an incoming null cone $\Cb$, endowed with an incoming geodesic foliation. We have by construction that \eqref{Introd:GCMspheres2} and \eqref{Introd:GCMspheres3} hold on the last sphere $\S$ while \eqref{Geodesiccondtions} holds on the null cone $\Cb$.
\end{rk}
\subsection{First version of main results}\label{secfirstmain}
We are now ready to state simple versions of our main results and discuss the main ideas in their proofs.
\subsubsection{Regularity of conformal metric up to \texorpdfstring{$\II^+$}{}}
Let $(\M,\g)$ be a $\KSAF$ spacetime  as defined in Remark \ref{remark:final state1}. We denote $(\Mt,\gt)$ its Penrose compactification:
\begin{align}\label{compactfirstdefined}
\Mt:=\M\cup\II^+,\qquad\quad\gt:=\vr^2\g.
\end{align}
with $\vr:=r^{-1}$ the \emph{boundary defining function} and the future null infinity $\II^+$ is defined by:
\begin{align*}
    \II^+:=\{\vr=0\}.
\end{align*}
\begin{thm}\label{firstregularity}
Let $(\M,\g)$ be a $\KSAF$ spacetime and let $(\Mt,\gt)$ be the compactified spacetime defined in \eqref{compactfirstdefined}. Then we have
$$
\gt\in C^{\ks}(\M)\cap C^{1,\frac{1}{2}+\dec}(\Mt),
$$
where $\ks$ is an integer that denotes the regularity of the decay estimates of \cite{KS:main} and $0<\dec\ll 1$ is a fixed small constant\footnote{See Section 3.4.1 in \cite{KS:main} for more explanations on the choice of $\ks$ and $\dec$.}.
\end{thm}
Theorem \ref{firstregularity} is restated as Theorem \ref{regularityKerr} in Section \ref{secscr}. The proof is based on the asymptotic behaviors of the metric components of $\KSAF$ spacetimes.
\subsubsection{Limiting GCM (LGCM) foliation on \texorpdfstring{$\II^+$}{}}
The following theorem shows the existence of a canonical foliation on the future null infinity $\II^+$, which is called a Limiting GCM (LGCM) foliation.
\begin{thm}\label{firstGCMlimit}
Let $(\M,\g)$ be a $\KSAF(a_f,m_f)$ spacetime, endowed with a background outgoing PG $S(u,r)$--foliation. Then, there exists a sphere foliation $S'(u',r')$ and a null frame $(e'_3, e'_4,e_1',e_2')$ near $\II^+$, relative to which the Ricci coefficients and curvature components have the following asymptotic behavior:
\begin{align}
\begin{split}\label{eq:Def-LGCM-gen}
\lim_{u\to\infty}\lim_{C_u,r\to\infty}r^2\left(\trch'-\frac{2}{r'}\right)&=0,\qquad\quad\;\,\lim_{u\to\infty}\lim_{C_u,r\to\infty}r^2\left(\trchb'+\frac{2}{r'}\right)=4m_f,\\
\lim_{u\to\infty}\lim_{C_u,r\to\infty}(r^3\mu')_{\ell\geq 2}&=0,\qquad\qquad\quad\, \lim_{u\to\infty}\lim_{C_u,r\to\infty}(\div'\b')_{\ell=1}=0,\\ \lim_{C_u,r\to\infty}\left(r^2\,\atrch',r^2\,\atrchb'\right)&=0,\qquad\qquad\lim_{u\to\infty}\lim_{C_u,r\to\infty}r^5(\div'\b')_{\ell=1}=2a_fm_f,
\end{split}
\end{align}
where the modes are taken w.r.t. to an approximate basis of $\ell=1$ modes $\Jp$ introduced in Definition \ref{jpdefintro}. We also have the following asymptotic behavior:
\begin{equation}\label{geodesiclim}
    \lim_{C_u,r\to\infty}(r\eta,r\xib,r\omb)=0.
\end{equation}
Moreover, a sphere foliation near $\II^+$ that satisfies \eqref{eq:Def-LGCM-gen} and \eqref{geodesiclim} is called a \emph{Limiting GCM (LGCM)} foliation.
\end{thm}
Theorem \ref{firstGCMlimit} is restated as Theorem \ref{LGCMconstruction} and proved in Section \ref{ssecendproof}, which constructs the desired foliation from the background foliation of $(\M,\g)$. The construction is based on a limiting process, which shows that a sequence of incoming null cones, emanating from intrinsic GCM spheres, endowed with incoming geodesic foliation, converges to the desired foliation on $\II^+$.
\subsubsection{Absence of supertranslation and spatial translation}
The next theorem illustrates the uniqueness of the LGCM foliation defined in Theorem \ref{firstGCMlimit}.
\begin{thm}\label{firstunique}
Let $(\M,\g)$ be a $\KSAF$ spacetime. Any two LGCM foliations of $(\M,\g)$ near $\II^+$, differ by a translation along $\II^+$. More precisely, denoting $(f,\fb,\la)$ the transition functions from the frame of the first LGCM to that of the second, as defined by \eqref{eq:Generalframetransf-intro}, we must have
\begin{align}\label{ffbovla=0first}
\lim_{C_u,r\to\infty}(rf,r\fb,r\ovla)=0.
\end{align}
\end{thm}
Theorem \ref{firstunique} is restated as Theorem \ref{translation} and proved in Section \ref{secsuper},  based on the fact that, since the GCM conditions \eqref{eq:Def-LGCM-gen} hold for both LGCM foliations near $i^+$ we can derive an elliptic system for $(f,\fb,\la)$ near $i^+$, which implies that \eqref{ffbovla=0first} holds near $i^+$. Since geodesic conditions \eqref{geodesiclim} hold for both LGCM foliations, we obtain a transport system for $(f,\fb,\la)$ along $\II^+$, from which the uniqueness of $(f,\fb,\la)$ on $\II^+$ follows easily.
\begin{rk}
The condition \eqref{ffbovla=0first} implies that the physical quantities defined in two different LGCM foliations coincide with each other. This eliminates the supertranslation ambiguity and spatial translation, see Remark \ref{nosuperspatialrk} for more explanations.
\end{rk}
\subsubsection{Physical laws on \texorpdfstring{$\II^+$}{}}
The uniqueness of the LGCM foliation, proved in Theorem \ref{firstunique}, allows us to make unambiguous definitions of the main physical quantities (energy, linear and angular momentum and center of mass)  at null infinity.
\begin{df}\label{firstphysicalquantities}
Relative to the associated LGCM foliation of a $\KSAF$  spacetime, constructed by Theorem \ref{firstGCMlimit}, we introduce the following limiting quantities \footnote{The existence of the all the limits here is proved in Theorem \ref{LGCMconstruction}.}
\begin{itemize}
\item Define the \emph{shear} $\sc\The$ and \emph{news} $\sc\Thb$ at $\II^+$:
\begin{align}\label{shearnews}
\sc\The:=\lim_{C_u,r\to\infty}r^2\,\sc\hch,\qquad\quad\sc\Thb:=\lim_{C_u,r\to\infty}r\,\sc\hchb.
\end{align}
\item Define the \emph{energy} $\sc\EE$, \emph{linear momentum} $\sc\PP$, \emph{center of mass} $\sc\CC$ and \emph{angular momentum} $\sc\JJ$ as follows:
\begin{align}
\begin{split}\label{firstdfphy}
    \sc\EE&:=\lim_{C_u,r\to\infty}\frac{r^3}{2}\left(\sc\mu+\sc\mub\right)_{\ell=0},\\
    \sc\PP&:=\lim_{C_u,r\to\infty}\frac{r^3}{2}\left(\sc\mu+\sc\mub\right)_{\ell=1},\\
    \sc\CC&:=\lim_{C_u,r\to\infty}r^5\left(\sc\div\sc\b\right)_{\ell=1},\\
    \sc\JJ&:=\lim_{C_u,r\to\infty}r^5\left(\sc\curl\sc\b\right)_{\ell=1}.
\end{split}
\end{align}
\item For any $1$--form $f$ and tensorfield $U$ satisfying
$$
\lim_{C_u,r\to\infty}f=\Fscr, \qquad\quad\lim_{C_u,r\to\infty}U=\Uscr,
$$
we define the limiting differential operators:
\begin{align}\label{dfnabo}
    (\divo,\curlo)\Fscr:=\lim_{C_u,r\to\infty}r(\div,\curl)f,\qquad (\nabo_3,\nabo)\Uscr:=\lim_{C_u,r\to\infty}(\nab_3,r\nab)U.
\end{align}
\end{itemize}
\end{df}
\begin{remark}
The shear $\The$ and news $\Thb$ defined in \eqref{shearnews} describe the nonlinear terms in the evolution of physical quantities defined in \eqref{firstdfphy}, see Theorem \ref{firstphysical} below. Note also that relative to a coordinates chart $(u,\th,\vphi)$ on $\II^+$, the limiting differential operators $\nabo_3$ and $\nabo_a$ with $a=1,2$, defined in \eqref{dfnabo},  are linear combinations of $\pr_u$, $\pr_\th$ and $\pr_\vphi$.
\end{remark}
We are now ready to describe the evolution of physical quantities along $\II^+$ in the LGCM foliation.
\begin{thm}\label{firstphysical}
Let $(\M,\g)$ be a $\KSAF$ spacetime. Then in the LGCM foliation, the null energy, linear momentum, center of mass and angular momentum introduced in Definition \ref{firstphysicalquantities} satisfy the following evolution equations:\footnote{Here, the $\ell=0$ modes of a scalar means its average on the sphere of the LGCM foliation.}
\begin{align}
    \begin{split}\label{firstevolution}
        \pr_u\left(\sc\EE\right)&=-\frac{1}{4}\left(\big|\sc\Thb\big|^2\right)_{\ell=0},\\
        \pr_u\left(\sc\PP\right)&=-\frac{1}{4}\left(\big|\sc\Thb\big|^2\right)_{\ell=1},\\
        \pr_u\left(\sc\CC\right)&=\sc\PP+\left(\frac{1}{2}\sc\The\c\sc\Thb+\sc\divo\left(\sc\divo\sc\Thb\c\sc\The\right)\right)_{\ell=1},\\
        \pr_u\left(\sc\JJ\right)&=\left(\frac{1}{2}\sc\The\wedge\sc\Thb+\sc\curlo\left(\sc\divo\sc\Thb\c\sc\The\right)\right)_{\ell=1},
    \end{split}
\end{align}
where all the modes are taken w.r.t. the $\ell=1$ basis $\Jp$ of the LGCM foliation.
\end{thm}
Theorem \ref{firstphysical} is restated as Theorem \ref{thm-BHdynamic} and proved in Section \ref{ssecphy}. The evolution laws in \eqref{firstevolution} follow by taking the limits of the null structure equations in the direction $e_3$.
\begin{rk}
The first equation in \eqref{firstevolution} implies the Bondi mass loss formula. It was first derived and used by Christodoulou \cite{Ch-memory} in connection with his famous memory effect. The last equation in \eqref{firstevolution} was derived and used in \cite{AHS} to derive an angular momentum memory type effect. See Remark \ref{rkBHdynamic} for more explanations about \eqref{firstevolution}.
\end{rk}
\subsubsection{Gravitational wave recoil}
\begin{thm}\label{firstkick}
Let $(\M,\g)$ be the final spacetime derived in the proof of Theorem \ref{MainThm-firstversion} and let $\LL_0$ be the initial layer region, endowed with a $S^{(0)}$--foliation, given as the initial assumption in \cite{KS:main}.\footnote{In \cite{KS:main}, the curvature components on the initial layer $\LL_0$ are assumed to be $O(\ep_0 r^{-\frac{7}{2}-\dec})$     close to the reference Kerr data, see Section 3.4.2 in \cite{KS:main} for more details. The decay rate $r^{-\frac{7}{2}-\dec}$ is consistent with the main result of the external stability of Kerr in \cite{ShenKerr}.} Let $S(1,r)$ be a leaf of the LGCM foliation constructed in Theorem \ref{firstGCMlimit} and let $S^{(0)}_1$ be a leaf of the initial layer foliation that has a point common to $S(1,r)$. We define the center of mass on $S(1,r)$ and $S^{(0)}_1$ as in \eqref{firstdfphy}:
\begin{align*}
    \CC_1:=r^5(\div\b)_{\ell=1,S(1,r)},\qquad \CC^{(0)}_1:=(r^{(0)})^5\left(\div^{(0)}\b^{(0)}\right)_{\ell=1,S^{(0)}}.
\end{align*}
Then, we have
\begin{align*}
    \CC_1-\CC^{(0)}_1=O\left(\ep_0 r^{\frac{1}{2}-\dec}\right).
\end{align*}
\end{thm}
Theorem \ref{firstkick} is restated as Theorem \ref{kickthm} and proved in Section \ref{seckick}  by combining the evolution equations of the center of mass in \eqref{firstevolution} with the initial perturbation assumption of \cite{KS:main}.
\begin{rk}
Theorem \ref{firstkick} encodes the fact that there is a large displacement of the center of mass of the initial layer foliation relative to that of the LGCM foliation. This leads to non trivial technical difficulties encountered in the proof of Theorem M0 in \cite{KS:main}. The large deformation of the center of mass is sometimes referred to as \emph{gravitational wave recoil} or \emph{black hole kick} in the physical literature, see \cite{Fitchett,HHS,VV}. See also Remark \ref{anormalfb} for more explanations.
\end{rk}
\begin{rk}
    The initial assumption in \cite{KS:main} implies that\footnote{Here and the line below, we ignore the difference of $r$ and $r^{(0)}$ since one can show that $r\sim r^{(0)}$.}
    \begin{align*}
        \div^{(0)}\b^{(0)}=O\left(\frac{\ep_0}{r^{\frac{9}{2}+\dec}}\right).
    \end{align*}
    Hence, we have
    \begin{align*}
        \CC^{(0)}_1=(r^{(0)})^5\left(\div^{(0)}\b^{(0)}\right)_{\ell=1,S^{(0)}}=O(\ep_0 r^{\frac{1}{2}-\dec}),
    \end{align*}
    which implies the non triviality of \emph{kick}. However, under stronger assumptions of the decay rates of the initial data for a black hole stability problem, one can show that the black hole kick has size $O(\ep_0)$, see Theorem \ref{kickthmgeneral} for more details.
\end{rk}
\subsection{Structure of the paper}
\begin{itemize}
\item In Section \ref{secpre}, we introduce the basic geometric setup and recall the fundamental notions and the basic equations. We also establish the null frame transformation formulae.
\item In Section \ref{seclimit}, we study the asymptotic behavior of the null Ricci and curvature components. More precisely, we prove the existence of their appropriate weighted limits and deduce their Taylor expansions in terms of $r^{-1}$ near future null infinity $\II^+$.
\item In Section \ref{secscr}, we introduce the Penrose conformal compactification of the $\KSAF$ spacetime $(\M,\g)$, with $\II^+$ its future null boundary. We show that the regularity of the conformal metric up to $\II^+$ is determined by its decay properties in $r$.
\item In Section \ref{secLG}, we recall the main definitions and results of the GCM spheres introduced in \cite{KS:Kerr1,KS:Kerr2} and construct a sequence of incoming null cones embedded in $(\M,\g)$, which converges to $\II^+$. We prove that in the limit they induce a canonical foliation on $\II^+$, called a Limiting GCM (LGCM) foliation.
\item In Section \ref{secsuper}, we show that the LGCM foliation defined in Section \ref{secLG} is unique on $\II^+$ up to a trivial transformation. Thus, LGCM foliations exclude Pernrose's supertranslation ambiguity and also the freedom of spatial translation on $\II^+$.
\item In Section \ref{ssecphy}, we use LGCM foliation to define the null energy, linear momentum, center of mass and angular momentum on $\II^+$. In view of the uniqueness of LGCM foliation, these quantities are unambiguously defined. We derive the basic evolution laws for these physical quantities along null infinity and draw comparisons with the ones for the classical mechanics of isolated systems.
\item In Section \ref{seckick}, we compare the center of mass for the initial layer foliation, introduced as the initial assumptions\footnote{These assumptions where rigorously analyzed in \cite{ShenKerr}.} of \cite{KS:main} on $\LL_0$, and the LGCM foliation of $\II^+$. We show that there is a large deformation of the two centers of mass near $\II^+\cap\LL_0$. This phenomenon corresponds to a gravitational wave recoil, or \emph{black hole kick} in the physics literature, see for example \cite{Fitchett,HHS,VV}.
\item In Appendix \ref{section:appendix}, we recall the main equations and commutation identities in \cite{GKS,KS:main} that are used in this paper. We also provide the proof of Theorem \ref{expansionexist}, which describes the asymptotic behavior of all geometric quantities near $\II^+$.
\item In Appendix \ref{summary}, we summarize all the limiting quantities we defined in this paper for the convenience of the reader.
\end{itemize}
\subsection{Acknowledgements}
The authors are very grateful to Xinliang An, Elena Giorgi, Taoran He, J\'er\'emie Szeftel, Mu-Tao Wang and Pin Yu for their interest in this work and for many helpful discussions and comments.
\section{Preliminaries}\label{secpre}
We summarize the general formalism introduced in \cite{GKS} and the main assumptions in \cite{KS:main}, which will be used in this paper, see also Appendix \ref{section:appendix} for more details.
\subsection{Null horizontal structures}
\label{sect:horizontal-structures}
Let $(\MM,\g)$ be a Lorentzian spacetime. Consider a fixed null pair of vectorfields  $(e_3,e_4)$, i.e. 
\begin{align*}
\g(e_3, e_3)=\g(e_4, e_4)=0, \qquad \g(e_3,e_4)=-2, 
\end{align*}
and denote by $\HH$ the vector space of horizontal vectorfields $X$ on $\MM$, i.e.  $\g(e_3, X)= \g(e_4, X)=0$. A null frame $(e_3, e_4, e_1, e_2)$ on $\MM$ consists, in addition to the null pair $(e_3, e_4)$, of a choice of horizontal vectorfields $(e_1, e_2)$, such that 
$$
\g(e_a, e_b)=\de_{ab}\qquad a, b=1,2.
$$
However, the commutator $[X,Y]$ of two horizontal vectorfields may not be horizontal. We say that the pair $(e_3, e_4)$ is integrable if $\HH$ forms an integrable distribution, i.e. $X, Y\in\HH $ implies $[X,Y]\in\HH$. As is well known, the principal null pair in Kerr fails to be integrable, see also Remark \ref{rem:nonintegrabilityandatrchatrchb}. Given an arbitrary vectorfield $X$ we denote by $^{(h)}X$ its horizontal projection, 
$$
^{(h)}X=X+\frac 1 2 \g(X,e_3)e_4+\frac 1 2\g(X,e_4) e_3.
$$ 
For any $X,Y\in \O(\MM)$ we define $\ga(X, Y)=\g(X, Y)$ and\footnote{In the particular case where the horizontal structure is integrable, $\ga$ is the induced metric, and $\chi$ and $\chib$ are the null second fundamental forms.} 
$$
\chi(X,Y)=\g(\D_Xe_4 ,Y), \qquad \chib(X,Y)=\g(\D_Xe_3,Y).
$$
Observe that $\chi$ and $\chib$ are symmetric if and only if the horizontal structure is integrable. Indeed this follows easily from the formulae\footnote{Note that we can view  $\chi$ and $\chib$ as horizontal 2-covariant tensor-fields by extending their definition to arbitrary $X,Y$, $\chi(X, Y)= \chi( ^{(h)}X, ^{(h)}Y)$, $\chib(X, Y)= \chib(^{(h)}X,^{(h)}Y)$.},
\begin{align*}
\chi(X,Y)-\chi(Y,X)&=\g(\D_X e_4, Y)-\g(\D_Ye_4,X)=-\g(e_4, [X,Y]),\\
\chib(X,Y)-\chib(Y,X)&=\g(\D_X e_3, Y)-\g(\D_Ye_3,X)=-\g(e_3, [X,Y]).
\end{align*}
We define their trace $\trch$, $\trchb$, and anti-trace $\atrch$, $\atrchb$ as follows
$$
\trch:=\de^{ab}\chi_{ab},\qquad\trchb:=\de^{ab}\chib_{ab},\qquad \atrch:=\in^{ab}\chi_{ab},\qquad\atrchb:=\in^{ab}\chib_{ab}.
$$
Accordingly, we decompose $\chi,\chib$ as follows
\begin{align*}
\chi_{ab}&=\chih_{ab} +\frac 1 2\de_{ab} \trch+\frac 1 2 \in_{ab}\atrch,\\
\chib_{ab}&=\chibh_{ab} +\frac 1 2\de_{ab} \trchb+\frac 1 2 \in_{ab}\atrchb.
\end{align*}
where $\chih$ and $\chibh$ are the symmetric traceless parts of $\chi$ and $\chib$ respectively.
\begin{remark}\label{rem:nonintegrabilityandatrchatrchb}
The non integrability of $(e_3,e_4)$ corresponds to non vanishing $\atrch$ and $\atrchb$. This is the case of the principal null frame of Kerr for which $\atrch$ and $\atrchb$ are indeed non trivial. The quantities $\trch$ and $\trchb$ are also called null expansions.
\end{remark}
We define the horizontal covariant operator $\nab$ as follows:
\begin{align*}
\nab_X Y:=^{(h)}(\D_XY)=\D_XY- \frac 1 2 \chib(X,Y)e_4 -\frac 1 2 \chi(X,Y) e_3,\quad X, Y\in \HH.
\end{align*}
Note that,
$$
\nab_X Y-\nab_Y X=[X, Y]-\frac 1 2\left(\atrchb\,e_4+\atrch\, e_3\right)\in(X, Y).
$$
Note that $\nab$ acts like a Levi-Civita connection i.e., for all $X,Y,Z\in\HH$,
$$
Z\ga (X,Y)=\ga(\nab_ZX, Y)+\ga(X,\nab_ZY).
$$
We can then define connection  and curvature coefficients   similar to   those  in the integrable case, as in \cite{Ch-Kl}.
\begin{df}\label{df:allquantities} Given a null pair $(e_3,e_4)$  and $e_a$ a basis of the horizontal space $\HH$ we define.
\begin{enumerate}
\item \textit{Ricci coefficients}
\begin{align*}
\chib_{ab}&=\g(\D_ae_3, e_b),\quad \,\,\chi_{ab}=\g(\D_ae_4, e_b),\quad\,\,\,\,
\xib_a=\frac 1 2 \g(\D_3 e_3 , e_a),\quad \,\,\,\xi_a=\frac 1 2 \g(\D_4 e_4, e_a),\\
\omb&=\frac 1 4 \g(\D_3e_3 , e_4),\quad\,\, \om=\frac 1 4 \g(\D_4 e_4, e_3),\quad 
\etab_a=\frac 1 2 \g(\D_4 e_3, e_a),\quad \quad \eta_a=\frac 1 2 \g(\D_3 e_4, e_a),\\
\ze_a&=\frac 1 2 \g(\D_{a}e_4,  e_3).
\end{align*}
\item \textit{Curvature coefficients} 
\begin{align*}
\a_{ab}&=\R_{a4b4},\qquad\quad\,\b_a=\frac 12 \R_{a434},\qquad\quad\bb_a=\frac 1 2 \R_{a334},  \qquad\quad\aa_{ab}=\R_{a3b3},\\
\rho&=\frac 1 4 \R_{3434},\qquad\;\,\rhod=\frac 1 4\dual \R_{3434},
\end{align*}
\item\textit{Mass aspect function and conjugate mass aspect function}
\begin{align*}
\mu:=- \div \ze -\rho+\frac 1 2 \chih \c\chibh ,\qquad \mub:=\div \ze -\rho+\frac 1 2\hch\c\hchb
\end{align*}
\item
We define the following complex notations according to \cite{GKS1}
\begin{align*}
A:=\a+i\dual\a, \quad B:=\b+i\dual\b, \quad P:=\rho+i\si,\quad \Bb:=\bb+i\dual\bb, \quad \Ab:=\aa+i\dual\aa,
\end{align*}    
 and   
\begin{align*}
X&=\chi+i\dual\chi,\quad \Xb=\chib+i\dual\chib, \quad H=\eta+i\dual\eta,\quad \Hb=\etab+i\dual \etab, \quad Z=\ze+i\dual\ze, \\ 
\Xi&=\xi+i\dual\xi,\quad\;\; \Xib=\xib+i\dual\xib.
\end{align*}   
In particular, note that
\begin{align*}
\tr X=\trch-i\atrch,\quad \Xh=\hch+i\dual\hch,\quad\tr\Xb =\trchb-i\atrchb,\quad \Xbh=\hchb+i\dual\hchb.
\end{align*}   
\end{enumerate}
\end{df}
\subsection{Hodge systems}\label{ssec7.2}
\begin{df}\label{tensorfields}
For tensor fields defined on a $2$--sphere $S$, we denote by $\sk_0:=\sk_0(S)$ the set of pairs of scalar functions, $\sk_1:=\sk_1(S)$ the set of $1$--forms and $\sk_2:=\sk_2(S)$ the set of symmetric traceless $2$--tensors.
\end{df}
\begin{df}\label{def7.2}
Given $\xi\in\sk_1$, we define its Hodge dual
\begin{equation*}
    {^*\xi}_a:=\in_{ab}\xi^b.
\end{equation*}
Clearly $^*\xi \in \sk_1$ and
\begin{equation*}
    ^*(^*\xi)=-\xi.
\end{equation*}
Given $U \in \sk_2$, we define its Hodge dual
\begin{equation*}
{^*U}_{ab}:=\in_{ac}{U^a}_b.
\end{equation*}
Observe that $^*U\in\sk_2$ and
\begin{equation*}
     ^*(^*U)=-U.
\end{equation*}
\end{df}
\begin{df}
    Given $\xi,\eta\in\sk_1$, we denote
\begin{align*}
    \xi\cdot\eta&:=\de^{ab}\xi_a\eta_b,\\
    \xi\wedge\eta&:=\in^{ab}\xi_a\eta_b,\\
    (\xi\hot\eta)_{ab}&:=\xi_a\eta_b+\xi_b\eta_a-\de_{ab}\xi\c\eta.
\end{align*}
Given $\xi\in \sk_1$, $U\in\sk_2$, we denote
\begin{align*}
    (\xi\c U)_a:=\de^{bc}\xi_b U_{ac}.
\end{align*}
Given $U,V\in \sk_2$, we denote 
\begin{align*}
    U\wedge V& :=\in^{ab}U_{ac}V_{cb},\\
    U\c V& :=\de^{ab}U_{ac}V_{cb}.
\end{align*}
\end{df}
\begin{df}
    For a given $\xi\in\sk_1$, we define the following differential operators:
    \begin{align*}
        \sdiv\xi&:= \de^{ab}\nab_a\xi_b,\\
        \curl\xi&:= \in^{ab}\nab_a\xi_b,\\
        (\nab\hot\xi)_{ab}&:=\nab_a\xi_b+\nab_b\xi_a-\de_{ab}(\sdiv\xi).
    \end{align*}
\end{df}
\begin{df}
    We define the following Hodge type operators:
    \begin{itemize}
        \item $d_1$ takes $\sk_1$ into $\sk_0$ and is given by:
        \begin{equation*}
            d_1\xi:=(\sdiv\xi,\curl \xi),
        \end{equation*}
        \item $d_2$ takes $\sk_2$ into $\sk_1$ and is given by:
        \begin{equation*}
            (d_2 U)_a:=\nab^{b} U_{ab}, 
        \end{equation*}
        \item $d_1^*$ takes $\sk_0$ into $\sk_1$ and is given by:
        \begin{align*}
            d_1^*(f,f_*)_{ a}:=-\nab_af+\in_{ab}\nab_b f_*,
        \end{align*}
        \item $d_2^*$ takes $\sk_1$ into $\sk_2$ and is given by:
        \begin{align*}
            d_2^*\xi:=-\frac{1}{2}\nab\hot\xi.
        \end{align*}
    \end{itemize}
\end{df}
\begin{prop}\label{ddddprop}
    We have the following identities:
\begin{align}
    \begin{split}\label{dddd}
        d_1^*d_1&=-\De_1+\mathbf{K},\qquad\qquad  d_1 d_1^*=-(\De_0,\De_0),\\
        d_2^*d_2&=-\frac{1}{2}\De_2+\mathbf{K},\qquad\quad\; d_2 d_2^*=-\frac{1}{2}(\De_1+\mathbf{K}).
    \end{split}
\end{align}
where $\mathbf{K}$ denotes the Gauss curvature on $S$. 
\end{prop}
\begin{proof}
See for example (2.2.2) in \cite{Ch-Kl}.
\end{proof}
\begin{df}\label{dfcomplexD}
We define derivatives of complex quantities as follows
\begin{itemize}
\item For two scalar functions $u$ and $v$, we define
$$
\DD(u+iv):=(\nab+i\dual\nab)(u+iv).
$$
\item For a 1-form $f$, we define
\begin{align*}
    \DD\c(f+i\dual f):=(\nab+i\dual\nab)\c(f+i\dual f),\\
    \DD\hot(f+i\dual f):=(\nabla+i\dual\nab)\hot(f+i\dual f).
\end{align*}
\item For a symmetric traceless 2-form $u$, we define
$$
\DD\c(U+i\dual U):=(\nab+i\dual\nab)\c(U+i\dual U).
$$
\end{itemize}
\end{df}
\begin{prop}\label{ellipticLp}
The following statements hold for all $p\in(1,+\infty)$:
\begin{enumerate}
\item Let $\phi\in\sk_0$ be a solution of $\De \phi=f$. Then, we have
\begin{align*}
    \|\nab^2 \phi\|_{L^p}+r^{-1}\|\nab\phi\|_{L^p}+r^{-2}\|\phi-\overline{\phi}\|_{L^p}&\les \|f\|_{L^p}.
\end{align*}
\item Let $\phi\in\sk_0$ be a solution of $\left(\De+\frac{2}{r^2}\right)\phi=f$. Then, we have
\begin{align*}
    \|\nab^2 \phi\|_{L^p}+r^{-1}\|\nab\phi\|_{L^p}+r^{-2}\|\phi\|_{L^p}&\les \|f\|_{L^p}+r^{-2}|(\phi)_{\ell=1}|.
\end{align*}
\item Let $\xi\in\sk_1$ be a solution of $d_1\xi=(f,f_*)$. Then we have
\begin{align*}
    \|\nab\xi\|_{L^p}+r^{-1}\|\xi\|_{L^p}\les \|(f,f_*)\|_{L^p}.
\end{align*}
\item Let $U\in\sk_2$ be a solution of $d_2 U=f$. Then we have
\begin{align*}
\|\nab U\|_{L^p}+r^{-1}\|U\|_{L^p}\les\|f\|_{L^p}.
\end{align*}
\end{enumerate}
\end{prop}
\begin{proof}
See Corollary 2.3.1.1 in \cite{Ch-Kl}.
\end{proof}
\begin{df}\label{def:dk-definition}
We define the following weighted differential operator
$$\dk=(\nab_3,r\nab_4,\dkb),$$
with $\dkb$ defined as follows:
\begin{align*}
\dkb(f,f_*)&:=rd_1^*(f,f_*),\qquad \forall (f,f_*)\in\sk_0,\\
\dkb\xi&:=rd_1\xi,\qquad\qquad\qquad\;\forall\xi\in\sk_1,\\
\dkb U&:=rd_2U,\qquad\qquad\quad\;\;\;\forall U\in\sk_2.
\end{align*}
Moreover, for any quantity $X$, we denote
\begin{align}\label{df(k)}
\dkk X:=\left(X,\,\dkb X,\,\dkb^2X,\ldots,\,\dkb^kX\right),\qquad k\in\mathbb{N}.
\end{align}
\end{df}
We now introduce the definition of the approximate basis of $\ell=1$ modes on almost round spheres, which plays an essential role throughout this paper.
\begin{df}\label{jpdef}
    On an $\ep$-almost round sphere $(S, g^S)$, in the sense of
\begin{equation}\label{almostround}
    \left\| K^S - \frac{1}{(r^S)^2} \right\|_{L^\infty(S)} \leq \frac{\ep}{(r^S)^2}.
\end{equation}
\begin{enumerate}
    \item We define an $\ep$--approximate basis of $\ell=1$ modes on $S$ to be a triplet of functions $\Jp$ on $S$ verifying\footnote{The properties  \eqref{def:Jpsphericalharmonics} of the scalar functions $\Jp$ are motivated by the fact that the $\ell=1$ spherical harmonics on the standard sphere $\SSS^2$, given by $J^{(0)}=\cos\th,\, J^{(+)}=\sin\th\cos\vphi,\, J^{(-)}=\sin\th\sin\vphi$, satisfy \eqref{def:Jpsphericalharmonics}  with $\ep=0$. Note also that on $\SSS^2$, there holds
\begin{align*}
\int_{\mathbb{S}^2}(\cos\th)^2=\int_{\mathbb{S}^2}(\sin\th\cos\vphi)^2=\int_{\mathbb{S}^2}(\sin\th\sin\vphi)^2=\frac{4\pi}{3},\qquad |\SSS^2|=4\pi.
\end{align*}}
\begin{align}
\begin{split}\label{def:Jpsphericalharmonics}
(r^2\De+2)\Jp&= O(\ep),\qquad p=0,+,-,\\
\frac{1}{|S|}\int_{S}\Jp J^{(q)}&=\frac{1}{3}\de_{pq}+O(\ep),\qquad p,q=0,+,-,\\
\frac{1}{|S|}\int_{S}\Jp&=O(\ep),\qquad p=0,+,-,
\end{split}
\end{align}
where $\ep>0$ is a sufficiently small constant.
    \item Let $(\Phi, u)$ be the unique, up to isometries of $\mathbb{S}^2$, uniformization\footnote{This was called effective uniformization in  \cite{KS:Kerr2}.} in Corollary 3.8. of \cite{KS:Kerr2}, that is, a unique diffeomorphism $\Phi: \mathbb{S}^2 \to S$ and a unique centered conformal factor $u$ s.t.
\[\Phi^\#(g^S) = (r^S)^2 e^{2u} \gamma_0.
\]
We define the \emph{canonical} choice of $\ep$--approximate basis of $\ell=1$ modes on $S$ by
\[
J^S := J^{\mathbb{S}^2} \circ \Phi^{-1},
\]
where $J^{\mathbb{S}^2}$ denotes the standard $\ell=1$ spherical harmonic on round sphere. 
\end{enumerate}
\end{df}
\begin{proposition}\label{jpprop}
For an $\ep$--approximate basis of $\ell=1$ modes $\Jp$ defined in Definition \ref{jpdef}, we have
\begin{align}\label{estjpepg}
d_2^*d_1^*\Jp=r^{-2}O(\ep).
\end{align}
\end{proposition}
\begin{proof}
See Lemma 5.37 in \cite{KS:main}.
\end{proof}
\begin{df}\label{defl=1}
Given a scalar function $f$ defined on a sphere $S$, we define the $\ell=1$ modes of $f$ w.r.t. a choice of $\ep$--approximate basis of $\ell=1$ modes $\Jp$ by the triplet
$$
(f)_{\ell=1}:=\left\{\frac{1}{|S|}\int_S f \Jp,\quad p=0,+,-\right\}.
$$
\end{df}
\subsection{Outgoing PG structures}\label{section:outgoingPG}
We now recall the definition of \emph{outgoing  Principal Geodesic (PG) structure}, which was introduced and developed in Chapter 2 of \cite{KS:main}.
\begin{df}\label{dfoutgoingPG}
An outgoing PG structure consists of a null pair $(e_3, e_4)$ and the induced horizontal structure $\HH$, together with a scalar function r such that:
\begin{enumerate}
    \item $e_4$ is a null outgoing geodesic vectorfield, i.e. $\D_4e_4=0$,
    \item $r$ is an affine parameter, i.e. $e_4(r)=1$,
    \item the gradient of $r$, defined by $\grad r:=\g^{\a\b}\pr_\b r\pr_\a$, is perpendicular to $\HH$.
\end{enumerate}
\end{df}
\begin{lm}\label{PGconditionslm}
    Given an outgoing PG structure as above, we have
    \begin{align*}
        \xi=0,\qquad \xi=0,\qquad \etab+\ze=0.
    \end{align*}
\end{lm}
\begin{proof}
    See Lemma 2.17 in \cite{KS:main}.
\end{proof}
\begin{rk}
    By abuse of language, we call an outgoing PG structure $\{S(u, r),(e_3, e_4, \HH)\}$ simply an outgoing PG $S(u,r)$--foliation.
\end{rk}
\subsection{\texorpdfstring{$\KSAF$}{} spacetime \texorpdfstring{$(\MM,\g)$}{}}\label{defKSAF}
As mentioned in Remark \ref{remark:final state1}, all our results, (except those in Section \ref{seckick}), apply to any $\KSAF$ spacetime $(\MM,\g)$, which verifies the same properties as the final spacetime $\Mext_\infty$ in Theorem \ref{MainThm-firstversion}. We state below the definition of a $\KSAF$ spacetime, which contains some of the main features of $\MM$. More information can be found in Chapter 2 of \cite{KS:main}.
\begin{df}\label{dfKSAF}
A spacetime $(\M,\g)$ is called a \emph{Klainerman-Szeftel Asymptotically Flat Future Complete} spacetime with final parameters $(a,m)$ $(\KSAF(a,m)$ spacetime$)$, if it satisfies the following properties:
\begin{itemize}
\item $\MM$ is equipped with an outgoing PG structure $\{ S(u,r),(e_3,e_4;\HH)\}$ in the sense of Defintion \ref{dfoutgoingPG}.
\item $\MM$ is endowed with a pair of constants $(a, m)$.\footnote{The constants $(a, m)$  are fixed on $S_*$ according to the results in \cite{KS:Kerr2}. On $\MM_\infty$ they are the precise $(a_f, m_f)$ in Theorem \ref{MainThm-firstversion}.}
\item $\MM$ is endowed with a scalar function $q:=r+ia\cos\th$, with $\th$ a scalar function satisfying 
$$e_4(\th)=0.$$
\item $\MM$ is endowed with complex horizontal 1-forms $\Jk$ and $\Jk_\pm$ satisfying
\begin{align}\label{dfJJk}
    \nab_4(q\Jk)=0,\qquad\qquad\nab_4(q\Jk_\pm)=0.
\end{align}
\item $\MM$ is endowed with a \emph{retarded time} function $u$ verifying $e_4(u)=0$.
\item The initial layer region is defined by:
\begin{align*}
    \LL_0:=\left\{p\in\M\big/\, 0\leq u(p)\leq 3\right\}.
\end{align*}
\item Specific decay rates in $u$ and $r$ hold for all linearized Ricci and curvature coefficients\footnote{For any quantity $W$, we define its linearized quantity by $\widecheck{W}:=W-W_{Kerr}$ where $W_{Kerr}$ denotes its corresponding Kerr value. In other word, $\widecheck{W}$ is defined by subtracting from the Ricci and curvature coefficients the corresponding values in Kerr, expressed relative to $(a,m,r,\cos\th,\Jk)$. We refer the reader to Section 4.1.1 of \cite{GKS} for precise definitions of all such quantities. The real Ricci coefficients $\Ga$ and curvature coefficients $R$ relative to a general horizontal structure are defined in Section 2.2 of \cite{GKS}.} Moreover precisely, the following decay estimates hold on $\M$:\footnote{Recall that the weighted differential operator $\dk$ is defined in Definition \ref{def:dk-definition}.}
\begin{align}\label{mainest}
\begin{split}
|\dk^{\leq\ks}(A,B)|&\les\frac{\ep_0}{r^{\frac{7}{2}+\dec}},\\
\left|\dk^{\leq\ks}\left(\trXc,\,\Xh,\,\Zc,\,\trXbc,r\Pc\right)\right|&\les\frac{\ep_0}{r^2u^{\frac{1}{2}+\dec}},\\
\left|\dk^{\leq\ks}(\Hc,\,\Xbh,\,\ombc,\,\Xib,\,r\Bb,\Ab)\right|&\les\frac{\ep_0}{ru^{1+\dec}},
\end{split}
\end{align}
where $\ep_0$ is a constant small enough which describe the size of perturbation while $\ks$ is a sufficient large integer\footnote{The precise decay estimates in \cite{KS:main} hold for $\ks$ derivatives. The top derivatives $\kl$ have weights in $r$ but no $u$--decay, see Section 3.3 in \cite{KS:main}.} and $\dec>0$ is a sufficiently small constant, which are used to describe the decay estimates in \cite{KS:main}.
\item The spacetime $(\MM,\g) $ is covered by three coordinate systems:
\begin{itemize}
\item[-] in the $(u,r,\th,\vphi)$ coordinates system, we have for $\frac{\pi}{4}<\th<\frac{3\pi}{4}$, 
\begin{align*}
\g=\g_{a,m}+\Big(du, dr, rd\th, r\sin\th d\vphi\Big)^2O\left(\frac{\ep_0}{u^{1+\dec}}\right),
\end{align*}
\item[-] in the $(u,r,x^1,x^2)$ coordinates system, with $x^1=J^{(+)}$ and $x^2=J^{(-)}$\footnote{Here, we denote $J^{(+)}:=\sin\th\cos\vphi$ and $J^{(-)}:=\sin\th\sin\vphi$.}, we have, for $0\leq\th<\frac{\pi}{3}$ and for   $\frac{2\pi}{3}<\th\leq\pi$,
\begin{align*}
\g &= \g_{a,m}+\Big(du, dr, rdx^1, rdx^2\Big)^2O\left(\frac{\ep_0}{u^{1+\dec}}\right).
\end{align*}
See Section 3.4.3 in \cite{KS:main} and Proposition \ref{Mext-metric} below for more details.
\end{itemize}
\end{itemize}
\end{df}
\begin{rk}
    In the sequel, a $\KSAF(a,m)$ spacetime $(\M,\g)$ is simply called a $\KSAF$ spacetime if we do not need to describe its parameters $(a,m)$.
\end{rk}
\subsection{Schematic notation \texorpdfstring{$\Gag$}{} and \texorpdfstring{$\Gab$}{}}\label{subsection:Ga_gGa_b}
\begin{df}\label{dfGagGab}
The Ricci coefficients and curvature components on $\MM$ can be divided into the following groups:
\begin{enumerate}
\item The set $\Gag$ with
 \begin{align*}
 \Gag&:=\Big\{rB, \quad rA,\quad r\Xi,\quad \om,\quad \Hbc,\quad\trXc,\quad \Xh,\quad \Zc,\quad r\Pc,\quad\trXbc\Big\}.
 \end{align*}
 \item The set $\Gab=\Ga_{b,1}\cup\Ga_{b,2}\cup\Ga_{b,3}\cup\Ga_{b,4}$ with
 \begin{align*}
 \Ga_{b,1}&:=\Big\{\Hc,\quad \Xbh,\quad \ombc,\quad \Xib,\quad r\Bb,\quad \Ab\Big\},\\
  \Ga_{b, 2}&:= \Big\{r^{-1}\widecheck{e_3(r)}, \quad \widecheck{\DD q},\quad \widecheck{\DD \ov{q}}, \quad \widecheck{\DD u},  \quad r^{-1}\widecheck{e_3(u)}\Big\},\\
  \Ga_{b,3}&:=\Big\{\widecheck{\DD(J^{(0)})}, \quad \widecheck{\DD(J^{(\pm)})}, \quad e_3(J^{(0)}),\quad\widecheck{e_3(J^{(\pm)})}\Big\},\\
\Ga_{b,4}&:=\left\{r\widecheck{\ov{\DD}\c\Jk}, \quad r\DD\hot\Jk, \quad r\widecheck{\nab_3\Jk},\quad r\widecheck{\ov{\DD}\c\Jk_\pm},\quad r\DD\hot\Jk_\pm, \quad r\widecheck{\nab_3\Jk_\pm} \right\}. 
 \end{align*}
\end{enumerate}
\end{df}
Using the schematic notation $\Gag$ and $\Gab$, the decay estimates in a $\KSAF$ spacetime can be summarized in the following proposition.
\begin{prop}\label{GagGabdecay}
Let $(\M,\g)$ be a $\KSAF$ spacetime endowed with an outgoing PG structure. Then, we have for $k\leq\ks$,
\begin{align*}
\sup_{\M} \Big(r^{\frac{7}{2}+\dec}|\dk^k(A,B)|+r^4u^{1+\dec}|\dk^{k-1}\nab_3(A,B)|\Big)&\les\ep_0,\\
\sup_{\M} \Big(ru^{1+\dec}|\dk^k\Ga_b|+r^2u^{\frac{1}{2}+\dec}|\dk^k\Ga_g|+r^2u^{1+\dec}|\dk^{k-1}\nab_3\Ga_g|\Big)&\les\ep_0.
\end{align*}
\end{prop}
\begin{proof}
    See Proposition 6.49 in \cite{KS:main}.
\end{proof}
\begin{prop}\label{Mext-metric}
    The spacetime $\KSAF$ spacetime $(\M,\g)$ is covered by three coordinate charts $\{(u,r,\th,\vphi):\frac{\pi}{4}<\th<\frac{3\pi}{4}\}$, $\{(u,r,x^1,x^2): 0\leq\th<\frac{\pi}{3}\}$ and $\{(u,r,x^1,x^2):\frac{2\pi}{3}<\th\leq\pi\}$ where $x^1=\sin\th\cos\vphi$ and $x^2=\sin\th\sin\vphi$. More precisely, we have:
    \begin{itemize}
    \item in the $(u,r,\th,\vphi)$ coordinates system for $\frac{\pi}{4}<\th<\frac{3\pi}{4}$, $\g$ is given by:
    \begin{align*}
    \g=&-2(du+dr)du+\left(1+\frac{2mr}{|q|^2}\right)\Big(du-a(\sin\th)^2d\vphi\Big)^2+|q|^2\Big[ (d\th)^2+(\sin\th)^2 (d\vphi)^2\Big]\\
    &+\Big(du,dr,rd\th,r\sin\th d\vphi\Big)^2(\sin\th)^{-2}(r\Ga_{b,2}\cup r\Ga_{b,3}\cup r\Ga_{b,4});
    \end{align*}
    \item in the $(u,r,x^1,x^2)$ coordinates system for $0\leq\th<\frac{\pi}{3}$ or $\frac{2\pi}{3}<\th\leq\pi$, $\g$ is given by
    \begin{align*}
    \g&=-2(du+dr)du+\left(1+\frac{2mr}{|q|^2}\right)\Big(du-a(-x^2dx^1+x^1 dx^2)\Big)^2 \\
    &+\frac{|q|^2}{1-(x^1)^2-(x^2)^2}\Bigg[\Big(1-(x^2)^2\Big)(dx^1)^2+2x^1x^2dx^1dx^2+ \Big(1-(x^1)^2\Big)(dx^2)^2\Bigg]\\
    &+\Big(du,dr,rdx^1,rdx^2\Big)^2(\cos\th)^{-2}(r\Ga_{b,2}\cup r\Ga_{b,3}\cup r\Ga_{b,4});
    \end{align*}
    \item in the $(u,r,\th,\vphi)$ coordinates system for $\frac{\pi}{4}<\th<\frac{3\pi}{4}$, $\g^{-1}$ is given by:\footnote{Here and the line below, $\g_{a,m}$ denotes the Kerr metric with parameter $(a,m)$.}
\begin{align*}
    \g^{-1}=\g^{-1}_{a,m}+\begin{pmatrix}
    r^{-2} & 1 & r^{-2}& r^{-2} \\
    1 & 1 & r^{-1} & r^{-1} \\
    r^{-2} & r^{-1} & r^{-2} & r^{-2}\\
    r^{-2}& r^{-1} & r^{-2} & r^{-2}
\end{pmatrix}(\sin\th)^{-4}(r\Ga_{b,2}\cup r\Ga_{b,3}\cup r\Ga_{b,4});
\end{align*}
\item in the $(u,r,x^1,x^2)$ coordinates system for $0\leq\th<\frac{\pi}{3}$ or $\frac{2\pi}{3}<\th\leq\pi$, $\g^{-1}$ is given by:
    \begin{align*}
    \g^{-1}=\g^{-1}_{a,m}+\begin{pmatrix}
    r^{-2} & 1 & r^{-2}& r^{-2} \\
    1 & 1 & r^{-1} & r^{-1} \\
    r^{-2} & r^{-1} & r^{-2} & r^{-2}\\
    r^{-2}& r^{-1} & r^{-2} & r^{-2}
\end{pmatrix}(r\Ga_{b,2}\cup r\Ga_{b,3}\cup r\Ga_{b,4}).
    \end{align*}
\end{itemize}
\end{prop}
\begin{proof}
See Proposition 4.1 and Corollary 4.5 in \cite{KS:main}.
\end{proof}
\subsection{Commutation formulae}
We have the following commutation formulae in the outgoing PG structure.
\begin{lemma}\label{lemma:comm-gen}
Let $U_{A}= U_{a_1\ldots a_k} $ be a general $k$-horizontal tensorfield. Then, the following commutation identities hold:
\begin{align*}
[\nab_3,\nab_b] U_A&=-\chib_{bc}\nab_c U_{A}+(\eta_b-\ze_b)\nab_3 U_A+\sum_{i=1}^k \big( \chib_{a_ib}\eta_c-\chib_{bc} \eta_{a_i} \big) U_{a_1\ldots }\,^ c\,_{\ldots a_k} \\
&+\sum_{i=1}^k\Big(\chi_{a_ic}\xib_c-\chi_{bc}\,\xib_{a_i}-\in_{a_ic}\dual\bb_b  \Big) U_{a_1\ldots}\,^ c\,_{\ldots a_k}+\xib_b\nab_4 U_{A},\\
[\nab_4,\nab_b]U_A&=-\chi_{bc}\nab_c U_{A}+\sum_{i=1}^k\big(\chi_{a_ib}\etab_c-\chi_{bc} \etab_{a_i} \big) U_{a_1\ldots c\ldots a_k}+\in_{a_ic}\dual\b_b\, U_{a_1\ldots}\,^ c\,_{\ldots a_k},\\
[\nab_4,\nab_3]U_A&=-2(\ze_b+\eta_b)\nab_b U_A-2\omb\nab_4U_A-2\sum_{i=1}^k\big(\eta_{a_i}\ze_b-\ze_{a_i}\eta_b+\in_{a_i b}\dual\rho)U_{a_1\ldots}\,^b\,_{\ldots a_k}.
\end{align*}
\end{lemma}
\begin{proof}
See Lemma 2.2 in \cite{KS:main}.
\end{proof}
\begin{lemma}\label{lemma:commutation-complexM6}
The following commutation formulae are valid.
\begin{enumerate}
\item For a scalar complex function $F$, we have
\begin{equation}\label{eq:commutation-complexM6-scalar}
[\nab_4,\DD]F=-\frac 1 2 \tr X \DD F +r^{-1}\Gag\c\dkb F.
\end{equation}
\item For an anti-self dual horizontal 1-form $U$, we have 
\begin{align}
    [\nab_4,\DD\hot]U=-\frac 1 2\tr X(\DD\hot U-Z\hot U)+r^{-1}\Gag\c\dkb^{\leq 1}U,\label{eq:commutation-complexM6-1form} \\
    [\nab_4,\DDov\c]U=-\frac 1 2\ov{\tr X}\big(\DDov\c U+\ov{Z}\c U)+r^{-1}\Ga_g\c \dkb^{\leq 1}U.\label{eq:commutation-complexM6-1form:div}
\end{align}
\item For an anti-self dual symmetric traceless horizontal 2-form $U$,  we have 
\begin{equation}\label{eq:commutation-complexM6-2form:div}
[\nab_4,\DDov\c]U=-\frac 1 2\ov{\tr X}\big(\DDov\c U+2\ov{Z}\c U)+r^{-1}\Ga_g\c\dkb^{\leq 1}U.
\end{equation}
\end{enumerate}
Moreover, the $\Gag$ appearing in \eqref{eq:commutation-complexM6-scalar}--\eqref{eq:commutation-complexM6-2form:div} contain only $\Xh$, $\Zc$ and $B$.
\end{lemma}
\begin{proof}
See Section 4.2 in \cite{GKS}.
\end{proof}
\begin{cor}\label{commsch}
We have the following schematic commutation formulae:
\begin{align*}
[\nab_4,q\DD]U&=Z\c U+\Gag\c\dko U,\\
[\nab_4,\qb\DDov]U&=\ov{Z}\c U+\Gag\c\dko U,\\
[\nab_4,\nab_3]U&=\eta\c\nab U-2\omb\nab_4U+O(r^{-3})\dko U+O(r^{-1})\Gag\c\dko U,
\end{align*}
where $\Gag$ appeared here contains only $\hch$, $\zec$, $\dual\rho$ and $\b$.
\end{cor}
\begin{proof}
It follows directly from Lemmas \ref{lemma:comm-gen} and \ref{lemma:commutation-complexM6}.
\end{proof}
\subsection{Main equations for an outgoing PG structure}
We recall the null structure equations and the Bianchi equations for an outgoing PG structure, see Appendix \ref{appmain} for these equations in the general setting.
\subsubsection{Main equations in complex notations}
\begin{prop}\label{prop-nullstrandBianchi:complex:outgoing}
In an outgoing PG structure, the null structure equations are given by:
\begin{align*}
\nab_4\tr X +\frac{1}{2}(\tr X)^2&=-\frac{1}{2}\Xh\c\ov{\Xh},\\
\nab_4\Xh+\Re(\tr X)\Xh&=-A,\\
\nab_4\tr\Xb +\frac{1}{2}\tr X\tr\Xb&= -\DD\c\ov{Z}+Z\c\ov{Z}+2\ov{P}-\frac{1}{2}\Xh\c\ov{\Xbh},\\
\nab_4\widehat{\Xb}+\frac{1}{2}\tr X\,\widehat{\Xb}&=-\frac{1}{2}\DD\hot Z+\frac{1}{2}Z\hot Z-\frac{1}{2}\ov{\tr\Xb}\widehat{X},\\
\nab_3\tr\Xb+\frac{1}{2}(\tr\Xb)^2+2\omb\,\tr\Xb&=\DD\c\ov{\Xib}-\Xib\c\ov{Z}+\ov{\Xib}\c(H-2Z)-\frac{1}{2}\Xbh\c\ov{\Xbh},\\
\nab_3\Xbh+\Re(\tr\Xb)\Xbh+2\omb\,\Xbh&=\frac{1}{2}\DD\hot\Xib+\frac{1}{2}\Xib\hot(H-3Z)-\Ab,\\
\nab_3\tr X +\frac{1}{2}\tr\Xb\tr X-2\omb\tr X&=\DD\c\ov{H}+H\c\ov{H}+2P-\frac{1}{2}\Xbh\c\ov{\Xh},\\
\nab_3\Xh+\frac{1}{2}\tr\Xb\,\Xh-2\omb\widehat{X}&=\frac{1}{2}\DD\hot H+\frac{1}{2}H\hot H -\frac{1}{2}\ov{\tr X} \widehat{\Xb},\\
\nab_4Z +\tr X Z&=-\widehat{X}\c\ov{Z} -B,\\
\nab_4H +\frac{1}{2}\ov{\tr X}(H+Z) &=-\frac{1}{2}\Xh\c(\ov{H}+\ov{Z}) -B,\\
\nab_3Z +\frac{1}{2}\tr\Xb(Z+H)-2\omb(Z-H) &= -2\DD\omb -\frac{1}{2}\widehat{\Xb}\c(\ov{Z}+\ov{H})+\frac{1}{2}\tr X\Xib -\Bb+\frac{1}{2}\ov{\Xib}\c\Xh,\\
\nab_3Z +\nab_4\Xib &=-\frac{1}{2}\ov{\tr\Xb}(Z+H) -\frac{1}{2}\Xbh\c(\ov{Z}+\ov{H}) -\Bb,\\
\frac{1}{2}\ov{\DD}\c\Xh +\frac{1}{2}\Xh\c\ov{Z}&=\frac{1}{2}\DD\,\ov{\tr X}+\frac{1}{2}\ov{\tr X}Z-i\Im(\tr X)H -B,\\
\frac{1}{2}\ov{\DD}\c\Xbh -\frac{1}{2}\Xbh\c\ov{Z}&=\frac{1}{2}\DD\,\ov{\tr\Xb}-\frac{1}{2}\ov{\tr\Xb}Z-i\Im(\tr\Xb)(-Z+\Xib)+\Bb,\\
\nab_4\omb-(2\eta+\ze)\c\ze&=\rho.
\end{align*}
The Bianchi equations are given by:
\begin{align*}
\nab_3A-\frac{1}{2}\DD\hot B&=-\frac{1}{2}\tr\Xb A+4\omb A+\frac{1}{2}(Z+4H)\hot B -3\ov{P}\Xh,\\
\nab_4B-\frac{1}{2}\ov{\DD}\c A&=-2\ov{\tr X}B+\frac{1}{2}A\c\ov{Z},\\
\nab_3B-\DD\ov{P}&=-\tr\Xb B+2\omb B+\ov{\Bb}\c\Xh+3\ov{P}H +\frac{1}{2}A\c\ov{\Xib},\\
\nab_4P-\frac{1}{2}\DD\c\ov{B}&=-\frac{3}{2}\tr X P-\frac{1}{2}Z\c\ov{B}-\frac{1}{4}\Xbh\c\ov{A}, \\
\nab_3P +\frac{1}{2}\ov{\DD}\c\Bb &=-\frac{3}{2}\ov{\tr\Xb}P-\frac{1}{2}(\ov{2H-Z})\c\Bb +\Xib\c \ov{B} -\frac{1}{4}\ov{\Xh}\c\Ab, \\
\nab_4\Bb+\DD P &= -\tr X\Bb+\ov{B}\c \Xbh+3P Z,\\
\nab_3\Bb +\frac{1}{2}\ov{\DD}\c\Ab&=-2\ov{\tr\Xb}\,\Bb-2\omb\,\Bb -\frac{1}{2}\Ab\c (\ov{-2Z +H})-3P \,\Xib,\\
\nab_4\Ab +\frac{1}{2}\DD\hot\Bb&=-\frac{1}{2}\ov{\tr X} \Ab +\frac{5}{2}Z\hot \Bb -3P\Xbh.
\end{align*}
\end{prop} 
\begin{proof}
See Proposition 2.19 in \cite{KS:main}.
\end{proof}
\subsubsection{Linearized equations in complex notation}
The following proposition provides the linearized main equations in the $e_4$ direction.
\begin{prp}\label{mainequationsM4}
 In an outgoing PG structure, the linearized null structure equations in the $e_4$ direction are given by:
\begin{align*}
\nab_4(\trXc)+\frac{2}{q}\trXc &=\Gag\c\Gag,\\
\nab_4\Xh+\frac{2r}{|q|^2}\Xh &=-A+\Gag\c\Gag,\\
\nab_4\Zc + \frac{2}{q}\Zc&=-B+r^{-2}\Gag+\Ga_g\c\Ga_g,\\
\nab_4\Hc+\frac{1}{\ov{q}}\Hc &=-B+r^{-1}\Gag+\Ga_b\c\Ga_g,\\
\nab_4\trXbc+\frac{1}{q}\trXbc&=r^{-1}\Gag+\Ga_b\c\Ga_g,\\
\nab_4\Xbh +\frac{1}{q} \Xbh &=-\frac{1}{2}\DD\hot\Zc+r^{-1}\Gag+\Ga_b\c\Ga_g,\\
\nab_4\ombc&=\Re(\Pc)+r^{-2}\Gab+\Gab\c\Gag,\\
\nab_4\Xib+\frac{1}{q}\Xib&=O(r^{-1})\dko\ombc+r^{-2}\Gab+\Gab\c\left(\ombc,\Gag\right).
\end{align*}
The linearized Bianchi equations in the $e_4$ direction are given by:
\begin{align*}
\nab_4B +\frac{4}{\ov{q}}B&=\frac{1}{2}\ov{\DD}\c A+\frac{aq}{2|q|^2}\ov{\Jk}\c A+\Gag\c(B,A),\\
\nab_4\Pc+\frac{3}{q}\Pc&=\frac{1}{2}\DD\c \ov{B}-\frac{a\ov{q}}{2|q|^2}\Jk\c\ov{B}+O(r^{-3})\trXc+r^{-1}\Gag\c\Gag+\Gab\c A,\\
\nab_4\Bb+\frac{2}{q}\Bb&=-\DD\Pc +O(r^{-2})\Pc+O(r^{-3})\Zc+O(r^{-4})\widecheck{\DD\cos\th}+r^{-1}\Gab\c\Gag,\\
\nab_4\Ab+\frac{1}{\qb}\Ab&=O(r^{-1})\dkb^{\leq 1}\Bb+O(r^{-3})\Xbh+\Gag\c\Gab.
\end{align*}
\end{prp}
\begin{proof}
See Lemma 6.15 in \cite{KS:main}.
\end{proof}
\subsubsection{Transport equations for metric components}\label{ssecmetrictransport}
\begin{proposition}\label{prop:e_4(xyz)}
The following equations hold true for the coordinates $(u,r,\th,\vphi)$ associated to an outgoing PG structure
\begin{align*}
e_4(e_3(r))&=-2\omb,\\
\nab_4(\DD u)+\frac{1}{2}\tr X(\DD u)&=-\frac{1}{2}\Xh\c\ov{\DD}u,\\
e_4(e_3(u))&=-\Re\Big((Z+H)\c\ov{\DD} u\Big),\\
\nab_4(\DD\cos\th)+\frac{1}{2}\tr X(\DD\cos\th)&=-\frac{1}{2}\Xh\c\ov{\DD}(\cos \th),\\
e_4(e_3(\cos\th))&=-\Re\Big((Z+H)\c\ov{\DD}(\cos\th)\Big).
\end{align*}
\end{proposition}
\begin{proof}
See Proposition 6.10 in \cite{KS:main}.
\end{proof}
\begin{prop}\label{prop:e_4J}
The following equations hold true for the coordinates $(u,r,\th,\vphi)$ associated to an outgoing PG structure
\begin{align*}
\nab_4(q\DD\Jp)&=Z\c\Jp+\Gag\c\dko\Jp,\\
\nab_4(\nab_3\Jp)&=\eta\c\nab\Jp+O(r^{-3})\dko\Jp+r^{-1}\Gag\c\dko\Jp,\\
\nab_4\left(|q|^2\DDb\c\Jk\right)&=\ov{Z}\c q\Jk+\Gag\c\dko(q\Jk),\\
\nab_4\left(|q|^2\DDb\c\Jk_\pm\right)&=\ov{Z}\c q\Jk_\pm+\Gag\c\dko(q\Jk_\pm),\\
\nab_4\left(|q|^2\DDb\hot\Jk\right)&=\ov{Z}\c q\Jk+\Gag\c\dko(q\Jk),\\
\nab_4\left(|q|^2\DDb\hot\Jk_\pm\right)&=\ov{Z}\c q\Jk_\pm+\Gag\c\dko(q\Jk_\pm),\\
\nab_4(\nab_3(q\Jk))&=H\c\nab(q\Jk)+O(r^{-3})\dkb^{\leq 1}(q\Jk)+r^{-1}\Gag\c\dko(q\Jk),\\
\nab_4(\nab_3(q\Jk_\pm))&=H\c\nab(q\Jk_\pm)+O(r^{-3})\dkb^{\leq 1}(q\Jk_\pm)+r^{-1}\Gag\c\dko(q\Jk_\pm),
\end{align*}
where $\Jp$ is defined in Definition \ref{jpdef} and $\Jk$, $\Jk_\pm$ are introduced in \eqref{dfJJk}.
\end{prop}
\begin{proof}
The proof follows directly from that
\begin{align*}
    \nab_4(\Jp)=0,\qquad\nab_4(q\Jk)=0,\qquad\nab_4(q\Jk_\pm)=0,
\end{align*}
and Corollary \ref{commsch}.
\end{proof}
\subsection{Null frame transformation}
Consider two null frames $(e_4,e_3,e_1,e_2)$ and $(e_4',e_3',e_1',e_2')$ on a $\KSAF$ spacetime $(\MM,\g)$ with $\HH=\{e_3, e_4\}^\perp$ and $\HH'=\{e_3',e_4'\}^\perp$  the corresponding horizontal structures. We also denote $(\Ga',R')$ and $(\Ga,R)$ the connection coefficients and curvature components relative to the two frames. We denote by $\nab, \nab \hot,\div,\curl,\nab_3,\nab_4$ the standard operators corresponding to $\HH$ and by $\nab',\nab' \hot',\div',\curl',\nab'_3,\nab'_4$ those corresponding to $\HH'$. The goal is to establish the transformation formulae between the Ricci and curvature coefficients of the two frames.
\begin{lm}\label{transformation}
A general transformation between two null frames $(e_4, e_3, e_1, e_2)$ and $(e_4', e_3', e_1', e_2')$ on $\MM$ can be written in the following form:
\begin{align}
\begin{split}\label{General-frametransformation}
e_4'&=\la\left(e_4+f^be_b+\frac 1 4|f|^2e_3\right),\\
e_a'&=\left(\de_a^b +\frac{1}{2}\fb_af^b\right)e_b+\frac 1 2\fb_a e_4 +\left(\frac 1 2 f_a +\frac{1}{8}|f|^2\fb_a\right)e_3,\\
e_3'&=\la^{-1}\left(\left(1+\frac{1}{2}f\c\fb+\frac{1}{16}|f|^2|\fb|^2\right) e_3 + \left(\fb^b+\frac 1 4 |\fb|^2f^b\right) e_b+\frac{1}{4}|\fb|^2 e_4\right),
\end{split}
\end{align}
where $(f,\fb,\la)$ are called the \emph{transition functions} of the frame transformation.
\end{lm}
\begin{proof}
See Lemma 2.10 in \cite{KS:main}.
\end{proof}
\begin{rk}
As a consequence of the above lemma, $(f,\fb)$ can be regarded as horizontal vectors on both $\HH=\{e_3,e_4\}^\perp$ and $\HH'=\{e_3',e_4'\}^\perp$, see Definition 2.32 and Remark 2.33 in \cite{graf} for more explanations.
\end{rk}
The following propositions will be used frequently throughout this paper. 
\begin{prp}\label{Riccitransfer}
Under a null frame transformation of the type \eqref{General-frametransformation}, the Ricci coefficients transform as follows:
\begin{itemize}
\item The transformation formulae for $\xi$ and $\xib$ are given by 
\begin{align*}
\la^{-2}\xi'&=\xi+\frac{1}{2}\la^{-1}\nab'_4f+\frac{1}{4}\trch\, f+O(r^{-2})f+f\c\Gag+\lot,\\
\la^2\xib'&=\xib+\frac{1}{2}\la\nab_3'\fb+\frac{1}{4}\trchb\,\fb+O(r^{-2})\fb+\fb\c\Gab+\lot
\end{align*}
\item The transformation formulae for $\chi$ are given by 
\begin{align*}
\la^{-1}\trch'&=\trch+\sdiv'f+f\c\eta+O(r^{-2})f+(f,\fb)\c\Gag+\lot,\\
\la^{-1}\atrch'&=\atrch+\curl'f+f\wedge\eta+O(r^{-2})f+(f,\fb)\c\Gag+\lot,\\
\la^{-1}\hch'&=\hch+\nab'\hot f+f\hot\eta+O(r^{-2})f+(f,\fb)\c\Gag+\lot
\end{align*}
\item The transformation formulae for $\chib$ are given by 
\begin{align*}
\la\trchb'&=\trchb+\div'\fb+f\c\xib+O(r^{-2})\fb+\fb\c\Gag+\lot,\\
\la\atrchb'&=\atrchb+\curl'\fb+f\wedge\xib+O(r^{-2})\fb+\fb\c\Gag+\lot,\\
\la\hchb' &= \hchb +\nab'\hot\fb+f\hot\xib+O(r^{-2})\fb+\fb\c\Gag+\lot
\end{align*}
\item  The transformation formula for $\ze$ is given by 
\begin{align*}
\ze'&=\ze-\nab'(\log\la)-\frac{1}{4}\trchb f-\omb\, f+\frac{1}{4}\trch\,\fb-\frac{1}{2}\hchb\c f\\
&+O(r^{-2})(f,\fb)+r^{-1}(f,\fb)\c\dko(f,\fb)+(f,\fb)\c\Gag+\lot
\end{align*}
\item The transformation formulae for $\eta$ and $\etab$ are given by
\begin{align*}
\eta'&=\eta+\frac{1}{2}\la\nab_3'f+\frac{1}{4}\trch\,\fb-\omb\, f+O(r^{-2})\fb+\fb\c\Gag+\lot,\\
\etab'&=\etab+\frac{1}{2}\la^{-1}\nab'_4\fb+\frac{1}{4}\trchb\,f+\frac{1}{2}f\c\hchb+O(r^{-2})f+(f,\fb)\c\Gag+\lot
\end{align*}
\item The transformation formulae for $\om$ and $\omb$ are given by
\begin{align*}
\la^{-1}\om'&=\om-\frac{1}{2}\la^{-1}e'_4(\log\la)+(f,\fb)\c\Gag+\lot,\\
\la\omb'&=\omb+\frac{1}{2}\la e_3'(\log\la)-\frac{1}{2}\fb\c\eta +\frac{1}{2}f\c\xib-\frac{1}{4}\la \fb\c\nab_3'f+O(r^{-2})\fb+\fb\c\Gag+\lot,
\end{align*}
\end{itemize}
where $\lot$ denote expressions of the type
\begin{align*}
\lot=O\left((f,\fb)^2\right)\c\Ga.
\end{align*}
\end{prp}
\begin{proof}
See Proposition 3.3 in \cite{KS:Kerr1}.
\end{proof}
\begin{prop}\label{Curvaturetransfer}
Under a transformation of type \eqref{General-frametransformation}, the curvature components transform as follows:
\begin{align*}
\la^{-2}\a'&=\a+\big(f\hot\b-\dual f\hot\dual\b\big)+O(f^2)(\rho,\rhod)+O(f^3)\bb+O(f^4)\aa,\\
\la^{-1}\b'&=\b+\frac 3 2\big(f\rho+\dual f\rhod\big)+\frac 1 2\a\c\fb+O(f^2)\bb+O(f^3)\aa,\\
\rho'&=\rho+\fb\c\b-f\c\bb +O\big((f,\fb)^2\big)(\a,\rho,\rhod,\aa),\\
\rhod'&=\rhod-\fb\c\dual\b-f\c\dual\bb +O\big((f,\fb)^2\big)(\a,\rho,\rhod,\aa),\\
\la\bb'&=\bb-\frac 3 2\big(\fb\rho+\dual\fb\rhod\big)-\frac 1 2\aa\c f+O(\fb^2)\b+O(\fb^3)\a,\\
\la^2\aa'&=\aa-\big(\fb\hot\bb-\dual \fb \hot \dual\bb\big)+O(\fb^2)(\rho,\rhod)+O(\fb^3)\b+O(\fb^4)\a.
\end{align*}
The mass aspect function $\mu$ transforms as follows:
\begin{align*}
\mu'&=\mu+(\De'+V)\ovla+\left(\omb+\frac{1}{4}\trchb\right)(\trch'-\trch)-\left(\om+\frac{1}{4}\trch\right)(\trchb'-\trchb)\\
&+r^{-2}\dko\left((f,\fb,\ovla)\c\dko(f,\fb,\ovla)+(f,\fb,\ovla)\c r\Gab\right),
\end{align*}
with
\begin{align*}
    V:=-\frac{1}{2}\trch\trchb-\trch\,\omb-\trchb\,\om,\qquad\quad\ovla:=\la-1.
\end{align*}
\end{prop}
\begin{proof}
    See Proposition 3.3 and Lemma 4.3 in \cite{KS:Kerr1}.
\end{proof}
\section{Asymptotic behavior near future null infinity}\label{seclimit}
In this section, we study the asymptotic behavior of the geometric quantities near the future null infinity. Throughout Section \ref{seclimit}, we focus on the outgoing PG structure of a $\KSAF$ spacetime $(\M,\g)$.
\subsection{Limits of geometric quantities}
\begin{lem}\label{lhopital}
Let $h$ be a horizontal tensor field defined in $(\M,\g)$ and let $p>1$ be a constant. Assume that we have
    \begin{align*}
        \lim_{C_u,r\to\infty}h=0,\qquad \lim_{C_u,r\to\infty}r^p\nab_4(h)=h_{\infty},
    \end{align*}
where $h_\infty$ is a function depending on $u$ and $\th^a$. Then, we have
    \begin{align*}
        \lim_{r\to\infty}r^{p-1}h=\frac{h_\infty}{1-p}.
    \end{align*}
\end{lem}
\begin{proof}
    We have, consistent with  L'h\^opital's rule,
    \begin{align*}
    \lim_{C_u,r\to\infty}r^{p-1}h=\lim_{C_u,r\to\infty}\frac{h}{r^{1-p}}=\lim_{C_u,r\to\infty}\frac{\nab_4h}{e_4(r^{1-p})}=\lim_{C_u,r\to\infty}\frac{e_4(h)}{(1-p)(r^{-p})}=\frac{h_\infty}{1-p}.
    \end{align*}
    This concludes the proof of Lemma \ref{lhopital}.
\end{proof}
\begin{prp}\label{limitexist}
For each fixed $u$, the following limits exist on $\II^+$: 
\begin{align*}
(\Jscr,\Jscr_\pm)&:=\lim_{C_u,r\to\infty}r(\Jk,\Jk_\pm),\\
(\Xscr,\Xbscr)&:=\lim_{C_u,r\to\infty}r^2\left(\trchc,\trchb+\frac{2}{r}\right),\\
(\The,\Thb)&:=\lim_{C_u,r\to\infty}(r^2\hch,r\hchb),\\
\left(\aXscr,\aXbscr\right)&:=\lim_{C_u,r\to\infty}r^2\left(\atrch,\atrchb\right),\\
(\Zscr,\Hscr,\Ybscr)&:=\lim_{C_u,r\to\infty}r(r\ze,\eta,\xib),\\
(\Mscr,\Mbscr,\Pscr,\dual\Pscr)&:=\lim_{C_u,r\to\infty}r^3(\mu,\mub,\rho,\rhod),\\
\Bbscr&:=\lim_{C_u,r\to\infty}r^2\bb,\\
\Abscr&:=\lim_{C_u,r\to\infty}r\aa,\\
\Wbscrone&:=\lim_{C_u,r\to\infty}r\omb,\\
\Wbscr&:=\lim_{C_u,r\to\infty}r(r\omb-\Wbscrone).
\end{align*}
Moreover, we have the following identities:
\begin{align}
\begin{split}\label{limitidentities}
\Mscr&=-\divo\Zscr-\Pscr+\frac{1}{2}\The\c\Thb,\\
\Mbscr&=\divo\Zscr-\Pscr+\frac{1}{2}\The\c\Thb,\\
\Wbscrone&=0,\\
\Wbscr&=-\frac{\Pscr}{2}-\Hscr\c\Zscr,\\
\Xbscr&=-\Xscr+2\divo\Zscr-2\Pscr+\The\c\Thb,\\
\aXbscr&=-\aXscr+2\curlo\Zscr-2\dual\Pscr+\The\wedge\Thb,\\
\divo\The&=\frac{1}{2}\nabo\Xscr+\Zscr-\dual\nabo\,\aXscr-\aXscr\,\dual\Hscr,\\
\divo\Thb&=\Bbscr.
\end{split}
\end{align}
\end{prp}
\begin{rk}
    Notice that we have from Proposition \ref{GagGabdecay} that $r\omb=O(\ep_0)$. However, we can prove that $\Wbscrone=0$, see \eqref{Wbscronevanish} below.
\end{rk}
\begin{proof}[Proof of Theorem \ref{limitexist}]
The existence of $\Jscr$ and $\Jscr_\pm$ follow directly from
\begin{align*}
    \nab_4(q\Jk)=0,\qquad \nab_4(q\Jk_\pm)=0,
\end{align*}
and the fact that $q=r+ia\cos\th$.\\ \\
Next, we have from Proposition \ref{mainequationsM4} and the definitions of $\Gag$ and $\Gab$:
\begin{align}
    \begin{split}\label{nab4eqimportant}
        \nab_4(q^2\trXc)&=r^2\Gag\c\Gag,\\
        \nab_4(|q|^2\Xh)&=O(r^2)A+r^2\Gag\c\Gag,\\
        \nab_4(q^2\Zc)&=O(r^2)B+O(1)\Gag+r^2\Gag\c\Gag,\\
        \nab_4(\qb\Hc)&=O(r)B+\Gag+r\Gag\c\Gab,\\
        \nab_4(q\Xhb)&=\dko\Gag+r\Gag\c\Gab,\\
        \nab_4(q^3P)&=O(r^2)\dko B+r^3\Gab\c(A,B),\\
        \nab_4(q^2\Bb)&=O(1)\dko\Gag+r\Gag\c\Gab,\\
        \nab_4(\qb\Ab)&=O(1)\dko\Bb+r^{-2}\Gab+r\Gag\c\Gab,
    \end{split}
\end{align}
where $\Gag$ and $\Gab$ are defined in Definition \ref{dfGagGab}. Thus, we deduce from Proposition \ref{GagGabdecay}
\begin{align*}
\nab_4\left(q^2\trXc,\,|q|^2\Xh,\,\qb H,\,q\Xhb,q^2\Bb,\qb\Ab\right)&=O\left(\frac{\ep_0}{r^2}\right),\\
\nab_4\left(q^3P,q^2Z\right)&=O\left(\frac{\ep_0}{r^{\frac{3}{2}+\dec}}\right).
\end{align*}
The integrability of the right hand sides implies that the limit quantities $\Xscr$, $\aXscr$, $\The$, $\Thb$, $\Zscr$, $\Pscr$, $\Qscr$, $\Hscr$, $\Bbscr$, $\Abscr$, $\Mscr$ and $\Mbscr$\footnote{The existences of $\Mscr$ and $\Mbscr$ follow from the definition of $\mu$ and $\mub$ and the existences of $\The$, $\Thb$, $\Pscr$ and $\Zscr$.} are well defined.\\ \\
Next, since $\left|r^2\trXbc\right|\les\ep_0$,
\begin{align*}
    \lim_{C_u,r\to\infty}q\tr\Xb=-2.
\end{align*}
We also have from Proposition \ref{prop-nullstrandBianchi:complex:outgoing}
\begin{align}\label{nab4trXb}
\nab_4\tr\Xb +\frac{1}{q}\tr\Xb=-\frac{1}{2}\trXc \tr\Xb-\DD\c\ov{Z}+Z\c\ov{Z}+2\ov{P}-\frac{1}{2}\Xh\c\ov{\Xbh},
\end{align}
which implies
\begin{align*}
    \nab_4(q\tr\Xb)=-\frac{q}{2}\trXc\tr\Xb-q\DD\c\ov{Z}+qZ\c\ov{Z}+2q\ov{P}-\frac{q}{2}\Xh\c\ov{\Xbh}.
\end{align*}
    Multiplying it by $r^2$ and taking $r\to\infty$, we obtain
    \begin{align}
    \begin{split}\label{r2nab4trXb}
        \lim_{C_u,r\to\infty}r^2\nab_4(q\tr\Xb)&=(\Xscr-i\aXscr)-(\nabo+i\nabo)\c(\Zscr-i\dual\Zscr)+2(\Pscr-i\dual\Pscr)\\
        &-\frac{1}{2}(\The+i\The)\c(\Thb-i\Thb).
    \end{split}
    \end{align}
    Applying Lemma \ref{lhopital} with $p=2$ and $h=q\tr\Xb+2$, we deduce the existence of the following limit:
    \begin{align*}
        \Xbscr-i\aXbscr:=\lim_{C_u,r\to\infty} r^2\left(\tr\Xb+\frac{2}{q}\right)=-\lim_{C_u,r\to\infty}r^2\nab_4(q\tr\Xb).
    \end{align*}
    Moreover, we have from \eqref{r2nab4trXb}
    \begin{align*}
        \Xbscr&=-\Xscr+2\divo\Zscr-2\Pscr+\The\c\Thb,\\
        \aXbscr&=-\aXscr+2\curlo\Zscr-2\dual\Pscr+\The\wedge\Thb.
    \end{align*}
    Next, we have from Proposition \ref{prop-nullstrandBianchi:complex:outgoing}
    \begin{align*}
        \nab_4\omb=\rho+2\eta\c\ze+\ze\c\ze.
    \end{align*}
    Thus, we obtain
    \begin{align}\label{r3nab4omb}
        \lim_{C_u,r\to\infty}r^3\nab_4\omb=\lim_{C_u,r\to\infty}r^3\rho+2r\eta\c r^2\ze+r^{-1}r^2\ze\c r^2\ze=\Pscr+2\Hscr\c\Zscr.
    \end{align}
    Hence, we have from Lemma \ref{lhopital}:
    \begin{align*}
        \Wbscr:=\lim_{C_u,r\to\infty}r^2\omb=-\frac{1}{2}\lim_{C_u,r\to\infty}r^3\nab_4\omb=-\frac{\Pscr}{2}-\Hscr\c\Zscr.
    \end{align*}
    Note that the existence of $\Wbscr$ immediately implies that
    \begin{align}\label{Wbscronevanish}
        \Wbscrone:=\lim_{C_u,r\to\infty}r\omb=0.
    \end{align}
    Finally, we recall from Proposition \ref{mainequationsM4}
    \begin{align*}
        \nab_4\Xib+\frac{1}{q}\Xib=r^{-1}\dko\omb+r^{-2}\Gab+\Gab\c(\ombc,\Gag).
    \end{align*}
    Hence, we obtain
    \begin{align*}
        r^2\nab_4(q\Xib)=O(\ep_0),
    \end{align*}
    which implies the existence of the following limit:
    \begin{align*}
        \Ybscr:=\lim_{C_u,r\to\infty}q\Xib.
    \end{align*}
It remains to prove the last two identities in \eqref{limitidentities}. To this end, we have from Proposition \ref{prop-nullstr}
\begin{align*}
\div\hch +\ze\c\hch &= \frac{1}{2}\nab\trch+\frac{1}{2}\trch\,\ze -\frac{1}{2}\dual\nab\atrch-\frac{1}{2}\atrch\dual\ze -\atrch\dual\eta-\b,\\
\div\hchb-\ze\c\hchb &= \frac{1}{2}\nab\trchb-\frac{1}{2}\trchb\,\ze -\frac{1}{2}\dual\nab\atrchb+\atrchb\dual\ze-\atrch\dual\xib +\bb.
\end{align*}
Taking $r\to\infty$ and combining with the existence of the limits, we infer
\begin{align*}
    \divo\The&=\frac{1}{2}\nabo\Xscr+\Zscr-\dual\nabo\,\aXscr-\aXscr\,\dual\Hscr,\\
    \divo\Thb&=\Bbscr.
\end{align*}
This concludes the proof of Proposition \ref{limitexist}.
\end{proof}
\subsection{Limits of pointwise physical quantities}
\begin{prp}\label{limitM}
    For each fixed $u$, the following limit exists:
    \begin{align*}
    \Mk:=\lim_{C_u,r\to\infty}\frac{r^3}{2}(\mu+\mub),
    \end{align*}
    Moreover, the following identity holds:
    \begin{align*}
        \Mk=\frac{\Mscr+\Mbscr}{2}=-\Pscr +\frac{1}{2}\The\c\Thb.
    \end{align*}
    We call $\Mk$ the \emph{Bondi mass aspect function}. 
\end{prp}
\begin{rk}
Note that $\Mk$ is closely connected to the Bondi mass aspect function on the future null infinity in the Bondi-Sachs gauge, see Definition A.5 in \cite{CKWWY}. In Section \ref{ssecphy}, we use $\Mk$ relative to a LGCM foliation to define the null energy and linear momentum.
\end{rk}
\begin{proof}[Proof of Proposition \ref{limitM}]
    We have from Proposition \ref{limitexist}
    \begin{align*}
        \lim_{C_u,r\to\infty}\frac{r^3}{2}(\mu+\mub)=\frac{\Mscr+\Mbscr}{2}=-\Pscr+\frac{1}{2}\The\c\Thb.
    \end{align*}
    where we used the first two equations of \eqref{limitidentities} in the last step. This concludes the proof of Proposition \ref{limitM}.
\end{proof}
We now consider the limit of $r^4\b$ on the future null infinity, closely connected to the \emph{angular momentum aspect} on the future null infinity in the Bondi-Sachs gauge, see Definition A.5 in \cite{CKWWY}. We have from Proposition \ref{mainequationsM4}
\begin{align}\label{nab4B}
    \nab_4(\qb^4 B)=\frac{\qb^4}{2}\DDb\c A+\frac{a\qb^3}{2}\ov{\Jk}\c A+r^4\Gag\c(B,A).
\end{align}
Due to the lack of integrability of $\qb^4(\DDb\c A)$, the limit $\BB:=\lim_{C_u,r\to\infty}r^4\b$ may not exist. To overcome this  difficulty, we proceed as in the proof of Theorem M4 in \cite{KS:main} by considering instead the corresponding equation for  $r^5(d_1\b)_{\ell=1}$, which eliminates the first term on the R.H.S. of \eqref{nab4B}, since it is supported on the $\ell\geq 2$ modes. We recall the following lemma.
\begin{lem}\label{Proposition:Identity.ell=1Div B}
We introduce the following renormalized quantities:\footnote{The terms involving $\Jk$ in \eqref{dfren} are added in order to eliminate the second term on the R.H.S. of \eqref{nab4B}.}
\begin{align}
\begin{split}\label{dfren}
[B]_{ren}&:=B-\frac{3a}{2}\ov{\Pc}\Jk-\frac{a}{4}\ov{\Jk}\c A,\\
[\ov{\DD}\c]_{ren}&:=\ov{\DD}\c-\frac{a}{2}\ov{\Jk}\c\nab_4-\frac{a}{2}\ov{\Jk}\c\nab_3.
\end{split}
\end{align}
Then, the following identities hold true
\begin{align}\label{eq:Proposition.Identity.ell=1Div B1}
\begin{split}
&\;\;\;\,\,\,\nab_4\left(\int_{S(r,u)}\frac{rJ^{(0)}}{\Si}[\ov{\DD}\c]_{ren}\left(r^4[B]_{ren}\right)\right) \\
&= O(1)\dk^{\leq 1}\Xh+O(r)\dk^{\leq 2}B+O(r^2)\dk^{\leq 1}\nab_3B+O(r)\dk^{\leq 2}\Pc +O(1)\dk^{\leq 1}\trXc\\
&+O(r^{2})\dk^{\leq 1}\nab_3A+O(r)\dk^{\leq 2}A+r^4\dk^{\leq 1}\big(\Gag\c(B,A)\big)+r^4\dk^{\leq 1}\big(\Gab\c\nab_3 A\big)
\\
&+r^2\dk^{\leq 2}\big(\Gag\c\Gag\big)+O\left(\frac{r^2\ep}{u^{\frac{1}{2}+\dec}}\right)\dk^{\leq 1}\Big(A, B, r^{-1}\Pc\Big),
\end{split}
\end{align}  
and
\begin{align}\label{eq:Proposition.Identity.ell=1Div B2}
\begin{split}
&\;\;\;\,\,\,\nab_4 \left(\int_{S(r,u)}\frac{rJ^{(\pm)}}{\Si}[\ov{\DD}\c]_{ren} \left( r^4[B]_{ren} \right)\right)\mp\frac{a}{r^2}\int_{S(r,u)} \frac{rJ^{(\mp)}}{\Si}[\ov{\DD}\c]_{ren}\left(r^4[B]_{ren}\right)\\
&=O(1)\dk^{\leq 1}\Xh+O(r)\dk^{\leq 2}B+O(r^2)\dk^{\leq 1}\nab_3B+O(r)\dk^{\leq 2}\Pc\\
&+O(1)\dk^{\leq 1}\trXc
+O(r^{2})\dk^{\leq 1}\nab_3A+O(r)\dk^{\leq 2}A+r^4\dk^{\leq 1}\big(\Gag\c(B,A)\big)\\
&+r^4\dk^{\leq 1}\big(\Gab\c\nab_3 A\big)
+r^2\dk^{\leq 2}\big(\Gag\c\Gag\big)
+O\left(\frac{r^2\ep}{u^{\frac{1}{2}+\dec}}\right)\dk^{\leq 1}\Big(A, B, r^{-1}\Pc\Big),
\end{split}
\end{align}
where we denoted
\begin{align*}
    \Si^2:=(r^2+a^2)^2-a^2(\sin\th)^2\De,\qquad\quad\De:=r^2-2mr+a^2.
\end{align*}
\end{lem}
\begin{proof}
See Proposition 6.43 in \cite{KS:main}.
\end{proof}
\begin{prp}\label{limitB}
For each fixed $u$, the following limits exist on $\II^+$:
\begin{align*}
(\Bk,\Bkd):=\lim_{C_u,r\to\infty}r^5(\div\b,\curl\b)_{\ell=1}.
\end{align*}
\end{prp}
\begin{proof}
We have from Lemma \ref{Proposition:Identity.ell=1Div B} that for $p=0,+,-$,
\begin{align*}
&\;\;\;\,\,\,\nab_4\left(\int_{S(u,r)}\frac{r\Jp}{\Si}[\ov{\DD}\c]_{ren}\left( r^4[B]_{ren} \right)\right)+O(r^{-2})\left(\int_{S(u,r)}\frac{r\Jp}{\Si}[\ov{\DD}\c]_{ren}\left( r^4[B]_{ren} \right)\right)\\
&=\dk^{\leq 2}\Gag+O(r^2)\dk^{\leq 1}\nab_3(A,B)+r^4\dk^{\leq 1}\big(\Gag\c(A,B)\big)+r^4\dk^{\leq 1}\big(\Gab\c\nab_3 A\big)
\\
&+r^2\dk^{\leq 2}\big(\Gag\c\Gag\big)+O\left(\frac{r^2\ep}{u^{\frac{1}{2}+\dec}}\right)\dk^{\leq 1}\Big(A, B, r^{-1}\Pc\Big).
\end{align*}
Combining with Proposition \ref{GagGabdecay}, we infer
\begin{align*}
    \nab_4\left(\int_{S(u,r)}\frac{r\Jp}{\Si}[\ov{\DD}\c]_{ren}\left(r^4[B]_{ren} \right)\right)=O\left(r^{-{\frac{3}{2}-\dec}}\right).
\end{align*}
Hence, we deduce the existence of the following limit:\footnote{Notice that $\Si\sim r^2$ as $r\to\infty$.}
\begin{align*}
    \lim_{C_u,r\to\infty}\left(r[\DDb\c]_{ren}\left(r^4[B]_{ren}\right)\right)_{\ell=1}.
\end{align*}
Combining with Proposition \ref{limitexist}, we infer the existence of the following limit:
\begin{align*}
    \lim_{C_u,r\to\infty}r^5(\DDb\c B)_{\ell=1}.
\end{align*}
Noticing that we have from Definition \ref{dfcomplexD}
\begin{align*}
    \DDb\c B=2\div\b+2i\curl\b,
\end{align*}
this concludes the proof of Proposition \ref{limitB}.
\end{proof}
\subsection{Taylor expansions near future null infinity}
\begin{df}\label{dfHolder}
    Let $U$ be a horizontal tensorfield defined on a $\KSAF$ spacetime $(\M,\g)$. We denote
    \begin{align*}
        U=\OO_k(r^{-m}),
    \end{align*}
    if it satisfies for all $0\leq l\leq k$
    \begin{align*}
        |\nab_4^l U|\les r^{-m-l}.
    \end{align*}
\end{df}
We have the following theorem, which describes the asymptotic behaviors of all geometric quantities near the future null infinity $\II^+$.
\begin{thm}\label{expansionexist}
Let $(\M,\g)$ be a $\KSAF$ spacetime endowed with an outgoing PG foliation. Then, we have the following Taylor expansions:
\begin{align}
\begin{split}\label{GagTaylor}
\trch&=\frac{2}{r}+\frac{\Xscr}{r^2}+\OO_1\left({r^{-\frac{5}{2}-\dec}}\right),\\
\atrch&=\frac{\aXscr}{r^2}+\OO_1\left({r^{-\frac{5}{2}-\dec}}\right),\\
\ze&=\frac{\Zscr}{r^2}+\OO_1\left({r^{-\frac{5}{2}-\dec}}\right),\\
\hch&=\frac{\The}{r^2}+\OO_1\left({r^{-\frac{5}{2}-\dec}}\right),\\
(\mu,\mub)&=\frac{(\Mscr,\Mbscr)}{r^3}+\OO_1\left({r^{-\frac{7}{2}-\dec}}\right),\\
(\rho,\dual\rho)&=\frac{(\Pscr,\dual\Pscr)}{r^3}+\OO_1\left({r^{-\frac{7}{2}-\dec}}\right),
\end{split}
\end{align}
and
\begin{align}
\begin{split}\label{GabTaylor}
\trchb&=-\frac{2}{r}+\frac{\Xbscr}{r^2}+\OO_1\left({r^{-\frac{5}{2}-\dec}}\right),\\
\atrchb&=\frac{\aXbscr}{r^2}+\OO_1\left({r^{-\frac{5}{2}-\dec}}\right),\\
\omb&=\frac{\Wbscr}{r^2}+\OO_1\left({r^{-\frac{5}{2}-\dec}}\right),\\
\hchb&=\frac{\Thb}{r}+\frac{-\The+\frac{1}{2}\Xscr\Thb+\frac{1}{2}\dual\Thb\,\aXscr+\nabo\hot\Zscr}{r^2}+\OO_1\left(r^{-\frac{5}{2}-\dec}\right),\\
\eta&=\frac{\Hscr}{r}+\frac{\frac{1}{2}\Xscr\Hscr+\The\c\Hscr+\Zscr-\frac{1}{2}\aXscr\,\dual\Hscr}{r^2}+\OO_1\left(r^{-\frac{5}{2}-\dec}\right),\\
\xib&=\frac{\Ybscr}{r}+\frac{\frac{1}{2}\Xscr\Ybscr-2\nabo\Wbscr-\aXbscr\,\dual\Hscr+\The\c\Ybscr+\frac{1}{2}\aXscr\,\dual\Ybscr}{r^2}+\OO_1\left(r^{-\frac{5}{2}-\dec}\right),\\
\bb&=\frac{\Bbscr}{r^2}+\frac{\nabo\Pscr-\dual\nabo\,\dual\Pscr+\Xscr\Bbscr+\aXscr\,\dual\Bbscr}{r^3}+\OO_1\left(r^{-\frac{7}{2}-\dec}\right),\\
\aa&=\frac{\Abscr}{r}+\frac{\nabo\hot\Bbscr+\frac{1}{2}(\Xscr\Abscr+\aXscr\dual\Abscr)}{r^2}+\frac{\Abscr_3}{r^3}+\OO_1\left(r^{-\frac{7}{2}-\dec}\right),
\end{split}
\end{align}
where
\begin{align*}
    \Abscr_3&:=\frac{1}{2}\nabo\hot\left(\nabo\Pscr-\dual\nabo\,\dual\Pscr+\Xscr\Bbscr+\aXscr\,\dual\Bbscr\right)-\frac{5}{2}\Zscr\hot\Bbscr+\frac{3}{2}(\Pscr\Thb+\dual\Pscr\dual\Thb)\\
    &+\frac{\Xscr}{4}\left(\nabo\hot\Bbscr+\frac{1}{2}(\Xscr\Abscr+\aXscr\dual\Abscr)\right)+\frac{\aXscr}{4}\dual\left(\nabo\hot\Bbscr+\frac{1}{2}(\Xscr\Abscr+\aXscr\dual\Abscr)\right).
\end{align*}
\end{thm}
\begin{proof}
    The expansions follow easily from the transport equations in the $e_4$ direction of these quantities and the existence of the limits in Proposition \ref{limitexist}, see Appendix \ref{secTaylor}.
\end{proof}
\section{Conformal quantities and regularity up to \texorpdfstring{$\II^+$}{}}\label{secscr}
In this section, we deduce the regularity of the conformal metric up to the boundary. In Sections \ref{ssecg}--\ref{sseckerrreg}, we work in a given $\KSAF$ spacetime, while in Section \ref{ssecgeneralreg} we consider the same problems in other contexts.
\subsection{Asymptotic behavior of metric components}\label{ssecg}
\begin{df}\label{dfGgGb}
We divide the metric components $G$ into the following two groups $G:=G_\nab\cup G_3$:
\begin{align}
\begin{split}\label{GgGb}
    G_\nab&:=\Big\{q\DD u,\;q\DD\cos\th,\;q\DD\Jp,\;|q|^2\DDb\c\Jk,\;q^2\DD\hot\Jk,\;|q|^2\DDb\c\Jk_\pm,\;q^2\DD\hot\Jk_\pm\Big\},\\
    G_3&:=\Big\{e_3(r),\;e_3(u),\;e_3(\cos\th),\;e_3(\Jp),\;\nab_3\Jk,\;\nab_3\Jk_\pm\Big\}.
\end{split}
\end{align}
\end{df}
\begin{rk}\label{Gcbeh}
Definition \ref{dfGgGb} separates the metric components $G$ into two groups according to their transport equations in the direction $e_4$. Denoting $\Gc:=G-G_{Kerr}$ the linearized metric components, we have from Definition \ref{dfGagGab} that
    \begin{align*}
        \Gc=r\Ga_{b,2}\cup r\Ga_{b,3}\cup r\Ga_{b,4}.
    \end{align*}
\end{rk}
\begin{prop}\label{metricexpansion}
We have the following Taylor expansion for metric components:
\begin{align}
    G=G^{[0]}+\frac{G^{[1]}}{r}+\eerr_G.\label{Ggexp}
\end{align}
where we denoted
\begin{align*}
    G^{[0]}:=\lim_{C_u,r\to\infty}G,\qquad G^{[1]}:=-\lim_{C_u,r\to\infty}r^2\nab_4G,
\end{align*}
and $\eerr_G$ is a remainder term satisfying
\begin{align}\label{eerrGcondition}
    \lim_{C_u,r\to\infty}r^2\nab_4(\eerr_G)=0,\qquad (r^2\nab_4)^2(\eerr_G)=\OO_0\left(r^{\frac{1}{2}-\dec}\right).
\end{align}
\end{prop}
\begin{proof}
We have from Remark \ref{Gcbeh} and Proposition \ref{GagGabdecay}
\begin{align*}
    |G_\nab|\les 1,\qquad\quad|G_3|\les 1.
\end{align*}
Moreover, we have from Propositions \ref{prop:e_4(xyz)} and \ref{prop:e_4J}
\begin{align}\label{nab4Gg}
    r^2\nab_4(G_\nab)=r^2Z\c(\Jp,q\Jk,q\Jk_\pm)+r^2\Gag\c G_\nab.
\end{align}
The boundedness of the R.H.S. of \eqref{nab4Gg} implies the existence of the following limit:
\begin{align*}
    G_\nab^{[0]}:=\lim_{C_u,r\to\infty}G_\nab.
\end{align*}
Moreover, plugging it into \eqref{nab4Gg}, we deduce the existence of the following limit:
\begin{align*}
    G_{\nab}^{[1]}=-\lim_{C_u,r\to\infty}r^2\nab_4(G_\nab).
\end{align*}
Next, differentiating \eqref{nab4Gg} by $r^2\nab_4$ and applying Theorem \ref{expansionexist}, we obtain
\begin{align*}
(r^2\nab_4)^2(G_\nab)&=r^2\nab_4(r^2Z)\c(\Jp,q\Jk,q\Jk_\pm)+r^2\nab_4(r^2\Gag)\c G_\nab+r^2\Gag\c r^2\nab_4(G_\nab)\\
&=\OO_0\left(r^{\frac{1}{2}-\dec}\right)\c(\Jp,q\Jk,q\Jk_\pm,G_\nab)+\OO_0(1)\\
&=\OO_0\left(r^{\frac{1}{2}-\dec}\right).
\end{align*}
Similarly, we have from Propositions \ref{prop:e_4(xyz)} and \ref{prop:e_4J}
\begin{equation}\label{nab4Gb}
    r^2\nab_4(G_3)=r^2\omb+r\eta\c G_\nab+\left(r\Gag+O(r^{-1})\right)\c (\Jp,q\Jk,q\Jk_\pm,G_\nab)=\OO_0(1).
\end{equation}
Hence, the following limits exist:
\begin{align*}
    G_3^{[0]}=\lim_{C_u,r\to\infty}G_3,\qquad G_3^{[1]}=-\lim_{C_u,r\to\infty}r^2\nab_4(G_3).
\end{align*}
Differentiating \eqref{nab4Gb} by $r^2\nab_4$ and applying Theorem \ref{expansionexist} and \eqref{nab4Gg}, we infer
\begin{align*}
    (r^2\nab_4)^2(G_3)&=r^2\nab_4(r^2\omb)+r^2\nab_4(r\eta)\c G_\nab+\left(r\eta+r\Gag+O(r^{-1})\right)\c r^2\nab_4(G_\nab)+\OO_0(1)\\
    &=\OO_0\left(r^{\frac{1}{2}-\dec}\right)+O(1)\c r^2\nab_4(G_\nab)+\OO_0(1)\\
    &=\OO_0\left(r^{\frac{1}{2}-\dec}\right).
\end{align*}
Hence, we deduce that for $G=G_\nab\cup G_3$, the following limits exist:
\begin{align*}
    G^{[0]}=\lim_{C_u,r\to\infty}G,\qquad G^{[1]}=-\lim_{C_u,r\to\infty}r^2\nab_4G.
\end{align*}
Moreover, we have
\begin{align*}
    (r^2\nab_4)^2(G)=(r^2\nab_4)^2\left(G-G^{[0]}-\frac{G^{[1]}}{r}\right)=\OO_0\left(r^{\frac{1}{2}-\dec}\right).
\end{align*}
Taking the error term as follows:
\begin{align*}
    \eerr_G:=G-G^{[0]}-\frac{G^{[1]}}{r},
\end{align*}
This concludes the proof of Proposition \ref{metricexpansion} with the error term $\eerr_G=G-G^{[0]}-\frac{G^{[1]}}{r}$.
\end{proof}
\subsection{Conformal compactification of \texorpdfstring{$\M$}{}}
\begin{df}\label{compactifiedspacetimes}
Let $(\M,\g)$ be a $\KSAF$ spacetime. We define the conformal metric $\gt$ on $\M$ as follows:
    \begin{align}\label{gtappears}
        \gt:=\vr^2\g,\qquad\quad\vr:=r^{-1}.
    \end{align}
    Here, $\vr$ is called the \emph{boundary defining function} and we define the \emph{future null infinity} as follows:
    \begin{align*}
    \II^+:=\{\vr=0\}.
    \end{align*}
    We then denote
    \begin{align*}
        \Mt:=\Mext\cup\II^+,
    \end{align*}
    and extend the domain of the metric $\gt$ by continuity to $\II^+$. The spacetime $(\Mt,\gt)$ is called the conformal compactification of $(\M,\g)$.
\end{df}
We recall from Proposition \ref{Mext-metric} that $\M$ is covered by three coordinates charts $\{(u,r,\th,\vphi):\frac{\pi}{4}<\th<\frac{3\pi}{4}\}$, $\{(u,r,x^1,x^2): 0\leq \th<\frac{\pi}{3}\}$, and $\{(u,r,x^1,x^2): \frac{2\pi}{3}<\th\leq\pi\}$. Thus, we have the following lemma that describes the conformal metric $\gt$ on $\M$.
\begin{lem}\label{conformalgt}
The inverse of the conformal metric $\gt$ on $\M$ has the following expressions:
\begin{itemize}
\item In the $(u,\vr,\th,\vphi)$ coordinates system for $\frac{\pi}{4}<\th<\frac{3\pi}{4}$, we have\footnote{Recall that $\Gc$ is defined in Remark \ref{Gcbeh}.}
    \begin{align*}
    \gt^{-1}=\gt^{-1}_{a,m}+\begin{pmatrix}
    1 & 1 & 1& 1 \\
    1 & \vr^2 & \vr & \vr \\
    1 & \vr & 1 & 1\\
    1 & \vr & 1 & 1
\end{pmatrix}(\sin\th)^{-4}\Gc.
\end{align*}
\item In the $(u,\vr,x^1,x^2)$ coordinates system for $0\leq \th<\frac{\pi}{3}$ or $\frac{2\pi}{3}<\th\leq \pi$, we have
\begin{align*}
    \gt^{-1}=\gt^{-1}_{a,m}+\begin{pmatrix}
    1 & 1 & 1& 1 \\
    1 & \vr^2 & \vr & \vr \\
    1 & \vr & 1 & 1\\
    1 & \vr & 1 & 1
\end{pmatrix}\Gc.
\end{align*}
\end{itemize}
\end{lem}
\begin{proof}
    We have from Proposition \ref{Mext-metric} and Remark \ref{Gcbeh} that, in the $(u,r,\th,\vphi)$ coordinates system for $\frac{\pi}{4}<\th<\frac{3\pi}{4}$, $\g^{-1}$ is given by:
\begin{align*}
    \g^{-1}=\g^{-1}_{a,m}+\begin{pmatrix}
    r^{-2} & 1 & r^{-2}& r^{-2} \\
    1 & 1 & r^{-1} & r^{-1} \\
    r^{-2} & r^{-1} & r^{-2} & r^{-2}\\
    r^{-2}& r^{-1} & r^{-2} & r^{-2}
\end{pmatrix}(\sin\th)^{-4}\Gc.
\end{align*}
Noticing that $\pr_\vr=-r^2\pr_r$, we have
\begin{align*}
    \g^{rr}&=r^4\g^{\vr\vr},\qquad \g^{ur}=-r^2\g^{\vr u},\qquad \g^{r\th}=-r^2\g^{\vr\th},\qquad \g^{r\vphi}=-r^2\g^{\vr\vphi}.
\end{align*}
Thus, we obtain that in the $(u,\vr,\th,\vphi)$ coordinates system for $\frac{\pi}{4}<\th<\frac{3\pi}{4}$
\begin{align*}
    \g^{-1}=\g^{-1}_{a,m}+\begin{pmatrix}
    r^{-2} & r^{-2} & r^{-2}& r^{-2} \\
    r^{-2} & r^{-4} & r^{-3} & r^{-3} \\
    r^{-2} & r^{-3} & r^{-2} & r^{-2}\\
    r^{-2}& r^{-3} & r^{-2} & r^{-2}
\end{pmatrix}(\sin\th)^{-4}\Gc,
\end{align*}
which implies from \eqref{gtappears}
\begin{align*}
    \gt^{-1}=r^2\g^{-1}=\gt^{-1}_{a,m}+\begin{pmatrix}
    1 & 1 & 1& 1 \\
    1 & \vr^2 & \vr & \vr \\
    1 & \vr & 1 & 1\\
    1 & \vr & 1 & 1
\end{pmatrix}(\sin\th)^{-4}\Gc,
\end{align*}
where $\gt_{a,m}$ denotes the conformal Kerr metric.
Similarly, we also have in the $(u,\vr,x^1,x^2)$ coordinates system for $0\leq \th<\frac{\pi}{3}$ or $\frac{2\pi}{3}<\th\leq \pi$,
\begin{align*}
    \gt^{-1}=\gt^{-1}_{a,m}+\begin{pmatrix}
    1 & 1 & 1& 1 \\
    1 & \vr^2 & \vr & \vr \\
    1 & \vr & 1 & 1\\
    1 & \vr & 1 & 1
\end{pmatrix}\Gc.
\end{align*}
This concludes the proof of Lemma \ref{conformalgt}.
\end{proof}
\subsection{Regularity up to \texorpdfstring{$\II^+$}{} in perturbations of Kerr}\label{sseckerrreg}
The following fundamental lemma plays an essential role in describing the regularity up to the future null infinity.
\begin{lem}\label{Holder}
Let $(\M,\g)$ be a $\KSAF$ spacetime. Then, the linearized metric components $\Gc$ have the following regularity up to the boundary:
\begin{align*}
   \Gc\in C^{1,\frac{1}{2}+\dec}(\Mt).
\end{align*}
\end{lem}
\begin{proof}
We have from Proposition \ref{metricexpansion}
\begin{align}\label{Ggbexpansion}
    G=G^{[0]}+G^{[1]}\vr+\eerr_G.
\end{align}
Considering $\eerr_G$ as a function of $\vr$, we have from \eqref{eerrGcondition}
\begin{align*}
    \lim_{\vr\to 0}\pr_\vr\eerr_G=0,\qquad\quad\pr_\vr^2\eerr_G=\OO_0(\vr^{-\frac{1}{2}+\dec}).
\end{align*}
Hence, we have from fundamental calculus that
\begin{align*}
    |\pr_\vr\eerr_G(\vr_1)-\pr_\vr\eerr_G(\vr_2)|\les \left|\int_{\vr_1}^{\vr_2}\pr_\vr^2\eerr_G(\vr')d\vr'\right|\les\left|\int_{\vr_1}^{\vr_2}{\vr'}^{-\frac{1}{2}+\dec}d\vr'\right|\les\left|\rho_2^{\frac{1}{2}+\dec}-\rho_1^{\frac{1}{2}+\dec}\right|.
\end{align*}
Notice that for any $\ga\in(0,1)$ and $x,y>0$, the following inequality holds:
\begin{align*}
    x^\ga+y^\ga\geq (x+y)^\ga.
\end{align*}
Thus, we obtain
\begin{align*}
    |\pr_\vr\eerr_G(\vr_1)-\pr_\vr\eerr_G(\vr_2)|\les |\vr_2-\vr_1|^{\frac{1}{2}+\dec}.
\end{align*}
Recalling the definition of H\"older space, we obtain
\begin{align*}
    \pr_\vr\eerr_G\in C^{0,\frac{1}{2}+\dec}(\Mt),
\end{align*}
which implies
\begin{align*}
    \eerr_G\in C^{1,\frac{1}{2}+\dec}(\Mt).
\end{align*}
Combining with \eqref{Ggbexpansion}, we deduce
\begin{align*}
    G\in C^{1,\frac{1}{2}+\dec}(\Mt).
\end{align*}
Recalling that the Kerr metric components $G_{Kerr}$ are smooth function of $\vr$, we infer that
\begin{align*}
    \Gc=G-G_{Kerr}\in C^{1,\frac{1}{2}+\dec}(\Mt).
\end{align*}
This concludes the proof of Lemma \ref{Holder}.
\end{proof}
We now clarify the regularity of $\gt$ in the manifold with boundary $\Mt$, which includes regularity in the interior region $\M$, along and toward the boundary $\II^+$.
\begin{thm}\label{regularityKerr}
The conformal metric $\gt$ defined on the compactified spacetime $\Mt$ satisfies:
    \begin{align*}
        \gt\in C^{\ks}(\M)\cap C^{1,\frac{1}{2}+\dec}(\Mt).
    \end{align*}
\end{thm}
\begin{proof}
We have from Lemma \ref{conformalgt} that in the coordinate chart $(u,\vr,\th,\vphi)$:
\begin{align}\label{gt-1}
    \gt^{-1}=(\gt_{a,m})^{-1}+\begin{pmatrix}
    1 & 1 & 1& 1 \\
    1 & \vr^2 & \vr & \vr \\
    1 & \vr & 1 & 1\\
    1 & \vr & 1 & 1
\end{pmatrix}(\sin\th)^{-4}\Gc,
\end{align}
    where $\gt_{a,m}$ denotes the conformal Kerr metric, which is a smooth function of $\vr$. Combining with Lemma \ref{Holder}, we deduce in the coordinates chart $(u,\vr,\th,\vphi)$, for $\frac{\pi}{4}<\th<\frac{3\pi}{4}$,
    \begin{align}\label{sinGc}
    \begin{pmatrix}
    1 & 1 & 1& 1 \\
    1 & \vr^2 & \vr & \vr \\
    1 & \vr & 1 & 1\\
    1 & \vr & 1 & 1
\end{pmatrix}(\sin\th)^{-4}\Gc\in C^{1,\frac{1}{2}+\dec}(\Mt).
    \end{align}
    Similarly, we also have in the coordinate chart $(u,\vr,x^1,x^2)$, for $0\leq\th<\frac{\pi}{3}$ or $\frac{2\pi}{3}<\th\leq \pi$,
    \begin{align}\label{cosGc}
    \begin{pmatrix}
    1 & 1 & 1& 1 \\
    1 & \vr^2 & \vr & \vr \\
    1 & \vr & 1 & 1\\
    1 & \vr & 1 & 1
\end{pmatrix}\Gc\in C^{1,\frac{1}{2}+\dec}(\Mt).
    \end{align}
    Combining \eqref{gt-1}--\eqref{cosGc}, we obtain 
    \begin{align*}
        \gt^{-1}=\gt^{-1}_{a,m}+\err\left(\gt^{-1}\right),\qquad\quad\err\left(\gt^{-1}\right)\in C^{1,\frac{1}{2}+\dec}(\Mt).
    \end{align*}
    Hence, we have\footnote{As a consequence of Proposition \ref{GagGabdecay} and Remark \ref{Gcbeh}, the matrices in \eqref{sinGc} and \eqref{cosGc} have size $\ep_0$.}
    \begin{align*}
        \gt&=\left(\gt^{-1}_{a,m}+\err\left(\gt^{-1}\right)\right)^{-1}\\
        &=\Big(\gt_{a,m}^{-1}\c\left(I+\gt_{a,m}\c\err\left(\gt^{-1}\right)\right)\Big)^{-1}\\
        &=\left(I+\gt_{a,m}\c\err\left(\gt^{-1}\right)\right)^{-1}\c\gt_{a,m}\\
        &=\sum_{k=0}^\infty(-1)^k\left(\gt_{a,m}\c\err\left(\gt^{-1}\right)\right)^k\c \gt_{a,m}\\
        &\in C^{1,\frac{1}{2}+\dec}(\Mt),
    \end{align*}
    where we used the fact that $\gt_{a,m}$ is smooth and that $C^{1,\frac{1}{2}+\dec}(\Mt)$ forms a Banach algebra. Moreover, as an immediate consequence of Proposition \ref{GagGabdecay}, we have
    \begin{align*}
        \gt\in C^{\ks}(\MM).
    \end{align*}
    This concludes the proof of Theorem \ref{regularityKerr}.
\end{proof}
\subsection{Regularity up to \texorpdfstring{$\II^+$}{} for other spacetimes}\label{ssecgeneralreg}
We now study the regularity of the conformal metric up to the future null infinity $\II^+$ in different contexts of stability problems in general relativity. To this end, we introduce the following definition.
\begin{df}\label{def6.3}
Let $s\in(1,\infty)\setminus\{2\mathbb{Z}+1\}$ and $q\in\mathbb{N}$. A perturbation of a particular solution $(\MM_{(0)} ,\g_{(0)})$\footnote{For example: Minkowski spacetime, Schwarzschild black hole or Kerr black hole.} of Einstein vacuum equations is called $(s,q)$--asymptotic if, the conformal compactified spacetime $(\Mt,\gt)$ in Definition \ref{compactifiedspacetimes}
of the perturbed spacetime $(\M,\g)$ satisfies:
\begin{align}
\begin{split}\label{sqass}
    \left|\pr_\vr^{\leq\left\lfloor\frac{s-1}{2}\right\rfloor}\dk^{\leq q}(\gt-\gt_{(0)})_{\mu\nu}\right|&\les 1,\qquad\qquad\quad\,\,\mu,\nu=u,\vr,\th,\vphi,\\
    \left|\pr_\vr^{\left\lfloor\frac{s+1}{2}\right\rfloor}\dk^{\leq q}(\gt-\gt_{(0)})_{\mu\nu}\right|&\les r^{1-\left\{\frac{s-1}{2}\right\}},\qquad \mu,\nu=u,\vr,\th,\vphi,
\end{split}
\end{align}
where $\lfloor x\rfloor$ denotes the largest integer less than or equal to $x$ and $\{x\}:=x-\lfloor x\rfloor$.
\end{df}
\begin{rk}
The decay parameter $s$ and the regularity parameter $q$ introduced here are related to the $(s,q)$--asymptotically flat initial data introduced in \cite{Shen22,Shen23,Shen24}.
\end{rk}
The following theorem explains the regularity of the conformal metric $\gt$ up to the future null infinity in stability problems.
\begin{thm}\label{nullinfinity}
Let $s\in(1,\infty)\setminus\{2\mathbb{Z}+1\}$ and $q\in\mathbb{N}$ satisfying
\begin{align}\label{qsre}
    q>\frac{s-1}{2}.
\end{align}
Let $(\M,\g)$ be a general $(s,q)$--asymptotic perturbation of a particular solution $(\M_{(0)},\g_{(0)})$ of Einstein vacuum equations. Then, we have
\begin{align*}
\gt\in C^q(\MM)\cap C^{\left\lfloor\frac{s-1}{2}\right\rfloor,\left\{\frac{s-1}{2}\right\}}(\Mt),
\end{align*}
where $(\Mt,\gt)$ denotes the conformal compactification of $(\M,\g)$, defined as in Definition \ref{compactifiedspacetimes}.
\end{thm}
\begin{proof}
We have from \eqref{sqass}
\begin{align*}
    \left|\pr_\vr^{\leq\left\lfloor\frac{s-1}{2}\right\rfloor}\dk^{\leq q}(\gt_{\mu\nu}-(\gt_{(0)})_{\mu\nu})\right|&\les 1,\\
    \left|\pr_\vr\left(\pr_\vr^{\left\lfloor\frac{s-1}{2}\right\rfloor}\dk^{\leq q}(\gt_{\mu\nu}-(\gt_{(0)})_{\mu\nu})\right)\right|&\les\rho^{\left\{\frac{s-1}{2}\right\}-1}.
\end{align*}
Proceeding as in Lemma \ref{Holder}, we deduce that
\begin{align*}
\pr_\vr^{\leq\left\lfloor\frac{s-1}{2}\right\rfloor}\dk^{\leq q}(\gt_{\mu\nu}-(\gt_{(0)})_{\mu\nu})\in C^{0,\left\{\frac{s-1}{2}\right\}}(\Mt).
\end{align*}
Combining with \eqref{qsre}, we obtain
\begin{align}\label{boundaryregularity}
    \gt_{\mu\nu}-(\gt_{(0)})_{\mu\nu}\in C^{\left\lfloor\frac{s-1}{2}\right\rfloor,\left\{\frac{s-1}{2}\right\}}(\Mt).
\end{align}
Moreover, we have from \eqref{sqass}
\begin{align*}
    |\dk^{\leq q}(\gt-\gt_{(0)})|\les 1,\quad \mbox{ in }\,\M,
\end{align*}
which implies
\begin{align}\label{interiorregularity}
    \gt-\gt_{(0)}\in C^q(\M).
\end{align}
Combining \eqref{boundaryregularity} and \eqref{interiorregularity}, we obtain
\begin{align}\label{gtcreg}
    \gt-\gt_{(0)}\in C^q(\M)\cap  C^{\left\lfloor\frac{s-1}{2}\right\rfloor,\left\{\frac{s-1}{2}\right\}}(\Mt).
\end{align}
Since $\g_{(0)}$ is a particular solution of the Einstein vacuum equations, such as Minkowski spacetime, Schwarzschild black hole or Kerr black hole, we have
\begin{align*}
    \gt_{(0)}\in C^\infty(\Mt).
\end{align*}
Combining with \eqref{gtcreg}, this concludes the proof of Theorem \ref{nullinfinity}.
\end{proof}
\begin{rk}\label{alltheregularity}
Various choices of the decay rate parameter $s$ have been made in the literature on stability problems. The table below (Table \ref{tab:decay-rates-reg}) summarizes these choices along with the corresponding regularity of the conformal metric $\gt$:

\begin{threeparttable}[ht]
    \centering
    \renewcommand{\arraystretch}{1.5}
    \caption{Decay rates $s$ in various stability results and the corresponding regularity of $\gt$.}
    \label{tab:decay-rates-reg} 
    \begin{tabular}{|c|c|c|}
        \hline
        Decay Rate $s$ & Stability Results in Vacuum Case & Regularity of $\gt$ \\
        \hline
        $1-2\de$\tnote{*} & \cite{Shen24} for exterior stability of Minkowski\tnote{**} & Discontinuous \\
        \hline
        $1+2\de$ & \cite{Shen23} for global stability of Minkowski & $C^{0,\de}(\Mt)$ \\
        \hline
        $2$ & \cite{Bieri} for global stability of Minkowski & $C^{0,\frac{1}{2}}(\Mt)$ \\
        \hline
        $3+2\de$ & \cite{lr1,lr2,HV,Hintz} for global or exterior stability of Minkowski & $C^{1,\de}(\Mt)$ \\
        \hline
        $3+2\de$ & \cite{Shen22,ShenKerr} for exterior stability of Minkowski and Kerr& $C^{1,\de}(\Mt)$ \\
        \hline
        $4$ & \cite{Ch-Kl,Kl-Ni,graf} for global or exterior stability of Minkowski & $C^{1,\frac{1}{2}}(\Mt)$ \\
        \hline
        $4+2\de$ & \cite{KS} for stability of Schwarzschild black holes & $C^{1,\frac{1}{2}+\de}(\Mt)$ \\
        \hline
        $4+2\de$ & \cite{KS:Kerr1,KS:Kerr2,KS:main,GKS,Shen} for stability of slowly rotating Kerr black holes & $C^{1,\frac{1}{2}+\de}(\Mt)$ \\
        \hline
        $6+2\de$ & \cite{DHRT} for stability of Schwarzschild black holes & $C^{2,\frac{1}{2}+\de}(\Mt)$ \\
        \hline
        $7+2\de$ & \cite{knpeeling,Caciotta} for exterior stability of Minkowski and Kerr & $C^{3,\de}(\Mt)$ \\
        \hline
    \end{tabular}
    \begin{tablenotes}
        \footnotesize
        \item[*] Here and the lines below, $\de$ denotes a fixed constant satisfying $0<\de\ll 1$.
        \item[**] In \cite{Shen24}, the decay rates of the initial data correspond to the case $s=1$, but the final decay estimates correspond to the case $s=1-2\de$.
    \end{tablenotes}
\end{threeparttable}
\end{rk}
\begin{rk}
The work of Christodoulou \cite{Ch02} strongly suggests that for most important physical applications, the conformal metric is not more regular than $C^{1,\ga}$ for $\ga\in(0,1)$.
\end{rk}
\section{Existence of Limiting GCM (LGCM) foliation}\label{secLG}
Throughout this section, we focus on the \emph{far region}:
\begin{align}\label{dffar}
    \far:=\M\cap\left\{r\geq\frac{\de_*}{\ep_0}u^{1+\dec}\right\},
\end{align}
where $\de_*>0$ is a small constant fixed in \cite{KS:main} and $0<\ep_0\ll 1$ measures the size of perturbation of initial data.\footnote{See Section 3.4.1 of \cite{KS:main} for more explanations about the choice of $\de_*$ and $\ep_0$.}
\subsection{Definition of Limiting GCM (LGCM) foliation}
\begin{df}\label{limitinggeodesic}
Let $(\M,\g)$ be a $\KSAF(a,m)$ spacetime endowed with an outgoing PG structure $\{S(u, r), (e_3,e_4,e_1,e_2)\}$. A sphere foliation $S(u', r')$ in $\far$, associated with a null frame $(e'_3,e'_4,e'_1,e'_2)$, is called a \emph{Limiting GCM (LGCM)} foliation (associated to the background PG structure $\{S(u, r), (e_3,e_4,e_1,e_2)\}$), if the following conditions are verified:
\begin{enumerate}
\item The transition functions $(f,\fb,\ovla)$ from  $(e_3,e_4,e_1,e_2)$ to $(e_3',e_4',e_1',e_2')$ verify the following estimate in $\far$:
\begin{align}\label{usedasbootstrap}
    \|(f,\fb,\ovla)\|_{\hk_{4}(S')}\les 1,
\end{align}
where for any sphere $S$, $\hk_k(S)$ denotes the $k$--th order Sobolev space on $S$:
\begin{align*}
    \|f\|_{\hk_k(S)}:=\sum_{i=0}^k\left\|(\dkb^S)^if\right\|_{L^2(S)},\qquad\quad \dkb^S:=r\nab^S.
\end{align*}
\item The spheres $S'(u',r')$ are the level sets of two scalar functions $u'$ and $r'$, which satisfy the following estimates in $\far$:
\begin{align}
\begin{split}\label{urconditions}
    e'_a(u')&=O(r^{-2}),\qquad e'_3(u')=2+O(r^{-1}),\qquad\, e'_4(u')=O(r^{-2}),\\
    e'_a(r')&=O(r^{-2}),\qquad\, e'_3(r')=-1+O(\ep_0),\qquad e'_4(r')=1+O(r^{-1}),\\
    |r'-r|&\les 1.
\end{split}
\end{align}
\item Letting $r'\to\infty$, the spheres $S'(u',r')$ induce a $S'(u')$--foliation on $\II^+$. Moreover, the spheres $(S'(u'),r^{-2}\g|_{S'(u')})$ are round spheres.
\item There exists a $r^{-1}$--approximate basis of $\ell=1$ modes $\Jpp$ on every $S'(u',r')$ in the sense of Definition \ref{jpdef}. Moreover, we have
\begin{align}\label{nab3Jpp}
    e'_4(\Jpp)=O(r^{-2}),\qquad\quad e_3'(\Jpp)=O(r^{-1}).
\end{align}
\item Related to the null frame $(e'_3,e'_4,e'_1,e'_2)$, the scalar functions $u'$ and $r'$, and the $r^{-1}$--approximate $\ell=1$ basis $\Jpp$ on $S'(u',r')$, all the limiting quantities defined in Propositions \ref{limitexist} and \ref{limitB} exist.
\item The following limiting intrinsic GCM conditions hold at $i^+$:\footnote{Here and the lines below, the modes are taken w.r.t. the $r^{-1}$--approximate $\ell=1$ basis $\Jpp$.}
\begin{align}
\begin{split}\label{GCMS2limit}
\lim_{u'\to\infty}\left(\Xscr',\Xbscr'\right)&=(0,4m),\qquad \lim_{u'\to\infty}(\Mscr')_{\ell\geq 2}=0,\qquad \lim_{u'\to\infty}\Bk'=0.
\end{split}
\end{align}
\item The basis of $\ell=1$ modes $\Jpp$ is calibrated at $i^+$ as follows:
\begin{equation}\label{angularmomentumcalibration}
    \lim_{u'\to\infty}(\Bkd')_{\ell=1,0}=2am,\qquad\quad\lim_{u'\to\infty}(\Bkd')_{\ell=1,\pm}=0.
\end{equation}
\item The following limiting integrability conditions hold on $\II^+$:
\begin{align}\label{integrablelimitsindef}
    \aXscr'=0,\qquad\quad\aXbscr'=0.
\end{align}
\item The following limiting incoming geodesic conditions hold on $\II^+$:
    \begin{equation}\label{Incominggeodesic}
    \Hscr'=0,\qquad\quad\Ybscr'=0,\qquad\quad\Wbscrone'=0.
    \end{equation}
\end{enumerate}
\end{df}
\begin{rk}
In \cite{KS:Kerr2}, a sphere $\S$ is called an \emph{intrinsic GCM sphere} if it satisfies:
\begin{align}\label{def:GCMC2}
\begin{split}
\trch^\S&=\frac{2}{r^\S},\qquad\qquad\;\;\trchb^\S=-\frac{2\Up^\S}{r^\S},\qquad\quad\;\,(\mu^\S)_{\ell\geq 2}=0,\\
\atrch^\S&=0,\qquad\qquad\atrchb^\S=0,\qquad\quad\;\,\left(\div^\S \b^\S\right)_{\ell=1}=0,
\end{split}
\end{align}
where $r^\S$ and $m^\S$ denote respectively the area radius and Hawking mass of $\S$ and 
$$
\Up^\S:=1-\frac{2m^\S}{r^\S}.
$$
Moreover, the modes in \eqref{def:GCMC2} are defined w.r.t. the canonical choice of $\ep$--approximate basis of $\ell=1$ modes $\JpS$. The conditions in \eqref{GCMS2limit} and \eqref{integrablelimitsindef} are the limiting version of \eqref{def:GCMC2}.
\end{rk}
\begin{rk}
We recall that an incoming geodesic foliation satisfies:
    \begin{align}\label{IG}
        \eta=\ze,\qquad\quad \xib=0,\qquad\quad \omb=0.
    \end{align}
However, the existence of $\Zscr$ implies that
$$
\lim_{C_u,r\to\infty}r\eta=\lim_{C_u,r\to\infty}r\ze=0.
$$
Thus, the condition \eqref{Incominggeodesic} is the limiting version of \eqref{IG}.
\end{rk}
The following theorem implies the existence of the Limiting GCM (LGCM) foliation in a given $\KSAF$ spacetime $(\M,\g)$.
\begin{thm}\label{LGCMconstruction}
Let $(\M,\g)$ be a $\KSAF(a, m)$ spacetime endowed with a background outgoing PG $S(u,r)$--foliation. Then, there exists a LGCM foliation, denoted by $\{S(u',r'), (e'_3, e'_4, e'_1, e'_2)\}$, associated to the outgoing PG $S(u,r)$--foliation in the sense of Definition \ref{limitinggeodesic}.
\end{thm}
\begin{rk}
The LGCM foliation $\{S(u',r'), (e'_3, e'_4, e'_1, e'_2)\}$ constructed in Theorem \ref{LGCMconstruction} is not strictly speaking an outgoing PG foliation. In fact, related to the null frame $(e'_3,e'_4,e'_1,e'_2)$, the quantities $\xi'$, $\om'$ and $\ze'+\etab'$ do not necessarily vanish, but have a higher order in $r^{-1}$.
\end{rk}
The rest of Section \ref{secLG} focuses on the proof of Theorem \ref{LGCMconstruction}. To this end, we first review the main results of intrinsic GCM spheres in Section \ref{sec:GCMpapersreview}. In Section \ref{ssecmostdifficult}, we construct an auxiliary sequence of incoming null cones $\{\Cb_n\}_{n\in\NNN}$, initiating from the intrinsic GCM spheres $\S_{*,n}$. In Section \ref{ssecendproof}, we prove Theorem \ref{LGCMconstruction} by showing, in an appropriate sense, that the sequence $\{\Cb_n\}_{n\in\NNN}$ converges to $\II^+$.
\subsection{Intrinsic GCM spheres}\label{sec:GCMpapersreview}
We review the following main results in \cite{KS:Kerr1,KS:Kerr2} that will be used in this section.
\begin{df}\label{definition:Deformations}
We say that $\S$ is a deformation of $S:=S(u,r)$ if there exist smooth scalar functions $U, R$ defined on $S$ and a map $\Psi:S\to\S$ verifying, on either coordinate chart $(\th^1,\th^2)$ of $S$,
\begin{align}\label{URfistappear}
\Psi(u,s,\th^1,\th^2)=\left(u+U(\th^1,\th^2),r+R(\th^1,\th^2),\th^1,\th^2\right),
\end{align}
where $U$ and $R$ are two scalar functions defined on $S$.
\end{df}
\begin{df}\label{definition:framechange.toadaptedframes}
Given a deformation $\Psi:S\to \S$ we say that a new frame $(e_3', e_4',  e_1', e_2')$ on $\S$, obtained from the standard frame $(e_3, e_4, e_1, e_2)$ by the general frame transformation \eqref{General-frametransformation} through the transition functions $(f,\fb,\la)$, is $\S$--adapted if the horizontal vectorfields $e'_1, e'_2$ are tangent to $\S$.
\end{df}
\begin{prp}\label{KS5.14}
Consider a deformation map $\Psi:S\to\S$ generated by the functions $U$, $R$ in the sense of \eqref{URfistappear}. Let $(e'_3,e'_4,e'_1,e'_2)$ be a $\S$--adapted frame on $\S$. Then, we have
    \begin{align*}
        \pr_{\th^a}U&=f+O(\ep_0)f+O(f,\fb)^2,\\
        \pr_{\th^a}R&=\frac{1}{2}(-\Up f+\fb)+O(\ep_0)f+O(f,\fb)^2,
    \end{align*}
    where $(f,\fb)$ are defined in Definition \ref{definition:framechange.toadaptedframes}.
\end{prp}
\begin{proof}
    See Proposition 5.14 in \cite{KS:Kerr1}.
\end{proof}
In the following, we state a simple version of Theorem 7.3 in \cite{KS:Kerr2}, which is the main result of that paper, concerning the construction of intrinsic GCM spheres. 
\begin{thm}[Construction of intrinsic GCM spheres \cite{KS:Kerr2}]\label{GCMS2}
Let $(\M,\g)$ be a $\KSAF(a,m)$ spacetime. Let $S:=S(u,r)$ denotes a leaf of the background outgoing PG foliation satisfying
\begin{align}\label{dominantc}
    r=\frac{\de_*}{\ep_0}u^{1+\dec}.
\end{align}
Then, there exist a unique deformation $\Psi:S\to\S$ which satisfies the following properties:
\begin{enumerate}
\item The following intrinsic GCM conditions hold on $\S$:\footnote{Recall that $r^{\S}$ and $m^{\S}$ are respectively the area radius and Hawking mass of ${\S}$.}
\begin{align}\label{GCMS2conditi}
\begin{split}
\trch^\S=\frac{2}{r^\S},\qquad \trchb^\S+\frac{2\Up^\S}{r^\S}=0,\qquad \big(\mu^\S\big)_{\ell\geq 2}=0, \qquad
\left(\div^\S \b^\S\right)_{\ell=1}=0,
\end{split}
\end{align}
where all the modes are defined w.r.t. $\JpS$, which is the canonical choice of $\ep_0 r^{-1}$--approximate basis of $\ell=1$ modes for $\S$ in the sense of Definition \ref{jpdef}.
\item The $\ep_0 r^{-1}$--approximate $\ell=1$ basis $\JpS$ is calibrated as follows:
\begin{align}\label{curlbS=0}
    \left(\curl^\S\b^\S\right)_{\ell=1,0}=\frac{2a^\S m^\S}{(r^\S)^5},\qquad\quad \left(\curl^\S\b^\S\right)_{\ell=1,\pm}=0,
\end{align}
where $a^\S$ is a constant on $\S$ which represents the angular momentum on $\S$, $m^\S$ is the Hawking mass of $\S$ and $r^\S$ is the area radius of $\S$. We also have the following estimates:
\begin{align}\label{aaSmmSrrS}
    |a^\S-a|\les\frac{\ep_0}{u^{\frac{1}{2}+\dec}},\qquad |m^\S-m|\les\frac{\ep_0}{u^{\frac{1}{2}+\dec}},\qquad |r^\S-r|\les\frac{\ep_0}{u^{\frac{1}{2}+\dec}}.
\end{align}
\item Denoting $(f,\fb,\la)$ the transition functions from the outgoing PG frame $(e_3,e_4,e_1,e_2)$ to the $\S$--adapted null frame $(e_3^\S,e_4^\S,e_1^\S,e_2^\S)$, we have
\begin{equation}
    \|(f,\fb,\ovla)\|_{\hk_{\ks+3}(\S)}\les 1.
\end{equation}
\item Let $U$ and $R$ be the scalar functions on $S$ defined by \eqref{URfistappear}. Then, we have
\begin{align*}
    \|(U,R)\|_{\hk_{\ks+4}(S)}\les r.
\end{align*}
\end{enumerate}
We call such $\S$ the intrinsic GCM sphere, which is a deformation of the background sphere $S=S(u,r)$ with $u$ and $r$ satisfying \eqref{dominantc}.
\end{thm}
\begin{proof}
    See Theorem 7.3 and Corollary 7.7 in \cite{KS:Kerr2} and Step 7 of Theorem M7 in \cite{KS:main}.
\end{proof}
\begin{cor}\label{GCMS2cor}
Let $(\M,\g)$ be a $\KSAF(a,m)$ spacetime and let $S:=S(u,r)$ be a leaf of the background outgoing PG foliation that satisfies \eqref{dominantc}. Then, we have the following estimates on the intrinsic GCM sphere $\S$, which is a deformation of the background sphere $S$:
    \begin{align}
        \begin{split}
            \left\|\trch^\S-\frac{2}{r}\right\|_{\hk_{\ks+2}(\S)}&\les\frac{\ep_0}{ru^{\frac{1}{2}+\dec}},\qquad\quad\left\|\trchb^\S+\frac{2}{r}\right\|_{\hk_{\ks+2}(\S)}\les\frac{\ep_0}{ru^{\frac{1}{2}+\dec}},\\
            (\mu^\S)_{\ell\geq 2}&=0,\qquad\qquad\qquad\qquad\qquad\;\left(\div^\S\b^\S\right)_{\ell=1}=0,\\
            \left|\left(\curl^\S\b^\S\right)_{\ell=1,0}-\frac{2am}{r^5}\right|&\les\frac{\ep_0}{r^5u^{\frac{1}{2}+\dec}},\qquad\qquad\;\;\; \left|\left(\curl^\S\b^\S\right)_{\ell=1,\pm}\right|\les\ep_0.
        \end{split}
    \end{align}
Moreover, we have the following identities for $\JpS$ on $\S$:
\begin{align}
\begin{split}\label{canonicalonGCMS2}
    \De^\S\JpS+\frac{2}{(r^\S)^2}\JpS&=O\left(\frac{\ep_0}{r^3}\right),\\
    \frac{1}{|\S|}\int_\S \JpS J^{(q,\S)}&=\frac{4\pi}{3}\de_{pq}+O\left(\frac{\ep_0}{r}\right),\\
    \frac{1}{|\S|}\int_\S \JpS &=0.
\end{split}
\end{align}
\end{cor}
\begin{proof}
    The proof follows directly from \eqref{GCMS2conditi}--\eqref{aaSmmSrrS}.
\end{proof}
\subsection{Construction of Limiting GCM (LGCM) foliation}\label{ssecmostdifficult}
\subsubsection{Construction of a sequence of incoming null cones }\label{ssecCbn}
In this section, we construct a sequence of incoming null cones $\Cb_n$ in $\far$, endowed with geodesic foliations.
\begin{df}\label{constructionCbndf}
Let $(\M,\g)$ be a $\KSAF$ spacetime endowed with an outgoing PG $S(u,r)$--foliation. Let $u_n\to\infty$ be a sequence, and define\footnote{Recall that $\de_*$ and $\ep_0$ are introduced in \eqref{dffar}.}
\begin{align}\label{sequencedominant}
    r_n:=\frac{\de_*}{\ep_0}u_n^{1+\dec}.
\end{align}
We proceed as follows to construct a sequence of incoming null cones $\Cb_n$:
\begin{enumerate}
    \item For each fixed $n\in\NNN$, we denote $S_{*,n}:=S(u_n,r_n)$. Applying Theorem \ref{GCMS2} and \eqref{sequencedominant}, we construct from $S_{*,n}$, a unique intrinsic GCM sphere $\S_{*,n}$, with an adapted null frame $\left(e_3^{\S_{*,n}},e_4^{\S_{*,n}},e_1^{\S_{*,n}},e_2^{\S_{*,n}}\right)$ and a canonical choice of $\ep_0 r^{-1}$--approximate basis of $\ell=1$ modes $J^{(p,\S_{*,n})}$.
    \item For each fixed $n\in\NNN$, emanating from $\S_{*,n}$, we construct a unique incoming null cone $\Cb_n$ contained in $\far$ and extended to the initial layer region $\LL_0$, endowed with an incoming geodesic foliation:
\begin{align*}
    \Cb_n=\bigcup_{s}\,\widetilde{\S}_n(s),
\end{align*}
where $s$ is an affine parameter of the incoming geodesic foliation. See Figure \ref{fig:Cb_n} for a description of $\Cb_n$, $\S_{*,n}$ and $\widetilde{\S}_n(s)$.
\begin{figure}
    \centering
\begin{tikzpicture}[scale=1.3]
\draw[red, thin] (-3.05,-2.02375) -- (-1.05,1.92625) ;
\draw[red, thin] (0.9,2.2225) -- (2.78,-1.5);
\node[red, right] at (2.1,0) {$\Cb_n$};
\node[red, right] at (1.55,1.1) {$\widetilde{\S}_n(s)$};
\foreach \y in { 1.99,   0.5,   -1.5} {
    \draw[black] (0,\y) ellipse ({0.5*(4-\y)} and {0.15*(4-\y)});
}
\node[left] at (-1.1,2) {$S_{*,n}=S(u_n,r_n)$};
\node[left] at (-1.8,0.6) {$S(u_s,r_s)$};
\draw[red, line width=1.5pt, rotate around={10:(-0.5*2+1,1.3*2+1)}] (-0.35,2.06) ellipse ({0.49*(4-2)} and {0.18*(4-2)}) node[right] {$\S_{*,n}$};
\draw[red, rotate around={10:(-0.5*0.5+1,1.3*0.5+1)}] (-0.31,0.75) ellipse ({0.49*(4-0.5)} and {0.18*(4-0.5)});
\draw[red, rotate around={10:(0.73,-0.95)}] (-0.33,-1.45) ellipse (2.859 and 0.72);
\fill[black] (0,1.67) circle (2pt) node[below right] {$p_{*,n}$};
\fill[red] (0,-0.05) circle (2pt) node[below right] {$p_s$};
\fill[black] (0,-2.33) circle (2pt) node[below right]{\scriptsize{South Pole}};
\end{tikzpicture}
    \caption{\small Construction of the incoming null cone $\Cb_n$. $\S_{*,n}$ is the intrinsic GCM sphere deformed from a background sphere $S_{*,n}=S(u_n,r_n)$ and $\Cb_n$ is the incoming null cone emanating from $\S_{*,n}$. $\widetilde{\S}_n(s)$ is a leaf of the geodesic foliation of $\Cb_n$ and it coincides with a background sphere $S(u_s,r_s)$ at south pole $p_s$. The spheres $\widetilde{\S}_n(s)$ on $\Cb_n$ will be relabeled by the background function $u$ in \eqref{defus}.}
    \label{fig:Cb_n}
\end{figure}
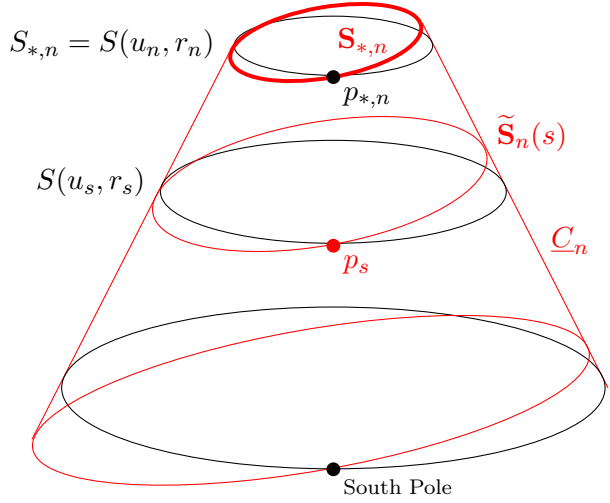
\item For any fixed $n\in\NNN$ and $s$, we denote $p_s$ the South Pole of $\widetilde{\S}_n(s)$\footnote{The South Pole of $\widetilde{\S}_n(s)$ can be defined as follows. We first define the South Pole of $\S_{*,n}$ by taking the common point of $\S_{*,n}$ and $S(u_n,r_n)$, denoted by $p_{*,n}$. We then denote $\ga_n$ the geodesic emanating from $p_{*,n}$ in the direction of $e_3^{\S_{*,n}}$ along $\Cb_n$. Then, we define $p_s:=\ga_n\cap\widetilde{\S}_n(s)$.}. Then, there is a unique background sphere $S(u_s,r_s)$, such that $p_s\in S(u_s,r_s)$. We define a function $\Phi_n$ as follows:
\begin{align*}
    r_s:=\Phi_n(u_s).
\end{align*}
Notice that we have by definition and \eqref{sequencedominant}
\begin{align*}
    \Phi_n(u_n)=r_n=\frac{\de_*}{\ep_0}u_n^{1+\dec}.
\end{align*}
\item Let $\Psi_n$ be the deformation map from $S(u_s,\Phi_n(u_s))$ to $\widetilde{\S}_n(s)$. Then, $\Psi_n$ is defined by two scalar functions $U_n$ and $R_n$ as follows: 
\begin{align}
\begin{split}\label{PsinUnRnrelated}
\Psi_n:S(u_s,\Phi_n(u_s))&\to\widetilde{\S}_n(s),\\(u_s,\Phi_n(u_s),\th^1,\th^2)&\to\left(u_s+U_n(u_s,\th^1,\th^2),\Phi_n(u_s)+R_n(u_s,\th^1,\th^2),\th^1,\th^2\right).
\end{split}
\end{align}
Since $\widetilde{\S}_n(s)=\Psi_n(S(u_s,\Phi_n(u_s))$, we have
\begin{align*}
    \Cb_n=\bigcup_{1\leq u_s\leq u_n}\widetilde{\S}_n(s)=\bigcup_{1\leq u_s\leq u_n}\Psi_n(S(u_s,\Phi_n(u_s))).
\end{align*}
\item In the sequel, we use the parameter $u_s$ instead of $s$ to label all the leaves of the geodesic foliation. More precisely, we define
\begin{align}\label{defus}
    \S_n(u_s):=\Psi_n(S(u_s,\Phi_n(u_s)))=\widetilde{\S}_n(s).
\end{align}
Hence, we have\footnote{Notice that we have by construction that $\S_n(u_n)=\S_{*,n}$ and that $\S_n(1)$ is contained in the initial layer region $\LL_0$, see Figure \ref{fig:triangle}. We also remark that the background function $u$ may not be constant on $\S_n(u)$.}
\begin{align*}
    \Cb_n=\bigcup_{1\leq u\leq u_n}\S_n(u)=\bigcup_{1\leq u\leq u_n}\Psi_n(u,\Phi_n(u)).
\end{align*}
\end{enumerate}
\end{df}
See Figure \ref{fig:triangle} for a geometric illustration of the sequence of incoming null cones $\Cb_n$ constructed in Definition \ref{constructionCbndf}.
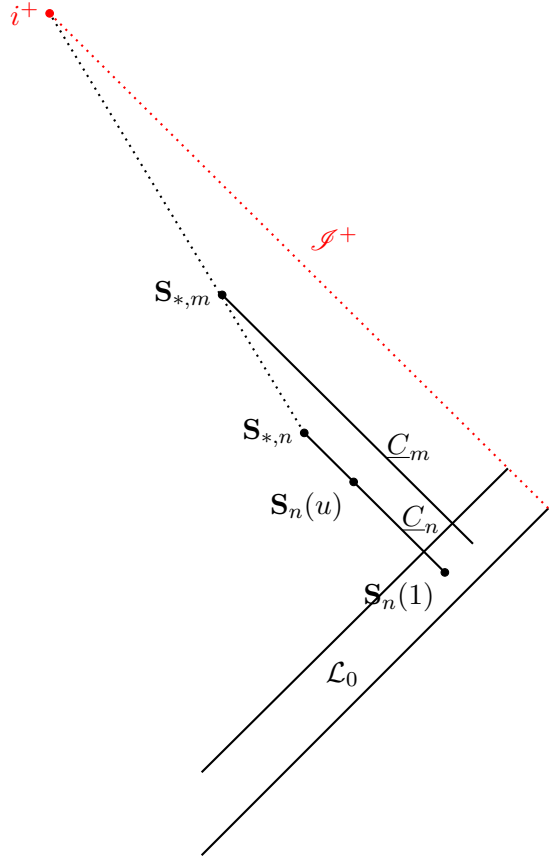
\begin{figure}[H]
    \centering
    \begin{tikzpicture}[scale=1]
    \coordinate (A) at (0, -0.3); 
    \coordinate (B) at (-2, 10.8623); 
    \coordinate (C) at (4.5623,4.3); 
    \draw[thick] (C) -- (A);
    \filldraw[red] (B) circle[radius=0.05] node[left] {$i^+$};
    \draw[red, dotted, thick] (B) -- (C) node[midway, above right] {$\II^+$};
    \coordinate (D) at (0, 0.8); 
    \coordinate (E) at (4.03, 4.83); 
    \draw[thick] (D) -- (E);
    \node[below right] at (1.5, 2.4) {$\LL_0$};
    \coordinate (F) at (1.35, 5.3);  
    \coordinate (G) at (3.2, 3.45); 
    \filldraw(G) circle[radius=0.05] node[below left]  {$\mathbf \S_{n}(1)$};
    \draw[thick] (F) -- (G);
    \filldraw (F) circle[radius=0.05] node[left] {$\mathbf S_{*,n}$};
    \coordinate (H) at (2, 4.65);
    \filldraw (H) circle[radius=0.05] node[below left] {$\S_n(u)$};
    \node[below right] at (2.5, 4.46) {$\Cb_n$};
    \draw[dotted, thick] (F) -- (B);
    \coordinate (I) at (0.27,7.13);
    \coordinate (J) at (3.57,3.83);
    \filldraw (I) circle[radius=0.05] node[left] {$\mathbf S_{*,m}$};
    \draw[thick] (I) -- (J);
    \node[above right] at (2.3, 4.8) {$\Cb_m$};
\end{tikzpicture}
\caption{\small Illustration of $\Cb_n$ and $\Cb_m$ with $n<m$. For any fixed $n\in\NNN$, $\Cb_n$ is foliated by the spheres $\S_n(u)$ for $1\leq u\leq u_n$. The sequence of the intrinsic spheres $\S_{*,n}$ converges to the timelike infinity $i^+$ while the sequence of the incoming null cones $\Cb_n$ approach the future null infinity $\II^+$ when $n\to\infty$.}
\label{fig:triangle}
\end{figure}
\subsubsection{Estimates for  the sequence of incoming null cones}
The goal of this section is to prove Proposition \ref{limittranstion} stated below, in which we obtain the estimates on $\Cb_n$.
\begin{prop}\label{limittranstion}
Let $(\M,\g)$ be a $\KSAF$ spacetime endowed with an outgoing PG $S(u,r)$--foliation. Let $\Cb_n$ be the sequence of incoming null cones constructed in Definition \ref{constructionCbndf}. For any fixed $n\in\NNN$ and $u\in[1,u_n]$, we introduce the following shorthand notation:
\begin{align}\label{shorthandS}
    \S:=\S_n(u),\qquad S:=S(u,\Phi_n(u)).
\end{align}
Then, the following properties hold on $\Cb_n$:
\begin{enumerate}
\item  Let $(f_n,\fb_n,\la_n)$ be the transition functions from the background null frame $(e_3,e_4,e_1,e_2)$ to the adapted null frame $(e_3^\S,e_4^\S,e_1^\S,e_2^\S)$ on $\S$.\footnote{For each fixed $n\in\NNN$ and $s\in[1,u_n]$, $f_n$ and $\fb_n$ are $1$--forms defined on $\S_n(u)\subseteq\Cb_n$. On the other hand, $\la_n$ is a scalar function defined in $\Cb_n$.} Then, we have
\begin{align}
\begin{split}\label{estfnfbnovlan}
\|(f_n,\fb_n,\ovla_n)\|_{\hk_{\ks+3}(\S)}\les 1,\qquad \|\nab_3^\S(f_n,\fb_n,\ovla_n)\|_{\hk_{\ks+2}(\S)}\les\frac{\ep_0}{u^{1+\dec}}.
\end{split}
\end{align}
\item Let $U_n$ and $R_n$ be the scalar functions defined in \eqref{PsinUnRnrelated}. Then, we have
\begin{align}
\begin{split}\label{estUnRn}
\left\|(U_n,R_n)\right\|_{\hk_{\ks+4}(S)}\les r,\qquad \left\|\nab_3(U_n,R_n)\right\|_{\hk_{\ks+3}(S)}\les \frac{\ep_0 r}{u^{1+\dec}}.
\end{split}
\end{align}
\item There exists an $\ep_0 r^{-1}$--approximate basis of $\ell=1$ modes $\JpS$ on spheres $\S$ along $\Cb_n$ satisfying
    \begin{align}\label{nab3JpS}
        e_3^\S(\JpS)=0\quad\mbox{ on }\,\Cb_n,\qquad\JpS=J^{(p,\S_{*,n})}\quad \mbox{ on }\,\S_{*,n},
    \end{align}
where $J^{(p,\S_{*,n})}$ denotes the canonical choice of $\ep_0r^{-1}$--approximate basis of $\ell=1$ modes on $\S_{*,n}$.
\item With respect to the $\ell=1$ basis $\JpS$, the following estimates hold on $\Cb_n$:
\begin{align}
    \begin{split}\label{GCMquantitiesdecayonSin}
    \left|(\div^\S\b^\S)_{\ell=1}\right|&\les\frac{\ep_0}{r^5u^{\dec}},\qquad\quad\;\;\left\|(\mu^\S)_{\ell\geq 2}\right\|_{\hk_{\ks+1}(\S_n)}\les\frac{\ep_0}{r^3u^{\dec}},\\
    \left\|\trch^\S-\frac{2}{r}\right\|_{\hk_{\ks+2}(\S_n)}&\les\frac{\ep_0}{r^2u^{\dec}},\qquad\left\|\trchb^\S+\frac{2\Up}{r}\right\|_{\hk_{\ks+2}(\S_n)}\les\frac{\ep_0}{r^2u^{\dec}},\\
    \left|(\curl^\S\b^\S)_{\ell=1,0}-\frac{2am}{r^5}\right|&\les\frac{\ep_0}{r^5u^{\dec}},\qquad\qquad\quad\;\;\,\left|(\curl^\S\b^\S)_{\ell=1,\pm}\right|\les\frac{\ep_0}{r^5u^{\dec}}.
    \end{split}
\end{align}
\end{enumerate}
\end{prop}
\begin{proof}
As an immediate consequence of Theorem \ref{GCMS2}, the following estimates hold respectively on $\S_{*,n}$ and $S_{*,n}=S(u_n,r_n)$:
\begin{align}\label{fnandsoonS*n}
    \|(f_n,\fb_n,\ovla_n)\|_{\hk_{\ks+3}(\S_{*,n})}\les 1,\qquad\quad\|(U_n,R_n)\|_{\hk_{\ks+4}(S_{*,n})}\les r.
\end{align}
We have from Proposition \ref{Riccitransfer} that on $\Cb_n$:\footnote{Here and below, we use the shorthand notation introduced in \eqref{shorthandS}.}
\begin{align*}
\la^2\xib^\S&=\xib+\frac{1}{2}\la_n\nab_3^\S\fb_n+O(r^{-2}),\\
\eta^\S-\ze^\S&=\eta-\ze+\frac{1}{2}\la_n\nab_3^\S f_n+O(r^{-2}),\\
\la\omb^\S&=\omb+\frac{1}{2}\la_n\nab_3^\S\ovla_n+O(r^{-2}).
\end{align*}
Applying the incoming geodesic conditions $\xib^\S=0, \omb^\S=0, \eta^\S=\ze^\S$ and the decay properties of $\xi$, $\omb$, $\eta$ and $\ze$ in Proposition \ref{GagGabdecay}, we obtain\footnote{We also apply the fact that $\Cb_n$ is contained in $\far$.}
\begin{align}\label{nab3fnfbnovlan}
    \|\nab_3^\S(f_n,\fb_n,\ovla_n)\|_{\hk_{\ks+2}(\S)}\les\frac{\ep_0}{u^{1+\dec}}.
\end{align}
Combining with \eqref{fnandsoonS*n}, we infer
\begin{align}\label{fnfbnovlancontrolled}
    \sup_{1\leq u\leq u_n}\|(f_n,\fb_n,\ovla_n)\|_{\hk_{\ks+2}(\S)}\les 1,
\end{align}
which implies \eqref{estfnfbnovlan}. Combining with Proposition \ref{KS5.14}, we easily deduce
\begin{align}
\begin{split}\label{UnRncontrolled}
    \sup_{1\leq u\leq u_n}\|(U_n,R_n)\|_{\hk_{\ks+3}(S)}&\les r,\\
    \sup_{1\leq u\leq u_n}\left\|\nab_3(U_n,R_n)\right\|_{\hk_{\ks+2}(S)}&\les \frac{\ep_0r}{u^{1+\dec}}.
\end{split}
\end{align}
This concludes the proof of \eqref{estUnRn}.\\ \\
We then define $\JpS$ on $\Cb_n$ by the following transport equation:
\begin{align*}
    e_3^\S(J^{(p,\S)})=0,\quad\mbox{ on }\,\Cb_n,\qquad J^{(p,\S)}=J^{(p,\S_{*,n})}\quad \mbox{ on }\,\S_{*,n}.
\end{align*}
Integrating it along $\Cb_n$ and combining with \eqref{canonicalonGCMS2}, we obtain that the following identities for $\JpS$ hold on $\Cb_n$:
\begin{align*}
    \De^\S\JpS+\frac{2}{r^2}\JpS&=O\left(\frac{\ep_0}{r^3}\right),\\
    \frac{1}{|\S|}\int_\S \JpS J^{(q,\S)}&=\frac{4\pi}{3}\de_{pq}+O\left(\frac{\ep_0}{r}\right),\\
    \frac{1}{|\S|}\int_\S \JpS &=O\left(\frac{\ep_0}{r}\right),
\end{align*}
which implies that $\JpS$ is an $\ep_0 r^{-1}$--approximate basis of $\ell=1$ modes on $\Cb_n$, this concludes \eqref{nab3JpS}. Next, we have from Proposition \ref{prop-nullstrandBianchi:complex:outgoing}\footnote{The bounds on the R.H.S. follow easily from Propositions \ref{Riccitransfer}, \ref{Curvaturetransfer}, \eqref{fnfbnovlancontrolled}, \eqref{UnRncontrolled} and the fact that $\Cb_n$ is contained in $\far$.}
\begin{align*}
\nab_3^\S\left(\trch^\S-\frac{2}{r},\trchb^\S+\frac{2\Up}{r}\right)&=O\left(\frac{\ep_0}{r^2u^{1+\dec}}\right),\\
\nab_3^\S(\div^\S\b^\S,\curl^\S\b^\S)&=O\left(\frac{\ep_0}{r^5u^{1+\dec}}\right),\\
\nab_3^\S(\mu^\S)&=O\left(\frac{\ep_0}{r^3u^{1+\dec}}\right).
\end{align*}
Integrating them along $\Cb_n$ and applying Theorem \ref{GCMS2} and Lemma \ref{lemma:comm-gen}, we easily obtain \eqref{GCMquantitiesdecayonSin}. This concludes the proof of Proposition \ref{limittranstion}.
\end{proof}
\subsubsection{Limiting process}\label{sec:limitingprocess}
In this section, we study the limits of the quantities defined on the sequence of incoming null cones $\Cb_n$, constructed in Definition \ref{constructionCbndf}. To this end, we first recall the following version of Arzel\`a-Ascoli Lemma.
\begin{lem}[Arzel\`a-Ascoli]\label{AA}
Let $K$ be a compact metric space and let $\phi_n\in C(K)$ be a sequence of functions defined on $K$. Assume that
\begin{enumerate}
    \item $\phi_n$ is uniformly bounded, that is
    \begin{align*}
        \sup_{x\in K}|\phi_n(x)|\leq M,\qquad\forall\, n\in\NNN.
    \end{align*}
    where $M$ is a constant independent of $n$.
    \item $\phi_n$ is equicontinuous, that is
    \begin{align*}
        \forall\ep>0,\;\exists \de>0\;\mbox{ such that for all }\;d(x_1,x_2)<\de \Rightarrow |\phi_n(x_1)-\phi_n(x_2)|<\ep,\quad\forall\, n\in\NNN.
    \end{align*}
\end{enumerate}
    Then, there exists a subsequence of $\phi_n$, which converges uniformly to a function $\phi\in C(K)$.
\end{lem}
\begin{proof}
    See for example Theorem 4.25 in \cite{Brezis}.
\end{proof}
\begin{prop}\label{AAapply}
Let $(\M,\g)$ be a $\KSAF$ spacetime endowed with an outgoing PG $S(u,r)$--foliation and let $\Cb_n$ be the sequence of incoming null cones constructed in Definition \ref{constructionCbndf}. We define the following sequence:
\begin{align}\label{dfQQn}
    \QQ_n:=\left(U_n,R_n,rf_n,r\fb_n,r\ovla_n,J^{(p,\S_n(u))}\right),
\end{align}
where all quantities are defined in Proposition \ref{limittranstion}. Then, there exists a subsequence of $\QQ_n$, still denoted by $\QQ_n$, s.t. the following limit exists:\footnote{Notice that the convergence of $\QQ_n$ is uniform on $[1,A]\times\SSS^2$ for any fixed $A\geq 1$.}
\begin{align}\label{eq:limitsFFbLaJp}
\lim_{n\to\infty}\QQ_n=\left(U_\infty,R_\infty,F,\Fb,\La,\Jpp\right),
\end{align}
which can be viewed as functions defined on $\II^+$.
\end{prop}
\begin{rk}
In order to describe the limit of the sequence of incoming null cones $\Cb_n\to\II^+$, we prove in Proposition \ref{AAapply}, the convergence of the quantities in \eqref{dfQQn}. Using their limits, we construct in Step 1 of Section \ref{ssecendproof}, a new $S'$--foliation of $\far$. This allows us to avoid the convergence of a sequence of submanifolds in $\far$.
\end{rk}
\begin{proof}[Proof of Proposition \ref{AAapply}] We consider the components of $\QQ_n$ as functions on $[1,u_n]\times\SSS^2$. In view of \eqref{estfnfbnovlan}--\eqref{nab3JpS}, the sequence $\{\QQ_n\}$ is uniformly bounded and equicontinuous on $[1,u_n]\times\SSS^2$. Let $m_1\geq 4$, then $\QQ_n$ is well defined in $[1,m_1]\times\SSS^2$ for $n$ large enough. Applying Lemma \ref{AA} to functions $\QQ_n$ on the compact set $[1,m_1]\times\SSS^2$, we deduce that there exists a subsequence, denoted by $\{\QQ^{(m_1)}_n\}$, which converges uniformly on $[1,m_1]\times\SSS^2$. Next, denoting $m_2=2m_1$, the sequence $\{\QQ^{(m_1)}_n\}$ is well defined on $[1,m_2]\times\SSS^2$ for $n$ large enough. We have from Lemma \ref{AA} that there exists a subsequence, denoted by $\{\QQ^{(m_2)}_n\}$, which converges uniformly on $[1,m_2]\times\SSS^2$. Proceeding as above, for any $m_k=2^{k-1}m_1$, we deduce the existence of a subsequence of $\QQ_n$, denoted by $\{\QQ^{(m_k)}_n\}$, which converges uniformly on $[1,m_k]\times\SSS^2$. Using the standard diagonal argument, we deduce that the subsequence $\{\QQ^{(m_n)}_n\}$ converges to a well defined group of functions on $\RRR\times\SSS^2$. We then denote
\begin{align*}
\left(U_\infty,R_\infty,F,\Fb,\La,\Jpp\right):=\lim_{n\to\infty}\QQ^{(m_n)}_n=\lim_{n\to\infty}\left(U_n,R_n,rf_n,r\fb_n,r\ovla_n,J^{(p,\S_n(u))}\right).
\end{align*}
This concludes the proof of Proposition \ref{AAapply}.
\end{proof}
\subsection{Proof of Theorem \ref{LGCMconstruction}}\label{ssecendproof}
We are now ready to prove Theorem \ref{LGCMconstruction}. The proof is divided into 6 steps as follows:
\begin{itemize}
    \item In Step 1, we construct $S'$--foliation using the functions $(U_\infty,R_\infty,F,\Fb,\La,\Jpp)$ on $\II^+$, which are obtained in Proposition \ref{AAapply}.
    \item In Step 2, we use the null transformation formulae in Propositions \ref{Riccitransfer} and \ref{Curvaturetransfer} to show the existence of the limiting quantities in the new $S'$--foliation, constructed in Step 1.
    \item In Step 3, we apply the null transformation formulae on $\Cb_n$. Taking $n\to\infty$ and applying Proposition \ref{AAapply}, we derive several identities of $(U_\infty,R_\infty,F,\Fb,\La,\Jpp)$ in relation to the limiting quantities obtained in Propositions \ref{limitexist} and \ref{limitB}.
    \item In Step 4, we use the identities for $(U_\infty,R_\infty,F,\Fb,\La,\Jpp)$ derived in Step 3 to prove that the LGCM conditions \eqref{GCMS2limit}--\eqref{Incominggeodesic} hold in the $S'$--foliation.
    \item In Step 5, we derive the essential properties of the new parameters $u'$ and $r'$ introduced in Step 1 through the relation \eqref{ovbest}. Moreover, $u'$ will be reparameterized in \eqref{eq:reparam}.
    \item In Step 6, we show that $\Jpp$ is a $r^{-1}$--approximate basis of $\ell=1$ modes on the spheres $S'(u',r')$ in the sense of Definition \ref{jpdef}.
    \item Combining the above results, we deduce that all the desired properties of $S'$--foliation stated in Definition \ref{limitinggeodesic} are verified. This concludes the proof  of the fact that the $S'$--foliation is a LGCM foliation of $\far$.
\end{itemize}
{\bf Step 1.} Construction of $S'$--foliation.\\ \\
Let $(U_\infty,R_\infty,F,\Fb,\La,\Jpp)$ be the functions on $\II^+$ obtained in Proposition \ref{AAapply}, we define a $S'$--foliation of $\far$ as follows:
\begin{enumerate}
\item For any fixed background sphere $S(u_0,r_0)$, we define a new sphere $S'$ by the following deformation map $\Psi$:
\begin{align}
\begin{split}\label{S'constructed}
\Psi:S(u_0,r_0)&\rightarrow S',\\
\left(u_0,r_0,\th^1,\th^2\right)&\rightarrow\left(u_0+U_\infty(u_0,\th^1,\th^2),r_0+R_\infty(u_0,\th^1,\th^2),\th^1,\th^2\right),
\end{split}
\end{align}
where $(\th^1,\th^2)$ can be any of the three coordinates charts. We define two scalar functions $u'$ and $r'$ on $\MM$ such that they are constants on each such $S'$ 
\begin{align}\label{u'r'df}
    u'\big|_{S'}=u_0,\qquad\quad r'\big|_{S'}=r_0.
\end{align}
Notice that we have from \eqref{u'r'df}
\begin{align*}
    u'(p)=u_0,\qquad\quad r'(p)=r_0, 
\end{align*}
with
\begin{align*}
    p=(u_0+U_\infty(u_0,\th^1,\th^2),r_0+R_\infty(u_0,\th^1,\th^2),\th^1,\th^2)\in\Psi(S(u_0,r_0)).
\end{align*}
Hence, we have, at any $p\in \Psi(S(u_0,r_0))$,
\begin{align*}
    u'(p)&=u_0=u(p)-U_\infty(u_0,\th^1,\th^2)=u(p)-U_\infty(\Psi^{-1}(p)),\\
    r'(p)&=r_0=r(p)-R_\infty(u_0,\th^1,\th^2)=r(p)-R_\infty(\Psi^{-1}(p)).
\end{align*}
We define the following scalar functions on $S'$:
\begin{align}\label{dfURdiese}
    U^\#:=U_\infty\circ\Psi^{-1},\qquad R^\#:=R_\infty\circ \Psi^{-1}.
\end{align}
Viewing $\Psi$ as globally defined on $\far$, and similarly $U^\#$ and $R^\#$ as globally defined, we have in the spacetime $\M$:  
\begin{align}\label{ovbest}  
    u' = u - U^\#,\qquad r' = r - R^\#.  
\end{align}
By construction, $S'$ is given by the level sets of $u'$ and $r'$. In the sequel, we denote $S'=S'(u',r')$.\footnote{The spheres $S'(u',r')$ generate a foliation in $\M$. Indeed, let $p=(u_p,r_p,\th^1_p,\th^2_p)\in\M$ where the coordinates are taken w.r.t. the background outgoing PG foliation. We introduce the following map $\Psi^{(p)}(u,r)=(u+U_\infty(u,\th^1_p,\th^2_p),r+R_\infty(u,\th^1_p,\th^2_p))$. Then, the determinant of the Jacobian matrix of $\Psi^{(p)}$ is given by $J=1+\pr_uU_\infty>0$. By the inverse mapping theorem, there exists a unique $(u',r')$ such that $\Psi^{(p)}(u',r')=(u_p,r_p)$, which shows that $p\in S'(u',r')=\Psi^{(p)}(S(u,r))$.} Moreover, $u'$ will be reparametrize by \eqref{eq:reparam} in Step 5.
\item We introduce the following transition functions:
\begin{align}\label{ffblaintrodu}
    f=\frac{F}{r},\qquad\quad \fb=\frac{\Fb}{r},\qquad\quad\la=1+\frac{\La}{r}.
\end{align}
Starting  with the background frame $(e_3,e_4,e_1,e_2)$, we use the transformation formula \eqref{General-frametransformation}, with  the  $(f,\fb,\la)$ as above, to define the new frame \footnote{The new frame $(e'_3,e'_4,e'_1,e'_2)$ may not be $S'$--adapted, i.e. $e'_1$ and $e'_2$ may not tangent to $S'$. $(e'_3,e'_4,e'_1,e'_2)$.}
\item Considering $\Jpp$ as a function of $(u,\th^1,\th^2)$, we extend it on any $S'(u',r')$, which is the deformation of a background sphere $S(u,r)$, as follows:
\begin{align}\label{Jppdf}
\Jpp(u+U_\infty(u,\th^1,\th^2),r+R_\infty(u,\th^1,\th^2),\th^1,\th^2)=\Jpp(u,\th^1,\th^2).
\end{align}
The properties of $\Jpp$ will be deduced in Step 6.
\item The background function $u$ (of the background outgoing PG foliation) induces a $u$--foliation on $\II^+$. We denote by $S(u)$ the corresponding spheres on $\II^+$. The  new $S'(u',r')$--foliation of $\far$ induces a $S'(u')$--foliation on $\II^+$ as follows:
\begin{align*}
\Psi:S(u)&\rightarrow S'(u'),\\
\left(u,\th^1,\th^2\right)&\rightarrow\left(u+U_\infty(u,\th^1,\th^2),\th^1,\th^2\right).
\end{align*}
\end{enumerate}
\noindent{\bf Step 2.} Null transformation formulae and limiting quantities.\\ \\
As an immediate consequence of \eqref{estfnfbnovlan}, \eqref{estUnRn} and \eqref{eq:limitsFFbLaJp}, we have
\begin{align}\label{limitingestimates}
    \|(F,\Fb,\La)\|_{\hk_{\ks+1}(S')}\les r,\qquad \|(U_\infty,R_\infty)\|_{\hk_{\ks+2}(S)}\les r.
\end{align}
Combining with \eqref{ovbest} and \eqref{ffblaintrodu}, we deduce
\begin{align}\label{assumptionusedhere}
    (f,\fb,\ovla)=O(r^{-1}), \qquad\quad r'=r-R^\#,
\end{align}
where $R^\#$ was defined in \eqref{dfURdiese}. Under the assumptions in \eqref{assumptionusedhere}, the transformation formulae in Proposition \ref{Riccitransfer} can be written as follows:
\begin{align}
\begin{split}\label{riccit}
r^2\left(\trch'-\frac{2}{r'}\right)&=r^2\left(\trch-\frac{2}{r}\right)-2R^\#+r\div'(rf)+rf\c r\eta+O(r^{-1}),\\
r^2\,\atrch'&=r^2\,\atrch+r\curl'(rf)+rf\wedge r\eta+O(r^{-1}),\\
r^2\left(\trchb'+\frac{2}{r'}\right)&=r^2\left(\trchb+\frac{2}{r}\right)+2R^\#+r\div'(r\fb)+rf\c r\xib+O(r^{-1}),\\
r^2\,\atrchb'&=r^2\,\atrchb+r\curl'(r\fb)+rf\wedge r\xib+O(r^{-1}),\\
r^2\ze'&=r^2\ze-r\nab'(r\ovla)+\frac{1}{2}(rf+r\fb)+O(r^{-1}),\\
r\eta'&=r\eta+\frac{1}{2}\nab'_3(rf)+O(r^{-1}),\\
r\xib'&=r\xib+\frac{1}{2}\nab'_3(r\fb)+O(r^{-1}),\\
r\omb'&=r\omb+\frac{1}{2}\nab'_3(r\ovla)+O(r^{-1}).
\end{split}
\end{align}
The transformation formulae in Proposition \ref{Curvaturetransfer} can be written as follows:
\begin{align}
\begin{split}\label{curvaturet}
r^3\mu'&=r^3\mu+(r^2\De'+2)(r\ovla)+\fl[\mu](rf,r\fb,r\ovla)+O(r^{-1}),\\
r^5\div'\b'&=r^5\div\b+\frac{3r^3\rho}{2}r\div'(rf)+\fl[\div\b](rf,r\fb,r\ovla)+O(r^{-1}),\\
r^5\curl'\b'&=r^5\curl\b+\frac{3r^3\rho}{2}r\curl'(rf)+\fl[\curl\b](rf,r\fb,r\ovla)+O(r^{-1}),\\
r^3\rho'&=r^3\rho+\fl[\rho](rf,r\fb,r\ovla)+O(r^{-1}),\\
r^3\rhod'&=r^3\rhod+\fl[\rhod](rf,r\fb,r\ovla)+O(r^{-1}),\\
r^2\bb'&=r^2\bb+\fl[\bb](rf,r\fb,r\ovla)+O(r^{-1}),\\
r\aa'&=r\aa+\fl[\aa](rf,r\fb,r\ovla)+O(r^{-1}),
\end{split}
\end{align}
where $\fl[...](rf,r\fb,r\ovla)$ denote error terms of the following type:
\begin{align*}
    \dko\left((rf,r\fb,r\ovla)^k\c\left(r\Gab,r^2\Gag\right)\right),\qquad 1\leq k\leq 3.
\end{align*}
Combining with Proposition \ref{limitexist} and \eqref{ffblaintrodu}, we deduce that all the limiting quantities defined in Proposition \ref{limitexist} related to $(e'_3,e'_4,e'_1,e'_2)$ also exist.\\ \\
We now focus on the existence of $(\Bk',\Bkd')$ defined as in Proposition \ref{limitB}. For any scalar function $h$ on a sphere $S$ with a $\ell=1$ basis $\Jp$, we introduce the following shorthand notations:
\begin{align*}
(h)_{\ell=0,S}:=\frac{1}{|S|}\int_S h,\qquad\quad (h)_{\ell=1,S,\Jp}:=\frac{1}{|S|}\int_S h\Jp,\qquad p=0,+,-.
\end{align*}
We also denote
\begin{align*}
    (h)_{\ell\geq 1,S}:=h-(h)_{\ell=0,S},\qquad (h)_{\ell\geq 2,S,\Jp}:=h-(h)_{\ell=0,S}-(h)_{\ell=1,S,\Jp}.
\end{align*}
Then, we have from \eqref{Jppdf}
\begin{align}
\begin{split}\label{newBkexists}
    &\quad\;|S'|\,(r^5d_1\b)_{\ell=1,S',\Jpp}\\
    &=\int_{\SSS^2}\left(r^5d_1\b\sqrt{\det\g}\right)(u+U_\infty,r+R_\infty,\th^1,\th^2)\Jpp(u,\th^1,\th^2) d\th^1 d\th^2\\
    &=\int_{\SSS^2}\left(r^5d_1\b\sqrt{\det\g}\right)(u,r,\th^1,\th^2)\Jpp(u,\th^1,\th^2) d\th^1 d\th^2\\
    &+\int_{\SSS^2}\int_0^1\pr_u\left(r^5d_1\b\sqrt{\det\g}\right)(u+tU_\infty,r+tR_\infty,\th^1,\th^2)dt\,(U_\infty\,\Jpp)(u,\th^1,\th^2)d\th^1d\th^2\\
    &+\int_{\SSS^2}\int_0^1\pr_r\left(r^5d_1\b\sqrt{\det\g}\right)(u+tU_\infty,r+tR_\infty,\th^1,\th^2)dt\,(R_\infty\,\Jpp)(u,\th^1,\th^2)d\th^1d\th^2.
\end{split}
\end{align}
Note that the existence of the following limits follows directly from Propositions \ref{prop-nullstrandBianchi:complex:outgoing} and \ref{limitexist}:
\begin{align*}
\lim_{C_u,r\to\infty}\left(\pr_u\left(r^5d_1\b\sqrt{\det\g}\right),\pr_r\left(r^5d_1\b\sqrt{\det\g}\right)\right).
\end{align*}
Moreover, we have from Proposition \ref{limitB} that the following limit exists:
\begin{align*}
\lim_{C_u,r\to\infty}\frac{1}{|S'|}\int_{\SSS^2}\left(r^5d_1\b\sqrt{\det\g}\right)(u,r,\th^1,\th^2)\Jpp(u,\th^1,\th^2) d\th^1 d\th^2.
\end{align*}
Hence, we infer the existence of the following limit:
\begin{align*}
    \lim_{C_u,r\to\infty}(r^5d_1\b)_{\ell=1,S',\Jpp}.
\end{align*}
Combining with \eqref{curvaturet}, \eqref{ffblaintrodu} and Proposition \ref{limitexist}, we deduce the existence of the following limit:
\begin{align*}
    (\Bk',\Bkd'):=\lim_{C_{u'},r'\to\infty}{r'}^5\left(\div'\b',\curl'\b'\right)_{\ell=1,S',\Jpp}.
\end{align*}
We recall from Proposition 2.1.43 in \cite{GKS} that the horizontal Gauss curvature is defined by
\begin{align*}
    ^{(h)}K:=-\frac{1}{4}\trch\,\trchb-\frac{1}{4}\atrch\,\atrchb+\frac{1}{2}\hch\c\hchb-\rho.
\end{align*}
As an immediate consequence of Proposition \ref{GagGabdecay}, we have
\begin{align*}
    {}^{(h)}K-\frac{1}{r^2}=O\left(\frac{\ep_0}{r^3}\right).
\end{align*}
Combining with \eqref{riccit} and \eqref{curvaturet}, we easily obtain that on $S'$\footnote{Here, we denoted $r'_{S'}$ the area radius of $S'$, which satisfies $|r'_{S'}-r'|\les 1$.}
\begin{align*}
    r^2\,{}^{(h)}K'-1=O\left(\frac{1}{{r}}\right).
\end{align*}
Letting $r'\to\infty$, we deduce that for any fixed $u'$, $(S'(u'),r^{-2}\g|_{S'(u')})$ is a round sphere on $\II^+$.\\ \\
{\bf Step 3.} Identities for $(U_\infty, R_\infty, F,\Fb,\La,\Jpp)$.\\ \\
Let $\Cb_n$ be the incoming null cone constructed in Proposition \ref{limittranstion}, together with the transition functions $(f_n,\fb_n,\la_n)$. Then, we have from \eqref{riccit}\footnote{Here, $(f_n,\fb_n,\ovla_n)$ satisfies the condition \eqref{estfnfbnovlan}.} that, the following hold on $\Cb_n$:\footnote{Throughout this step, we denote $\S:=\S_n(u)$, the leaves of the incoming geodesic foliation on $\Cb_n$ as in \eqref{shorthandS}, in order to simplify the notations.} 
\begin{align*}
    r^2\,\atrch^\S&=r^2\,\atrch+r\curl^\S(rf_n)+rf_n\wedge r\eta+O(r^{-1}),\\
    r^2\,\atrchb^\S&=r^2\,\atrchb+r\curl^\S(r\fb_n)+rf_n\wedge r\xib+O(r^{-1}),\\
    r(\eta^\S-\ze^\S)&=r(\eta-\ze)+\frac{1}{2}\nab^\S_3(rf_n)+O(r^{-1}),\\
    r\xib^\S&=r\xib+\frac{1}{2}\nab^\S_3(r\fb_n)+O(r^{-1}),\\
    r\omb^\S&=r\omb+\frac{1}{2}\nab^\S_3(r\ovla_n)+O(r^{-1}).
\end{align*}
Applying the incoming geodesic conditions and Proposition \ref{GagGabdecay}, we deduce
\begin{align*}
    r^2\,\atrch+r\curl^\S(rf_n)+rf_n\wedge r\eta&=O(r^{-1}),\\
    r^2\,\atrchb+r\curl^\S(r\fb_n)+rf_n\wedge r\xib&=O(r^{-1}),\\
    r\eta+\frac{1}{2}\nab_3^\S(rf_n)&=O(r^{-1}),\\
    r\xib+\frac{1}{2}\nab_3^\S(r\fb_n)&=O(r^{-1}),\\
    r\omb+\frac{1}{2}\nab_3^\S(r\ovla_n)&=O(r^{-1}).
\end{align*}
For any fixed $u$, taking $n\to\infty$ and applying Propositions \ref{limitexist} and \ref{limittranstion}, we infer\footnote{Recall that $\curlo$ and $\nabo_3$ is defined in \eqref{dfnabo}.}
\begin{align}
\begin{split}\label{Gconditions}
    \aXscr+\curlo'(F)+F\wedge\Hscr&=0,\\
    \aXbscr+\curlo'(\Fb)+F\wedge\Ybscr&=0,\\
    \Hscr+\frac{1}{2}\nabo_3'F&=0,\\
    \Ybscr+\frac{1}{2}\nabo_3'\Fb&=0,\\
    \nabo_3'\La&=0.
\end{split}
\end{align}
Next, we have from \eqref{riccit} and \eqref{curvaturet} that on $\Cb_n$:
\begin{align*}
    \left(r^5\div^\S\b^\S\right)_{\ell=1,\S,\JpS}&=(r^5\div\b)_{\ell=1,\S,\JpS}+\frac{3}{2}\left(r^4\rho\div^\S(rf_n)\right)_{\ell=1,\S,\JpS}\\
    &+\fl[\div\b](rf_n,r\fb_n,r\ovla_n)+O(r^{-1}),\\
    \left(r^5\curl^\S\b^\S\right)_{\ell=1,\S,\JpS}&=(r^5\curl\b)_{\ell=1,\S,\JpS}+\frac{3}{2}\left(r^4\rho\curl^\S(rf_n)\right)_{\ell=1,\S,\JpS}\\
    &+\fl[\curl\b](rf_n,r\fb_n,r\ovla_n)+O(r^{-1}),\\
    (r^3\mu^\S)_{\ell\geq 2,\S,\JpS}&=(r^3\mu)_{\ell\geq 2,\S,\JpS}+\left((r^2\De^\S+2)(r\ovla_n)\right)_{\ell\geq 2,\S,\JpS}\\
    &+\fl[\mu](rf_n,r\fb_n,r\ovla_n)+O(r^{-1}),\\
    r^2\left(\trch^\S-\frac{2}{\rS}\right)&=r^2\left(\trch-\frac{2}{r}\right)+r\div^\S(rf_n)-2R^\#_n+rf_n\c r\eta+O(r^{-1}),\\
    r^2\left(\trchb^\S+\frac{2}{\rS}\right)&=r^2\left(\trchb+\frac{2}{r}\right)+r\div^\S(r\fb_n)+2R^\#_n+rf_n\c r\xib+O(r^{-1}),
\end{align*}
where we denoted
\begin{align*}
    R^\#_n:=R_n\circ\Psi_n^{-1}.
\end{align*}
Combining with \eqref{GCMquantitiesdecayonSin}, we have on $\Cb_n$:
\begin{align*}
    (r^5\div\b)_{\ell=1,\S,\JpS}+\frac{3}{2}(r^4\rho\div^\S(rf_n))_{\ell=1,\S,\JpS}+\fl[\div\b](rf_n,r\fb_n,r\ovla_n)&=O\left(\frac{\ep_0}{u^{\dec}}\right),\\
    (r^5\curl\b)_{\ell=1,\S,\JpS}-(2am,0,0)&\\
    +\frac{3}{2}(r^4\rho\curl^\S(rf_n))_{\ell=1,\S,\JpS}+\fl[\curl\b](rf_n,r\fb_n,r\ovla_n)&=O\left(\frac{\ep_0}{u^{\dec}}\right),\\
    (r^3\mu)_{\ell\geq 2,\S,\JpS}+\left((r^2\De^\S+2)(r\ovla_n)\right)_{\ell\geq 2,\S,\JpS}+\fl[\mu](rf_n,r\fb_n,r\ovla_n)&=O\left(\frac{\ep_0}{u^{\dec}}\right),\\
    r^2\left(\trch-\frac{2}{r}\right)+r\div^\S(rf_n)-2R^\#_n+rf_n\c r\eta&=O\left(\frac{\ep_0}{u^{\dec}}\right),\\
    r^2\left(\trchb+\frac{2}{r}\right)-4m+r\div^\S(r\fb_n)+2R^\#_n+rf_n\c r\xib&=O\left(\frac{\ep_0}{u^{\dec}}\right).
\end{align*}
For any fixed $u$, taking $n\to\infty$ and applying Propositions \ref{limitexist} and \ref{limittranstion}, we infer\footnote{Proceeding as in \eqref{newBkexists}, we can easily deduce the existence of $\lim_{C_u,r\to\infty}(r^5d_1\b)_{\ell=1,S',\Jpp}$.}
\begin{align}
\begin{split}\label{GCMconditions}
    \lim_{C_u,r\to\infty}(r^5\div\b)_{\ell=1,S',\Jpp}+\frac{3}{2}(\Pscr\divo' F)_{\ell=1,S',\Jpp}+\fl[\div\b](F,\Fb,\La)&=O\left(\frac{\ep_0}{u^{\dec}}\right),\\
    \lim_{C_u,r\to\infty}(r^5\curl\b)_{\ell=1,S',\Jpp}-(2am,0,0)\\
    +\frac{3}{2}(\Pscr\curlo' F)_{\ell=1,S',\Jpp}+\fl[\curl\b](F,\Fb,\La)&=O\left(\frac{\ep_0}{u^{\dec}}\right),\\
    (\Mscr)_{\ell\geq 2,S',\Jpp}+((\Deo'+2)\La)_{\ell\geq 2,S',\Jpp}+\fl[\mu](F,F,\La)&=O\left(\frac{\ep_0}{u^{\dec}}\right),\\
    \Xscr+\divo'F-2R^\#+F\c\Hscr&=O\left(\frac{\ep_0}{u^{\dec}}\right),\\
    \Xbscr-4m+\divo'\Fb+2R^\#+F\c\Ybscr&=O\left(\frac{\ep_0}{u^{\dec}}\right),
\end{split}
\end{align}
where $\left(\divo',\curlo',\Deo',\Jpp\right)$ are defined w.r.t. $S'(u')$.\\ \\
{\bf Step 4.} Limiting GCM (LGCM) conditions.\\ \\
We have from \eqref{ffblaintrodu} and \eqref{riccit}
\begin{align*}
r^2\left(\trch'-\frac{2}{r'}\right)&=r^2\left(\trch-\frac{2}{r}\right)+r\div'F-2R^\#+F\c r\eta+O(r^{-1}),\\
r^2\,\atrch'&=r^2\,\atrch+r\curl'F+F\wedge r\eta+O(r^{-1}),\\
r^2\left(\trchb'+\frac{2}{r'}\right)&=r^2\left(\trchb+\frac{2}{r}\right)+r\div'\Fb+2R^\#+F\c r\xib+O(r^{-1}),\\
r^2\,\atrchb'&=r^2\,\atrchb+r\curl'\Fb+F\wedge r\xib+O(r^{-1}),\\
r\eta'&=r\eta+\frac{1}{2}\nab'_3F+O(r^{-1}),\\
r\xib'&=r\xib+\frac{1}{2}\nab'_3\Fb+O(r^{-1}),\\
r\omb'&=r\omb+\frac{1}{2}\nab'_3\La+O(r^{-1}).
\end{align*}
Taking $r\to\infty$ and combining with \eqref{Gconditions} and \eqref{GCMconditions}, we deduce
\begin{align}
\begin{split}\label{GCMfirstcon}
    \Xscr'&=\Xscr+\divo'F-2R^\#+F\c\Hscr=O\left(\frac{\ep_0}{u^{\dec}}\right),\\
    \aXscr'&=\aXscr+\curlo'F+F\wedge\Hscr=0,\\
    \Xbscr'&=\Xbscr+\divo'\Fb+2R^\#+F\c\Ybscr=4m+O\left(\frac{\ep_0}{u^{\dec}}\right),\\
    \aXbscr'&=\aXbscr+\curlo'\Fb+F\wedge\Ybscr=0,\\
    \Hscr'&=\Hscr+\frac{1}{2}\nabo'_3F=0,\\
    \Ybscr'&=\Ybscr+\frac{1}{2}\nabo'_3\Fb=0,\\
    \Wbscrone'&=\Wbscrone+\frac{1}{2}\nabo'_3\La=0.
    \end{split}
\end{align}
Next, we have from \eqref{ffblaintrodu} and \eqref{curvaturet}
\begin{align*}
r^3\mu'&=r^3\mu+(r^2\De'+2)\La+\fl[\mu](F,\Fb,\La)+O(r^{-1}),\\
r^5\div'\b'&=r^5\div\b+\frac{3r^3\rho}{2}r\div'F+\fl[\div\b](F,\Fb,\La)+O(r^{-1}),\\
r^5\curl'\b'&=r^5\curl\b+\frac{3r^3\rho}{2}r\curl'F+\fl[\curl\b](F,\Fb,\La)+O(r^{-1}),
\end{align*}
which implies from Proposition \ref{limitexist} and \eqref{GCMconditions}
\begin{align}
\begin{split}\label{GCMsecondcon}
    (\Mscr')_{\ell\geq 2,S',\Jpp}&=(\Mscr)_{\ell\geq 2,S',\Jpp}+\left((\Deo'+2)\La\right)_{\ell\geq 2,S',\Jpp}+\fl[\mu](F,\Fb,\La)\\
    &=O\left(\frac{\ep_0}{u^{\dec}}\right),\\
    \Bk'&=\lim_{C_u,r\to\infty}(r^5\div\b)_{\ell=1,S',\Jpp}+\frac{3\Pscr}{2}(\divo'F)_{\ell=1,S',\Jpp}+\fl[\div\b](F,\Fb,\La)\\
    &=O\left(\frac{\ep_0}{u^{\dec}}\right),\\
    \Bkd'&=\lim_{C_u,r\to\infty}(r^5\curl\b)_{\ell=1,S',\Jpp}+\frac{3\Pscr}{2}(\curlo'F)_{\ell=1,S',\Jpp}+\fl[\curl\b](F,\Fb,\La)\\
    &=(2am,0,0)+O\left(\frac{\ep_0}{u^{\dec}}\right).
\end{split}
\end{align}
Combining \eqref{GCMfirstcon} and \eqref{GCMsecondcon}, we deduce that the conditions \eqref{GCMS2limit}--\eqref{Incominggeodesic} hold.\\ \\
{\bf Step 5.} Properties of $u'$ and $r'$.\\ \\
We denote
\begin{align*}
    U^\#_n:=U_n\circ\Psi_n^{-1},\qquad R_n^\#:=R_n\circ \Psi_n^{-1}.
\end{align*}
Since the null frame $e_a^{\S_n}$ is adapted to $\S_n$, we have on $\Cb_n$
\begin{align*}
    re_a^{\S_n}(u-U^\#_n,r-R^\#_n)=0, \qquad a=1,2.
\end{align*}
Combining with \eqref{General-frametransformation}, we infer
\begin{align*}
    \left(re_a+\frac{1}{2}(rf_n)_a e_3\right)(u-U^\#_n, r-R^\#_n)=O(r^{-1}), \qquad a=1,2.
\end{align*}
Combining once again with \eqref{General-frametransformation} and \eqref{ffblaintrodu}, we obtain
\begin{align*}
    re'_a(u',r')&=\left(re_a+\frac{1}{2}F_a e_3\right)(u-U^\#,r-R^\#)+O(r^{-1})\\
    &=\left(re_a+\frac{1}{2}(rf_n)_a e_3\right)\left(u-U^\#,r-R^\#\right)-\frac{1}{2}(F-rf_n)_ae_3\left(u-U^\#,r-R^\#\right)+O(r^{-1})\\
    &=O(r^{-1}).
\end{align*}
We also have in $\far$
\begin{align*}
    e'_3(r')&=e_3(r-R^\#)=-1+O(\ep_0),\\
    e'_4(r')&=\la\left(e_4+f_ae_a+\frac{1}{4}|f|^2e_3\right)(r)=1+O(r^{-1}),\\
    e'_4(u')&=\la\left(e_4+f_ae_a+\frac{1}{4}|f|^2e_3\right)(u-U^\#)=O(r^{-2}).
\end{align*}
Moreover, we have from Proposition \ref{limittranstion}
\begin{align*}
    |r'-r|=|R^\#|=\lim_{n\to\infty}|R_n^\#|\les 1.
\end{align*}
Next, we have from Lemma \ref{lemma:comm-gen}\footnote{Recall that $\eta$ has only $r^{-1}$--decay in the background outgoing PG foliation, but we have $\eta'=O(r^{-2})$ as a consequence of our choice of the LGCM foliation.}
\begin{align*}
re'_a(e'_3(u'))=e'_3(re'_a(u'))+[re'_a,e'_3](u')=-r\eta_a'e_3'(u')+O(r^{-1})=O(r^{-1}).
\end{align*}
Hence, we deduce that on $\II^+$
\begin{align*}
    \nabo'(e'_3(u'))=0,
\end{align*}
which implies
\begin{align*}
e'_3(u')=z'+O(r^{-1}),
\end{align*}
where $z'$ is a function only depends on $u'$ satisfying
\begin{align*}
    z'=2+O(\ep_0).
\end{align*}
We now introduce a new parameter\footnote{The reparametrization \eqref{eq:reparam} is necessary to ensure $e'_3(u')=2+O(r^{-1})$, as stated in \eqref{urconditions}. Moreover, the condition $e'_3(u')=2$ on $\II^+$ is used to deduce the physical laws in Theorem \ref{thm-BHdynamic}.}
\begin{align}\label{eq:reparam}
    u'':=h(u'),
\end{align}
where $h$ is a one variable function satisfying
\begin{align}\label{dhdu'}
    \frac{dh}{du'}=\frac{2}{z'}=1+O(\ep_0).
\end{align}
Then, we have
\begin{align*}
    e'_3(u'')=\frac{dh}{du'}\left(e'_3(u')\right)=2+O(r^{-1}).
\end{align*}
Taking $u''$ as the new parameter, still denoted by $u'$, and combining with \eqref{dhdu'}, we deduce that all the desired estimates in \eqref{urconditions} hold.\footnote{Note that we have from \eqref{eq:reparam} that the level sets of $u''$ and $u'$ coincide with each other. We also have from \eqref{dhdu'} that $u'\sim u''$.}\\ \\
{\bf Step 6.} Properties of $\Jpp$.\\ \\
Differentiating \eqref{Jppdf} by $e_4$, we obtain
\begin{align*}
e_4(\Jpp)=0\quad \mbox{ on }\, S'(u',r').
\end{align*}
Combining with \eqref{General-frametransformation} and \eqref{ffblaintrodu}, we infer
\begin{align}\label{e4Jpp}
    e_4'(\Jpp)=\la\left(e_4+f_ae_a+\frac{1}{4}|f|^2e_3\right)\Jpp=O(r^{-2}).
\end{align}
Next, recalling that $\Jpp$ is obtained in Proposition \ref{AAapply} as the limit of $J^{(p,\S_n(u))}$, we have from \eqref{nab3JpS} that on $\Cb_n$
\begin{align*}
    e_3^{\S_n(u)}(J^{(p,\S_n(u))})=0,
\end{align*}
which implies from \eqref{General-frametransformation}
\begin{align*}
    e_3(J^{(p,\S_n(u))})=O(r^{-1}).
\end{align*}
Taking $n\to\infty$ and applying Proposition \ref{AAapply}, we easily obtain on $\II^+$
\begin{align*}
    e_3'(\Jpp)=0.
\end{align*}
Applying Lemma \ref{lemma:comm-gen}, we have
\begin{align*}
    e'_4(e'_3(\Jpp))=O(r^{-2}).
\end{align*}
Integrating it in the direction of $e'_4$, we deduce
\begin{align}\label{e3Jpp}
    e'_3(\Jpp)=O(r^{-1}).
\end{align}
Combining \eqref{e4Jpp} and \eqref{e3Jpp}, we obtain \eqref{nab3Jpp}. Recalling from Proposition \ref{limittranstion} that $J^{(p,\S_n(u))}$ is a $\ep_0 r^{-1}$--approximate basis of $\ell=1$ modes, we have\footnote{Similarly as before, we use the shorthand notation $\S=\S_n(u)$ introduced in \eqref{shorthandS}.}
\begin{align*}
    (r^2\De^\S+2)\JpS&=O\left(\frac{\ep_0}{r}\right),\\
    \frac{1}{|\S|}\int_\S \JpS J^{(q,\S)}&=\frac{4\pi}{3}\de_{pq}+O\left(\frac{\ep_0}{r}\right),\\
    \frac{1}{|\S|}\int_\S \JpS &=O\left(\frac{\ep_0}{r}\right).
\end{align*}
Taking $n\to\infty$ and recalling that the spheres $S'(u')$ on $\II^+$ are round spheres, we infer
\begin{align}
\begin{split}\label{Jppround}
    (\Deo'+2)\Jpp&=0,\\
    \int_{\SSS^2}\Jpp{J'}^{(q)}d\si_{\SSS^2}&=\frac{4\pi}{3}\de_{pq},\\
    \int_{\SSS^2}\Jpp d\si_{\SSS^2}&=0.
\end{split}
\end{align}
Next, we have from \eqref{e4Jpp} and Lemma \ref{lemma:comm-gen}
\begin{align*}
    e_4'\left((r^2\De'+2)\Jpp\right)&=O(r^{-2}),\\
    e_4'\left(\frac{1}{|S'|}\int_{S'}\Jpp{J'}^{(q)}\right)&=O(r^{-2}),\\
    e_4'\left(\frac{1}{|S'|}\int_{S'}\Jpp\right)&=O(r^{-2}).
\end{align*}
Combining with \eqref{Jppround}, we deduce that $\Jpp$
\begin{align*}
(r^2\De'+2)\Jpp&=O(r^{-1}),\\
\frac{1}{|S'|}\int_{S'}\Jpp{J'}^{(q)}&=\frac{4\pi}{3}\de_{pq}+O(r^{-1}),\\
\frac{1}{|S'|}\int_{S'}\Jpp&=O(r^{-1}),
\end{align*}
which implies that $\Jpp$ is a $r^{-1}$--approximate basis of $\ell=1$ modes on the spheres $S'(u',r')$ in the sense of Definition \ref{jpdef}.\\ \\
Combining the above results, we deduce that $\{S'(u',r'), (e_3',e_4',e_1',e_2')\}$ satisfies all the desired properties in Definition \ref{limitinggeodesic}. Hence, $\{S'(u',r'), (e_3',e_4',e_1',e_2')\}$ is a LGCM foliation as desired. This concludes the proof of Theorem \ref{LGCMconstruction}.
\section{Absence of supertranslation ambiguity and spatial translation freedom}\label{secsuper}
The purpose of this section is to prove the following theorem, which implies the uniqueness of the LGCM foliation introduced in Definition \ref{limitinggeodesic}.
\begin{thm}\label{translation}
Let $(\M,\g)$ be a $\KSAF(a,m)$ spacetime. Assume that there are two LGCM foliations $S'(u',r')$ and $S''(u'',r'')$ associated to the background outgoing PG $S(u,r)$--foliation of $\M$. We denote $(f,\fb,\la)$ the transition functions from the null frame $(e''_3,e''_4,e''_1,e''_2)$ of the $S''$--foliation to the null frame $(e_3',e_4',e_1',e_2')$ of the $S'$--foliation. Then, we have
    \begin{align*}
        f=o(r^{-1}),\qquad \fb=o(r^{-1}),\qquad \la=1+o(r^{-1}).
    \end{align*}
Moreover, the spheres of $S'$--foliation and $S''$--foliation coincide on $\II^+$.
\end{thm}
\begin{rk}\label{nosuperspatialrk}
    Theorem \ref{translation} implies that the conditions of the LGCM foliation completely fix the foliation on $\II^+$. Combining with \eqref{riccit} and \eqref{curvaturet}, all the limiting quantities, hence the physical quantities, in $S'$--foliation and $S''$--foliation coincide. This avoids supertranslation ambiguity and spatial translation freedom.
\end{rk}
\begin{proof}[Proof of Theorem \ref{translation}]
As an immediate consequence of \eqref{usedasbootstrap} and \eqref{urconditions}, we have
\begin{align}\label{bootrf}
    \|(f,\fb,\ovla)\|_{\hk_4(S')}\les 1,\qquad\quad |r'-r''|\les 1.
\end{align}
The proof is then divided into 2 steps as follows.\\ \\
{\bf Step 1.} We first prove the uniqueness at $i^+$. To this end, we deduce an elliptic system for $(f,\fb,\ovla,r'-r'')$ from the LGCM conditions \eqref{GCMS2limit}--\eqref{integrablelimitsindef} of $S'$--foliation and $S''$--foliation, Using this elliptic system, we can prove that $(rf,r\fb,r\ovla,r'-r'')$ vanishes at $i^+$.\\ \\
We have from Propositions \ref{Riccitransfer} and \ref{Curvaturetransfer}
\begin{align*}
    r^2\left(\trch'-\frac{2}{r'}\right)&=r^2\left(\trch''-\frac{2}{r''}\right)+r^2\div'f+2(r'-r'')+2r\ovla+O(r^{-1}),\\
    r^2\,\atrch'&=r^2\,\atrch''+r^2\curl'f+O(r^{-1}),\\
    r^2\left(\trchb'+\frac{2}{r'}\right)&=r^2\left(\trchb''+\frac{2}{r''}\right)+r^2\div'\fb+2(r''-r')+2r\ovla+O(r^{-1}),\\
    r^2\,\atrchb'&=r^2\,\atrchb''+r^2\curl'\fb+O(r^{-1}).
\end{align*}
Taking $r\to\infty$, we deduce
\begin{align*}
    \Xscr'&=\Xscr''+\lim_{C_u,r\to\infty}\left(2(r'-r'')+r\div'(rf)+2r\ovla\right),\\
    \aXscr'&=\aXscr''+\lim_{C_u,r\to\infty}\left(r\curl'(rf)\right),\\
    \Xbscr'&=\Xbscr''+\lim_{C_u,r\to\infty}\left(2(r''-r')+r\div'(r\fb)+2r\ovla\right),\\
    \aXbscr'&=\aXbscr''+\lim_{C_u,r\to\infty}\left(r\curl'(r\fb)\right).
\end{align*}
Combining with \eqref{GCMS2limit} and \eqref{integrablelimitsindef}, we infer
\begin{align}
\lim_{u\to\infty}\lim_{C_u,r\to\infty}\left(2(r'-r'')+r\div'(rf)+2r\ovla\right)&=0,\\
\lim_{C_u,r\to\infty}\left(r\curl'(rf)\right)&=0,\label{curlf=0}\\
\lim_{u\to\infty}\lim_{C_u,r\to\infty}\left(2(r'-r'')+r\div'(r\fb)+2r\ovla\right)&=0,\\
\lim_{C_u,r\to\infty}\left(r\curl'(r\fb)\right)&=0.\label{curlfb=0}
\end{align}
We also have from Proposition \ref{curvaturet}
\begin{align}
    r^3\mu'&=r^3\mu''+(r^2\De'+2)(r\ovla)-\frac{r^2}{2}\left(\div'f+\div'\fb\right)+\err[\ep_0]+O(r^{-1}),\label{muchange}\\
    r^5\div'\b'&=r^5\div''\b''+\frac{3}{2}r^5\rho\div'f+\err[\ep_0]+O(r^{-1}),\label{divbchange}\\
    r^5\curl'\b'&=r^5\curl''\b''+\frac{3}{2}r^5\rho\curl'f+\err[\ep_0]+O(r^{-1}),\label{curlbchange}
\end{align}
where we denoted\footnote{In the end of this step, one can show that $\dko(rf,r\fb,r\ovla,r'-r'')$ has size $\err[\ep_0]$ at $i^+$. Combining with the smallness of $\ep_0$, we will deduce the vanishing nature of $\dko(rf,r\fb,r\ovla,r'-r'')$ at $i^+$, see \eqref{i+vanishing} below.}
\begin{align}\label{errep0df}
    \err[\ep_0]:=O(\ep_0)\dko(rf,r\fb,r\ovla,r'-r'').
\end{align}
Projecting \eqref{curlbchange} to the $r^{-1}$--approximate $\ell=1$ basis $\Jpp$ on the spheres $S'$ and taking $r\to\infty$, we deduce
\begin{align*}
\Bkd'=\lim_{C_u,r\to\infty}\left(r^5\curl''\b''\right)_{\ell=1,S',\Jpp}+\err[\ep_0],
\end{align*}
where we used \eqref{curlf=0}. Note that we have by definition
\begin{align*}
    \Bkd''=\lim_{C_u,r\to\infty}\left(r^5\curl''\b''\right)_{\ell=1,S'',{J''}^{(p)}}.
\end{align*}
Combining the above two identities and taking $u\to\infty$, we easily deduce from \eqref{angularmomentumcalibration}
\begin{align}\label{Jpok}
    {J''}^{(p)}-\Jpp=\err[\ep_0].
\end{align}
Next, proceeding as in \eqref{newBkexists}, we have
\begin{align*}
    (r^5\div''\b'')_{\ell=1,S',\Jpp}-(r^5\div''\b'')_{\ell=1,S'',{J^{''}}^{(p)}}=\err[\ep_0].
\end{align*}
Combining with \eqref{divbchange}, we obtain
\begin{align*}
    \Bk'=\Bk''+\frac{3}{2}\Pscr\lim_{C_u,r\to\infty}r\div'(rf)+\err[\ep_0].
\end{align*}
Similarly, we have
\begin{align*}
    (\Mscr')_{\ell\geq 2,S',\Jpp}&=(\Mscr'')_{\ell\geq 2,S'',{J''}^{(p)}}+\lim_{C_u,r\to\infty}(r^2\De'+2)(r\ovla)_{\ell\geq 2,S',\Jp}\\
    &-\lim_{C_u,r\to\infty}\frac{r^2}{2}\left(\div'f+\div'\fb\right)_{\ell\geq 2,S',\Jp}+\err[\ep_0].
\end{align*}
Combining the above estimates, we infer that
\begin{align}
\lim_{u\to\infty}\lim_{C_u,r\to\infty}\left(2(r'-r'')+r\div'(rf)+2r\ovla\right)&=0,\label{trch'''}\\
\lim_{C_u,r\to\infty}\left(r\curl'(rf)\right)&=0,\label{atrch'''}\\
\lim_{u\to\infty}\lim_{C_u,r\to\infty}\left(2(r''-r')+r\div'(r\fb)+2r\ovla\right)&=0,\label{trchb'''}\\
\lim_{C_u,r\to\infty}\left(r\curl'(r\fb)\right)&=0,\label{atrchb'''}\\
\lim_{u\to\infty}\frac{3}{2}\Pscr\lim_{C_u,r\to\infty}r\div'(rf)_{\ell=1,S'}&=\err[\ep_0],\label{divb'''}\\
\lim_{u\to\infty}\lim_{C_u,r\to\infty}\left((r^2\De'+2)(r\ovla)_{\ell\geq 2,S'}-\frac{r^2}{2}\left(\div'f+\div'\fb\right)_{\ell\geq 2,S'}\right)&=\err[\ep_0].\label{mu'''}
\end{align}
We also have from \eqref{General-frametransformation} and \eqref{urconditions}
\begin{align*}
    e'_a(r''-r')&=e''_a(r'')+\frac{1}{2}f_ae_3''(r'')+\frac{1}{2}\fb_ae''_4(r'')+O(r^{-2})\\
    &=-\frac{1}{2}f_a+\frac{1}{2}\fb_a+\err[\ep_0]+O(r^{-2}),
\end{align*}
which implies
\begin{align}\label{Deovb}
    \De'(r''-r')=-\frac{1}{2}\div'(f-\fb)+\err[\ep_0]+O(r^{-3}).
\end{align}
Combining the \eqref{trch'''} and \eqref{atrch'''}, we infer
\begin{align}
\lim_{u\to\infty}\lim_{C_u,r\to\infty}\left(r\div'(rf+r\fb)+2r\ovla\right)&=0,\label{important1}\\
\lim_{u\to\infty}\lim_{C_u,r\to\infty}\left(-2\De'(r''-r')-4(r''-r')+2r\ovla\right)&=0.\label{important}
\end{align}
Injecting them into \eqref{mu'''}, we obtain
\begin{align}\label{ellgeq2ok}
    \lim_{u\to\infty}\lim_{C_u,r\to\infty}\left(r\div'(rf+r\fb),\;r\ovla,\;r'-r''\right)_{\ell\geq 2}=\err[\ep_0].
\end{align}
Next, we have from \eqref{divb'''} and the fact that $\Pscr\ne 0$
\begin{align}\label{divfl=1}
    \lim_{u\to\infty}\lim_{C_u,r\to\infty}r\div'(rf)_{\ell=1}=\err[\ep_0].
\end{align}
Projecting \eqref{important} onto its $\ell=1$ modes, we deduce
\begin{align*}
    \lim_{u\to\infty}\lim_{C_u,r\to\infty}(r\ovla)_{\ell=1}=\err[\ep_0].
\end{align*}
Next, projecting \eqref{important1} onto its $\ell=0$ modes, we infer
\begin{align}\label{estovlaok}
     \lim_{u\to\infty}\lim_{C_u,r\to\infty}(r\ovla)_{\ell=0}=\err[\ep_0].
\end{align}
Hence, we infer
\begin{align*}
    \lim_{u\to\infty}\lim_{C_u,r\to\infty}r\ovla=\err[\ep_0].
\end{align*}
Injecting it into \eqref{trch'''}, we obtain
\begin{align*}
    \lim_{u\to\infty}\lim_{C_u,r\to\infty}\left(2(r'-r'')+r\div'(rf)\right)=\err[\ep_0].
\end{align*}
Projecting it onto its $\ell\geq 2$ modes and applying \eqref{ellgeq2ok}, we deduce
\begin{align*}
    \lim_{u\to\infty}\lim_{C_u,r\to\infty}\left(r\div'(rf)\right)_{\ell\geq 2}=\err[\ep_0].
\end{align*}
Combining with \eqref{divfl=1}, we infer
\begin{align}\label{fok}
    \lim_{u\to\infty}\lim_{C_u,r\to\infty}r\div'(rf)=\err[\ep_0].
\end{align}
Next, we have from \eqref{estovlaok} and \eqref{important1}
\begin{align*}
    \lim_{u\to\infty}\lim_{C_u,r\to\infty}\left(r\div'(rf+r\fb)\right)&=\err[\ep_0],
\end{align*}
combining with \eqref{fok}, we deduce
\begin{align}\label{fbok}
    \lim_{u\to\infty}\lim_{C_u,r\to\infty}r\div'(r\fb)=\err[\ep_0].
\end{align}
Finally, we have from \eqref{estovlaok}, \eqref{fok} and \eqref{trch'''} 
\begin{align*}
    \lim_{u\to\infty}\lim_{C_u,r\to\infty}(r'-r'')&=\err[\ep_0].
\end{align*}
Combining the above identities and applying Proposition \ref{ellipticLp}, we obtain that the following identity hold at $i^+$:
\begin{align}\label{i+vanishing}
    \lim_{u\to\infty}\lim_{C_u,r\to\infty}\dko\left(rf,r\fb,r\ovla,r'-r'',\Jpp-{J''}^{(p)}\right)=\err[\ep_0].
\end{align}
Combining with \eqref{errep0df}, we deduce for $\ep_0$ small enough
\begin{align}\label{i+vanishing0}
    \lim_{u\to\infty}\lim_{C_u,r\to\infty}\dko\left(rf,r\fb,r\ovla,r'-r'',\Jpp-{J''}^{(p)}\right)=0.
\end{align}
{\bf Step 2.} We now prove the uniqueness on $\II^+$. To this end, we deduce a transport system for $(f,\fb,\ovla,r'-r'')$ from the limiting geodesic conditions \eqref{Incominggeodesic} of $S'$--foliation and $S''$--foliation, Combining this transport system and \eqref{i+vanishing0} obtained in Step 1, we can prove that $(rf,r\fb,r\ovla,r'-r'')$ vanishes on $\II^+$.\\ \\
We have from \eqref{riccit}
\begin{align*}
    r\eta'&=r\eta''+\frac{1}{2}\nab_3'(rf)+O(r^{-1}),\\
    r\xib'&=r\xib''+\frac{1}{2}\nab_3'(r\fb)+O(r^{-1}),\\
    r\omb'&=r\omb''+\frac{1}{2}\nab_3'(r\ovla)+O(r^{-1}).
\end{align*}
Taking $r\to\infty$, we obtain
\begin{align*}
    \Hscr'&=\Hscr''+\frac{1}{2}\lim_{C_u,r\to\infty}\nab'_3(rf),\\
    \Ybscr'&=\Ybscr''+\frac{1}{2}\lim_{C_u,r\to\infty}\nab'_3(r\fb),\\
    \Wbscrone'&=\Wbscrone''+\frac{1}{2}\lim_{C_u,r\to\infty}\nab'_3(r\ovla).
\end{align*}
Applying \eqref{Incominggeodesic} for both $S'$--foliation and $S''$--foliation, we deduce
\begin{align*}
    \lim_{C_u,r\to\infty}\nab'_3(rf,r\fb,r\ovla)=0.
\end{align*}
We also have
\begin{align*}
    e_3(r''-r')=\frac{r}{2}\left(\ov{\trchb''}-\ov{\trchb'}\right)+O(r^{-1})=O(r^{-1}),
\end{align*}
Combining with \eqref{i+vanishing0}, we infer
\begin{align}\label{allok}
    \lim_{C_u,r\to\infty}(rf,r\fb,r\ovla,r'-r'')=0.
\end{align}
Finally, we have from \eqref{General-frametransformation}, \eqref{urconditions} and \eqref{allok}
\begin{align*}
    e'_a(u'-u'')&=-e'_a(u'')+O(r^{-2})=-\left(e''_a+\frac{1}{2}f_ae''_3+\frac{1}{2}\fb_ae''_4\right)(u'')=o(r^{-1}),\\
    e'_3(u'-u'')&=2-e''_3(u'')+O(r^{-1})=O(r^{-1}),\\
    e'_4(u'-u'')&=-\la\left(e''_4+f_ae''_a+\frac{1}{4}|f|^2e''_3\right)u''+O(r^{-2})=O(r^{-2}).
\end{align*}
We deduce thtat
\begin{align*}
    \lim_{C_{u'},r\to\infty}(u'-u'')=C,
\end{align*}
where $C$ is a constant. This implies that the spheres of $S'$--foliation and $S''$--foliation coincide on $\II^+$. This concludes the proof of Theorem \ref{translation}.
\end{proof}
\section{Physical quantities in the LGCM foliation}\label{ssecphy}
In this section, we deduce the physical laws in the LGCM foliation constructed in Theorem \ref{LGCMconstruction}. Throughout Section \ref{ssecphy}, all quantities are defined w.r.t. a given LGCM foliation near $\II^+$.
\begin{df}\label{dfquantities}
Given a $S(u)$--foliation on $\II^+$ and a standard $\ell=1$ basis $J^{(p)}$ on the standard round spheres $\SSS^2$ on $\II^+$, induced from a LGCM foliation introduced in Definition \ref{limitinggeodesic}, we define the null energy $\EE$, linear momentum $\PP$, angular momentum $\JJ$ and center of mass $\CC$ of $u$-section of $\II^+$ as follows:
    \begin{align}
    \begin{split}\label{definitionsEPJC}
        \EE(u)&:=(\Mk)_{\ell=0,S(u)}=\int_{\mathbb S^2}\Mk(u,\cdot)d\si_{\SSS^2},\\
        \PP^{0,\pm}(u)&:=(\Mk)_{\ell=1,S(u),\Jp}=\int_{\mathbb S^2}\Mk(u,\cdot)\, J^{(0,\pm)}(\cdot)d\si_{\SSS^2},\\
        \JJ^{(0,\pm)}(u)&:=\Bk(u)=\int_{\mathbb S^2}\curlo \BB (u,\cdot)\, J^{(0,\pm)}(\cdot)d\si_{\SSS^2},\\
        \CC^{0,\pm}(u)&:=\Bkd(u)=\int_{\mathbb S^2}\divo \BB (u,\cdot)\, J^{(0,\pm)}(\cdot) d\si_{\SSS^2},
    \end{split}
    \end{align}
    where $\BB$ is the formal limit $\lim_{C_u,r\to\infty}r^4\b$, see Remark \ref{JJCCexist} below.
\end{df}
\begin{rk}\label{JJCCexist}
Notice that $\BB$ does not necessarily exist in a $\KSAF$ spacetime. However, by the Bianchi equation of $\nab_4\b$, we can show that $\JJ^{(0,\pm)}(u)$ and $\CC^{(0,\pm)}(u)$ exist, see Proposition \ref{limitB}.
\end{rk}
\begin{rk}
We now describe the classical physical quantities defined in the Bondi-Sachs gauge. Take $(u,x^a)$, where $a=1,2$ and $x^1=\th, x^2=\phi$, to be the limiting Bondi-Sachs coordinates of null infinity $\II^+=(\mathbb R\times \mathbb S^2,0\oplus \si_{\SSS^2})$ in the Bondi-Sachs gauge, and let\footnote{These quantities are related to the quantities defined in Definition \ref{df:allquantities} as follows, see Appendix A of \cite{CKWWY}:
\begin{align*}
    \Mk=\lim_{C_u,r\to\infty}\frac{r^3}{2}(\mu+\mub),\qquad \BB=\lim_{C_u,r\to\infty}r^4\b,\qquad \Theta=\lim_{C_u,r\to\infty}r^2\hch.
\end{align*}
Note that $\BB$ may not exist in $\KSAF$ spacetime due to the lack of decay of $\b$.}
$$
\Mk=\Mk(u,x^a),\qquad  \BB=\BB(u,x^a),\qquad \Theta=\Theta(u,x^a),
$$
to be the Bondi mass aspect function, angular momentum aspect $1$-form and the shear tensor of $u$--constant sections along $\II^+$. Then the Bondi energy $\EE_{Bondi}$, Bondi linear momentum $\PP_{Bondi}^{0,\pm}$, Dray-Strubel angular momentum $\JJ^{(0,\pm)}_{DS}$ \cite{DS} and Flanagan-Nichols center of mass $\CC^{0,\pm}_{FN}$ \cite{FN} of $u$-constant section on $\II^+$ are defined by
    \begin{align*}
        \EE_{Bondi}(u)&:=\int_{\SSS^2} \Mk(u,\cdot) d\si_{\SSS^2}, \\
        \PP^{0,\pm}_{Bondi}(u)&:=\int_{\mathbb S^2} \Mk(u,\cdot)\, J^{(0,\pm)}(\cdot) d\si_{\SSS^2},\\
        \JJ^{(0,\pm)}_{DS}(u)&:=\int_{\SSS^2}\curlo\left(\BB -\frac{1}{4} \Theta\c \divo\Theta\right)(u,\cdot)\, J^{(0,\pm)}(\cdot)d\si_{\SSS^2},\\
        \CC^{0,\pm}_{FN}(u)&:=\int_{\SSS^2}\divo \left(\BB-\frac{1}{2} u\nabo \Mk-\frac{1}{4}\Theta\c \divo\Theta-\frac{1}{16}\nabo|\The|^2\right)(u,\cdot)\, J^{(0,\pm)}(\cdot)d\si_{\SSS^2}.
    \end{align*}
where $J^{(0,\pm)}$ is a standard basis of $\ell=1$ modes on the standard round sphere $(\mathbb S^2,\si_{\SSS^2})$. We can see that our definitions of $\EE,\PP$ are the same as $\EE_{Bondi},\PP_{Bondi}$ and the definitions of $\JJ,\CC$ are the leading order terms of $\JJ_{DS},\CC_{FN}$. 
\end{rk}
\begin{rk}
The angular momentum $\JJ_{DS}$, first proposed by Dray and Streubel \cite{DS}, was not introduced with the intention of resolving the supertranslation ambiguity. Instead, they define an angular momentum quantity conjugate to the BMS symmetry, corresponding to a fixed rotational Killing field on the spheres of constant $u$. The center of mass $\CC_{FN}$ is a quantity conjugate to the BMS symmetry, corresponding to a boost \cite{FN}. General BMS charges conjugate to any BMS symmetry can be defined using the Hamiltonian framework, see \cite{WaldZoupas}. 
\end{rk}
\begin{rk}
The Chen-Wang-Yau's supertranslation-invariant definitions of angular momentum and center of mass for a $u$-section of $\II^+$ can be regarded as modifications of the classical definitions $\JJ^{(0,\pm)}_{DS}$ and $\CC^{0,\pm}_{FN}$, using the closed and co-closed potential functions of the shear tensor on the corresponding $u$-section. See \cite{CKWWY} for more details on Chen-Wang-Yau's angular momentum and center of mass.
\end{rk}
\begin{lem}\label{evolutiononscri}
In the LGCM foliation, the following equations hold on $\II^+$:
\begin{align*}
    \nab_3\Pscr&=-\divo\Bbscr-\frac{1}{2}\The\c\Abscr,\\
    \nab_3\The&=-\Thb,\\
    \nab_3\Thb&=-\Abscr,\\
    \nab_3(\Bk,\Bkd)&=\ddo_1\left(\ddo^*_1(-\Pscr,\dual\Pscr)+2\Bbscr\c \The\right)_{\ell=1},\\
    \divo\Thb&=\Bbscr,\\
    \divo\The&=\frac{1}{2}\nabo\Xscr+\Zscr,\\
    \curlo\Zscr&=\dual\Pscr-\frac{1}{2}\The\wedge\Thb.
\end{align*}
\end{lem}
\begin{proof}
The proof follows by combining Propositions \ref{prop-nullstrandBianchi:complex:outgoing}, \ref{limitexist} and \ref{limitB}, together with the fact that $\Hscr=0$ within the LGCM foliation, as defined in Definition \ref{limitinggeodesic}.
\end{proof}
In the following theorem, we derive the laws of black hole dynamics on $\II^+$.
\begin{thm}\label{thm-BHdynamic}
In the LGCM foliation, the null energy, linear momentum, angular momentum and center of mass of a perturbed Kerr black hole satisfy the following evolution equations:
    \begin{align}
        \pr_u\EE&=-\frac{1}{4}\left(|\Thb|^2\right)_{\ell=0},\label{eq-evole-energy}\\
        \pr_u\PP&=-\frac{1}{4}\left(|\Thb|^2\right)_{\ell=1},\label{eq-evole-momentum}\\
        \pr_u\CC&=\PP+\left(\frac{1}{2}\The\c\Thb+\divo\left(\divo\Thb\c\The\right) \right)_{\ell=1},\label{eq-evole-com}\\
        \pr_u\JJ&=\left(\frac{1}{2}\The\wedge\Thb+\curlo\left(\divo\Thb\c\The\right) \right)_{\ell=1}.\label{eq-evole-angularmomentum}
    \end{align}
\end{thm}
\begin{rk}\label{rkBHdynamic}
    The linear part of the evolutions obeys the rule of classical mechanics for an isolated system, while the nonlinear part of the evolutions are described by the shears $\The$ and news $\Thb$ that reflect the nontrivial nature of gravitational radiation. Note that \eqref{eq-evole-angularmomentum} was established by An-He-Shen as an angular momentum memory effect, see Theorem 10 in \cite{AHS}.
\end{rk}
\begin{proof}[Proof of Theorem \ref{thm-BHdynamic}] The proof is divided into 2 steps.\\ \\
{\bf Step 1.} Recall that the Bondi mass aspect is defined by:
\begin{align}\label{dfBondimassrecall}
    \Mk=-\Pscr+\frac{1}{2}\The\c\Thb.
\end{align}
We have from Lemma \ref{evolutiononscri}
\begin{align*}
    \nabo_3\Mk&=\divo\Bbscr+\frac{1}{2}\The\c\Abscr-\frac{1}{2}\Thb\c\Thb-\frac{1}{2}\The\c\Abscr\\
    &=-\frac{1}{2}|\Thb|^2+\divo\divo\Thb.
\end{align*}
Then, we have from \eqref{definitionsEPJC}
\begin{align}
\begin{split}\label{nabo3EEPP}
    \nabo_3\EE&=\int_{\SSS^2}(\nab_3\Mk)d\si_{\SSS^2}=-\frac{1}{2}\int_{\SSS^2}|\Thb|^2d\si_{\SSS^2}=-\frac{1}{2}\left(|\Thb|^2\right)_{\ell=0},\\
    \nabo_3\PP&=\int_{\SSS^2}(\nab_3\Mk)\Jp d\si_{\SSS^2}=-\frac{1}{2}\int_{\SSS^2}|\Thb|^2\Jp d\si_{\SSS^2}=-\frac{1}{2}\left(|\Thb|^2\right)_{\ell=1}.
\end{split}
\end{align}
{\bf Step 2.} We recall from \eqref{dfBondimassrecall} and Lemma \ref{evolutiononscri}
\begin{align*}
     (-\Pscr,\dual\Pscr)=(\Mk,\curlo\Zscr)+\frac{1}{2}(\The\c\Thb,\The\wedge\Thb).
\end{align*}
We compute from Lemma \ref{evolutiononscri}
\begin{align*}
    \nabo_3(\Bk,\Bkd)&=-\Deo(-\Pscr,\dual\Pscr)_{\ell=1}+\doo_1\left(2\Bbscr\c\The\right)_{\ell=1}\\
    &=2(-\Pscr,\dual\Pscr)_{\ell=1}+2\doo_1\left(\divo\Thb\c\The\right)_{\ell=1}\\
    &=\left(2(\Mk,\curlo\Zscr)+(\The\c\Thb,\The\wedge\Thb)+2\doo_1\left(\divo\Thb\c \The\right)\right)_{\ell=1}\\
    &=2\left(\PP,(\curlo\Zscr)_{\ell=1}\right)+(\The\c\Thb,\The\wedge\Thb)_{\ell=1}+2\doo_1\left(\divo\Thb\c\The\right)_{\ell=1},
\end{align*}
where we used the fact that $\ddo_1\ddo^{*}_1=-(\Deo_0,\Deo_0)$, which follows from Proposition \ref{ddddprop}. We notice that by Lemma \ref{evolutiononscri}
\begin{align*}
    \left(\curlo\Zscr\right)_{\ell=1}=\curlo\left(\divo\The-\frac{1}{2}\nabo\Xscr\right)_{\ell=1}=\curlo\left(\divo\The\right)_{\ell=1}=0.
\end{align*}
Therefore, we deduce the following evolution equation of center of mass $\CC$ and angular momentum $\JJ$ along the LGCM foliation,
\begin{align}
\begin{split}\label{nabo3CCJJ}
\nabo_3\CC&= \nabo_3(\Bk)=2\PP+\left(\The\c\Thb+2\divo\left(\divo\Thb\c\The\right)\right)_{\ell=1},\\
\nabo_3\JJ&=\nab_3(\Bkd)=\left(\The\wedge\Thb+2\curlo\left(\divo\Thb\c\The\right)\right)_{\ell=1}.
\end{split}
\end{align}
Recalling that we have in LGCM foliation
\begin{align*}
    e_3(u)=2,\qquad e_3(\Jp)=0\qquad\mbox{ on }\;\II^+,
\end{align*}
which implies
\begin{align*}
    e_3=2\pr_u\quad \mbox{ on }\,\II^+.
\end{align*}
Injecting it into \eqref{nabo3EEPP} and \eqref{nabo3CCJJ}, this concludes the proof of Theorem \ref{thm-BHdynamic}.
\end{proof}
\begin{rk}
\label{rm:distorted}
We can also derive evolution equations for the physical quantities relative to the background outgoing PG foliation. Indeed, if we set as in Definition \ref{dfquantities}
\begin{align*}
\left(\gcm\EE,\gcm\PP,\gcm\CC,\gcm\JJ\right):=\left(\gcm\Mk_{\ell=0},\gcm\Mk_{\ell=1},\gcm\Bk,\gcm\Bkd\right),
\end{align*}
the same computations as above lead to distorted evolution equations. We have the following evolution equation for $\gcm\EE$:
\begin{align*}
    \pr_u\left(\gcm\EE\right)&=\left(-\frac{1}{4}\left|\gcm\Thb\right|^2+\gcm\Hscr\c\left(\gcm\divo\gcm\Thb\right)\right)_{\ell=0}\\
    &+\frac{1}{4}\left(\left(\gcm\nabo\hot\gcm\Hscr\right)\c\gcm\Thb+\left(\gcm\Hscr\hot\gcm\Hscr\right)\c\gcm\Thb\right)_{\ell=0}.
\end{align*}
The distortions are due to the fact that $\gcm\Hscr$ is nonzero. This illustrates the importance of the disappearance of $\Hscr$ in LGCM foliation.
\end{rk}
\section{Displacement of the center of mass in Kerr stability}\label{seckick}
\begin{df}\label{fixu'}
Let $(\M,\g)$ be a $\KSAF$ spacetime endowed with an outgoing PG $S(u,r)$--foliation and a LGCM foliation $S'(u',r')$ constructed in Theorem \ref{LGCMconstruction}. Let $S'_1$ be a leaf of the LGCM foliation on $\II^+$, which intersects the initial layer region $\LL_0$. We then calibrate $u'$ by assuming\footnote{Notice that the definition of $u'$ in Theorem \ref{LGCMconstruction} is unique up to a constant.}
\begin{align*}
    u'=1\qquad\mbox{ on }\;S'_1.
\end{align*}
\end{df}
\subsection{Total flux of physical quantities}
\begin{proposition}\label{totalflux}
Let $(\M,\g)$ be a $\KSAF$ spacetime endowed with an outgoing PG $S(u,r)$--foliation and we denote $S'(u',r')$ the leaves of the LGCM foliation, constructed in Theorem \ref{LGCMconstruction}. Then, the total flux of energy $\EE'(u')$, linear momentum $\PP'(u')$, angular momentum $\JJ'(u')$ and center of mass $\CC'(u')$ verify the following estimates:
\begin{align*}
    \de\EE'&:=\EE'(+\infty)-\EE'(1)=O(\ep_0),\\
    \de\PP'&:=\PP'(+\infty)-\PP'(1)=O(\ep_0),\\
    \de\JJ'&:=\JJ'(+\infty)-\JJ'(1)=O(\ep_0),\\
    \de\CC'&:=\CC'(+\infty)-\CC'(1)=O(\ep_0).
\end{align*}
\end{proposition}
\begin{proof}
    We have from Proposition \ref{GagGabdecay} and Theorem \ref{thm-BHdynamic}
    \begin{equation}\label{physquant-evolve-asymp}
        \begin{split}
        \pr_{u'}\EE'&=O\left(\frac{\ep_0}{{u'}^{2+2\dec
        }}\right),\\
        \pr_{u'}\PP'&=O\left(\frac{\ep_0}{{u'}^{2+2\dec
        }}\right),\\
        \pr_{u'}\JJ'&=O\left(\frac{\ep_0}{{u'}^{\frac{3}{2}+2\dec
        }}\right),\\
        \pr_{u'}\CC'&=2\PP+O\left(\frac{\ep_0}{{u'}^{\frac{3}{2}+2\dec
        }}\right).
        \end{split}
    \end{equation}
    Integrating the first three equations of \eqref{physquant-evolve-asymp}, we deduce
    \begin{align}\label{othermemory}
        \de\EE',\de\PP',\de\JJ'=O(\ep_0).
    \end{align}
    Next, we have from \eqref{GCMS2limit}
    \begin{align*}
        \CC'(+\infty)=\lim_{u'\to\infty}\CC'(u')=\lim_{u' \to\infty}\Bk'(u')=0.
    \end{align*}
    Thus, we deduce from the last equation of \eqref{physquant-evolve-asymp} that\footnote{Notice that $C'(+\infty)=0$ implies that $\pr_{u'}\CC'(+\infty)=0$.}
    \begin{align*}
        \PP'(+\infty)=\lim_{u'\to\infty}\PP'(u')=0.
    \end{align*}
    We then integrate \eqref{physquant-evolve-asymp} backwardly to deduce
    \begin{align*}
        \PP'(u')=\PP'(+\infty)-\int^{\infty}_{u'}O\left(\frac{\ep_0}{{u''}^{2+2\dec
        }}\right)du''=O\left(\frac{\ep_0}{{u'}^{1+2\dec
        }}\right).
    \end{align*}
    Integrating the last equation of \eqref{physquant-evolve-asymp} once again, we obtain
    \begin{align*}
        \de\CC'=\int^{+\infty}_{1}\left(O\left(\frac{\ep_0}{{u'}^{1+2\dec
        }}\right)+O\left(\frac{\ep_0}{{u'}^{\frac{3}{2}+2\dec
        }}\right)\right)du'=O(\ep_0).
    \end{align*}
    Combining with \eqref{othermemory}, this concludes the proof of Proposition \ref{totalflux}.
\end{proof}
\begin{rk}
The LGCM foliations along future null infinity enable the definition of well-behaved physical quantities and provide a coherent framework for measuring the total flux of conserved charges throughout the evolution. Moreover, it selects a future inertial reference frame that encapsulates the stabilizing mechanism of the Einstein vacuum evolution. Consequently, although the final center of mass may deviate significantly from that of the reference Kerr spacetime, (see Section \ref{section:recoil} below), the total shift in the center of mass during the Einstein vacuum evolution remains small, reflecting the underlying Kerr stability result.
\end{rk}
\subsection{Gravitational wave recoil}\label{section:recoil}
\begin{df}\label{dfS1rS'1r'}
Let $(\M,\g)$ be the final spacetime derived in the proof of Theorem \ref{MainThm-firstversion},\footnote{Recall that the final spacetime $(\M,\g)$ derived Theorem \ref{MainThm-firstversion} is a particular example of $\KSAF$ spacetime. Thus, all the previous results hold for $(\M,\g)$.} endowed with an outgoing PG $S(u,r)$--foliation and a LGCM foliation $S'(u',r')$ constructed in Theorem \ref{LGCMconstruction}. Moreover, we assume that $u'$ is fixed as in Definition \ref{fixu'}. We then define\footnote{Notice that we have by definition $S'_1(\infty)=S'_1$.}
\begin{align*}
    S'_1(r'):=S'(1,r').
\end{align*}
Moreover, we assume that the initial layer region $\LL_0$ is endowed with a foliation $S^{(0)}(u^{(0)},r^{(0)})$ and a null frame $\left(e_3^{(0)},e_4^{(0)},e_1^{(0)},e_2^{(0)}\right)$. We denote by $S^{(0)}_1$ a sphere of the initial layer foliation that has a common point with $S'_1(r')$. 
\end{df}
The following theorem compares the difference of the center of mass at $S'_1(r')$ and $S^{(0)}_1$.
\begin{thm}\label{kickthm}
Let $(\M,\g)$ be the final spacetime derived in the proof of Theorem \ref{MainThm-firstversion} and let $S'_1(r')$ and $S^{(0)}_1$ be the spheres introduced in Definition \ref{dfS1rS'1r'}. As in Definition \ref{dfquantities}, the center of mass on the spheres $S_1'(r')$ and $S^{(0)}_1$ are defined respectively as follows:
\begin{align}\label{dfCCCC'}
    \CC'_1:=r^5(\div'\b')_{\ell=1,S'_1(r')},\qquad \CC^{(0)}_1:=r^5(\div\b)_{\ell=1,S^{(0)}_1}.
\end{align}
Then, we have for $r\gg \ep_0^{-3}$
\begin{align}\label{kickformula}
    \CC'_1-\CC_1^{(0)}=O\left(\ep_0r^{\frac{1}{2}-\dec}\right).
\end{align}
\end{thm}
\begin{proof}
    We first have from Proposition \ref{totalflux}
    \begin{align}\label{CCest}
        \CC'(+\infty)=0,\qquad\quad\CC'_1=O(\ep_0).
    \end{align}
     We also have from the initial assumption of \cite{KS:main}
    \begin{align}\label{divb0divbK}
        \div^{(0)}\b^{(0)}-\div_{Kerr}\b_{Kerr}=O\left(\frac{\ep_0}{{r}^{\frac{9}{2}+\dec}}\right),
    \end{align}
    with $\div_{Kerr}\b_{Kerr}$ denotes the Kerr value. We recall from Lemma 2.52 in \cite{KS:main} that
    \begin{align*}
        \div_{Kerr}\b_{Kerr}=O(r^{-6}).
    \end{align*}
    Combining with \eqref{dfCCCC'}, we obtain for $r\gg\ep_0^{-3}$
    \begin{align*}
    \big|\CC^{(0)}_1\big|\les r^{-1}+r^5\frac{\ep_0}{r^{\frac{9}{2}+\dec}}\les \ep_0r^{\frac{1}{2}-\dec}.
    \end{align*}
    Combining with \eqref{CCest}, we obtain that \eqref{kickformula} hold for $r\gg \ep_0^{-3}$. This concludes the proof of Theorem \ref{kickthm}.
\end{proof}
\begin{rk}\label{anormalfb}
Theorem \ref{kickthm} reflects the fact that the LGCM foliation $S'(u',r')$ captures the center of mass frame of the final Kerr solution, while the initial data layer foliation $S^{(0)}(u^{(0)},r^{(0)})$ captures the center of mass frame of the initial Kerr solution. The behavior of $\CC'_1-\CC_1^{(0)}$ is consistent with the presence of a Lorentz boost between these two centers of mass frames. More precisely, denoting $(f,\fb,\la)$ the transition functions from the initial layer frame $\left(e_3^{(0)},e_4^{(0)},e_1^{(0)},e_2^{(0)}\right)$ to the LG frame $(e'_3,e'_4,e'_1,e'_2)$, we have from Step 13 of Theorem M0 in \cite{KS:main}
\begin{align*}
    |f|\les \frac{\ep_0}{r},\qquad\quad|(\fb,\ovla)|\les \ep_0.
\end{align*}
The anomalous behavior of $\fb$ and $\ovla$ reflects that there is a large displacement between the two centers of mass, which is related to the so-called black hole kick in physical literature, see for example \cite{Fitchett,HHS,VV}.
\end{rk}
\begin{rk}\label{explainpicture}
In the proof of Theorem \ref{kickthm}, \eqref{divb0divbK} implies that the difference of the center of mass of the initial data and that of the reference Kerr initial data has size $O(\ep_0r^{\frac{1}{2}-\dec})$. However, \eqref{totalflux} shows that the difference in the center of mass of the initial data and that of the final data only has size $O(\ep_0)$. See Figure \ref{figurekick} for a schematic illustration of the three centers of mass $\CC_{Kerr}$, $\CC^{(0)}_1$ and $\CC'_1$.
\begin{figure}[H]
\centering
    \begin{tikzpicture}[scale=1.3]
    \node at (4.9, 4.8) {Future inertial frame};
    \node at (-2.5, -0.1) {$\CC_{Kerr}$};
    \node at (-2.1, -0.9) {Reference Kerr initial data};
    \node at (0.85, 0) {$O\left(\ep_0 r^{\frac{1}{2}-\dec}\right)$};
    \node at (4.8, -0.9) {Perturbed initial data};
    \node at (4.1, -0.1) {$\CC^{(0)}_1$};
    \node at (5.05, -0.45) {$\CC'_1$};
    \filldraw (-2.5, -0.5) circle (2pt);
    \filldraw (4.5, -0.5) circle (2pt);
    \draw[->] (-2.4, -0.5) -- (4.4, -0.5);
    \draw[dashed] (4.5, -0.5) -- (4.5, 4.5);
    \draw[dashed] (4.8, -0.5) -- (4.8, 4.5);
    \draw[red, ->, domain=-0.5:4.2, samples=200] plot ({4.6 + 0.2*cos(3*\x r)}, {\x});
    \node[red] at (4.7, 2.2) {evolution};
    \filldraw (4.8, 4.35) circle (2pt);
    \node at (5.45, 4.25) {$\CC'(+\infty)$};
    \draw[->] (3.8, 3.2) -- (4.65, 4.35);
    \node at (3.6, 3) {$O(\ep_0)$};
    \draw [decorate,decoration={brace,amplitude=5pt,mirror},xshift=-0.2cm] (4.7, 4.5) -- (5, 4.5) node[midway,below,yshift=-5pt] {};
\end{tikzpicture}
\caption{\small Comparison of three centers of mass $\CC_{Kerr}$, $\CC^{(0)}_1$ and $\CC'_1$. The center of mass of the perturbed initial data, denoted by $\CC^{(0)}_1$, is far away from that of the reference Kerr initial data, denoted by $\CC_{Kerr}$. However, $\CC^{(0)}_1$ is close to the center of mass of the final state, denoted by $\CC'_1$.}
\label{figurekick}
\end{figure}
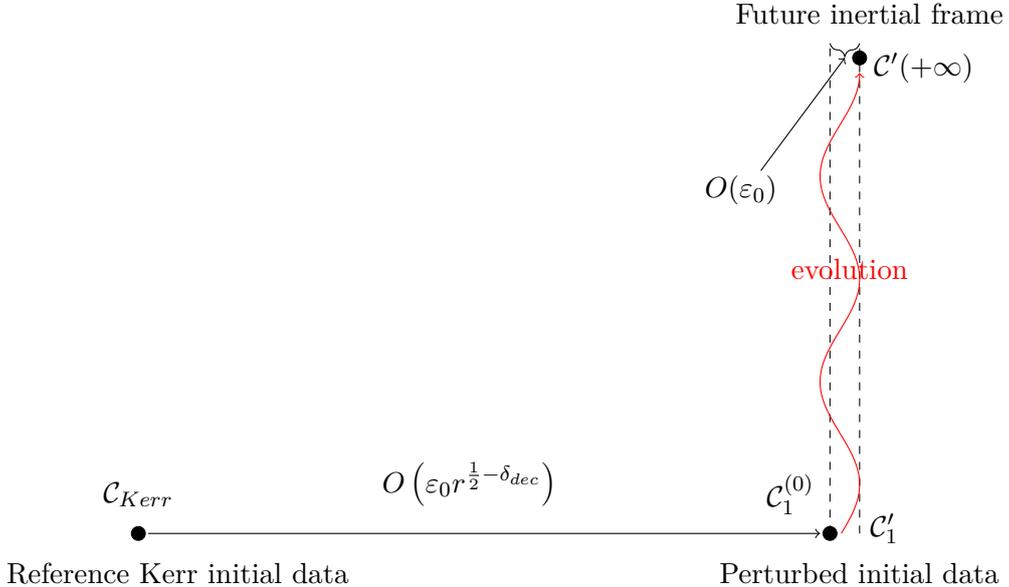
\end{rk}
\begin{thm}\label{kickthmgeneral}
Let $(\M,\g)$ be a perturbation of a particular solution of the Einstein vacuum equation,\footnote{Here, $(\M,\g)$ can be the final spacetime derived in various results of stability in the literature.} endowed with a sphere foliation and a null frame $(e_3,e_4,e_1,e_2)$. Let $\Gag$ and $\Gab$ be defined as in Definition \ref{dfGagGab}. Let $s>3$, $q\in\NNN$\footnote{The decay parameter $s$ and the regularity parameter $q$ are consistent to the $(s,q)$--asymptotic flat initial data introduced in \cite{Shen22,Shen23,Shen24}.} and $0<\ep_0\ll 1$ and we assume that the following decay estimates hold:
\begin{align*}
    \sup_{\M}|\dk^{\leq q}\Gag|\les\frac{\ep_0}{r^2u^\frac{s-3}{2}},\qquad \sup_{\M}|\dk^{\leq q}\Gab|\les\frac{\ep_0}{ru^\frac{s-1}{2}},
\end{align*}
We also define $S_1'(r')$ and $S^{(0)}$ as in Definition \ref{dfS1rS'1r'}. Then, the center of mass $\CC'_1$ and $\CC^{(0)}_1$ on the spheres $S'_1(r')$ and $S^{(0)}$ respectively, defined as in Theorem \ref{kickthm}, satisfy the following estimate for $r$ large enough:
\begin{align}\label{kickgeneral}
\Big|\CC'_1-\CC^{(0)}_1\Big|\les \Bigg\{ 
\begin{aligned}
    \ep_0 & r^{\frac{5-s}{2}} \qquad 3<s<5,\\
    \ep_0 & \qquad\qquad\quad\;\;\, s>5.
\end{aligned}
\end{align}
\end{thm}
\begin{rk}
Various choices of the parameter $s$ have been made in the literature on the stability of Schwarzschild and Kerr. In particular:
\begin{itemize}
\item $s>3$ is used in \cite{ShenKerr} for the Kerr stability in the external region.
\item $s=4+2\dec$ is used in \cite{KS} for the stability of Schwarzschild and in the series of works \cite{KS:Kerr1,KS:Kerr2,KS:main,GKS,Shen} on the stability of slowly rotating Kerr black holes.
\item $s=6+2\dec$ is used in \cite{DHRT} for the stability of Schwarzschild.
\item $s>7$ is used in \cite{Caciotta} for Kerr stability in the external region.
\end{itemize}
\end{rk}
\begin{rk}
    Theorem \ref{kickthmgeneral} shows that under stronger decay assumptions on the initial data, that is, in case $s>5$, the displacement of the center of mass has size $O(\ep_0)$, which refers to the disappearance of the black hole kick. In other words, the three centers of mass described in Figure \ref{figurekick} are close to each other in the case $s>5$. However, in the case of $s\in(3,5)$\footnote{The series works \cite{KS:Kerr1,KS:Kerr2,Shen,KS:main,GKS} shows Kerr stability for small angular momentum hold in the case of $s=4+2\dec$. However, we believe that similar results may hold in the general case $s>3$.}, the center of mass of the perturbed initial data can have a large displacement, as in Figure \ref{figurekick}. This refers to the non triviality of the black hole kick.
\end{rk}
\begin{proof}[Proof of Theorem \ref{kickthmgeneral}]
Proceeding as in Proposition \ref{totalflux}, we have
\begin{align}\label{CCestgeneral}
    \CC'(+\infty)=0,\qquad\quad\CC'_1=O(\ep_0).
\end{align}
Next, we have from the initial assumption
\begin{align*}
    \div^{(0)}\b^{(0)}=O\left(\frac{\ep_0}{r^{\frac{s+5}{2}}}\right).
\end{align*}
Hence, we have on $S_1^{(0)}$
\begin{align*}
    \left|\CC^{(0)}_1\right|\les\left|(r^5\div^{(0)}\b^{(0)})_{\ell=1}\right|\les\frac{\ep_0r^5}{r^{\frac{s+5}{2}}}\les\ep_0 r^{\frac{5-s}{2}}.
\end{align*}
Combining with \eqref{CCestgeneral}, we obtain that \eqref{kickgeneral} hold for $r$ large enough. This concludes the proof of Theorem \ref{kickthmgeneral}.
\end{proof}
\appendix
\section{Asymptotic behavior of geometric quantities}\label{section:appendix}
\subsection{Main equations in general setting}\label{appmain}
We state in the following the null structure equations and the Bianchi equations in the general setting\footnote{That is we make here no gauge conditions.}. We assume given a vacuum spacetime endowed with a general null frame $(e_3,e_4,e_1,e_2)$ relative to which we define our connection and curvature coefficients.
\begin{proposition}\label{prop-nullstr}
The connection coefficients verify the following equations:
\begin{align*}
\nab_3\trchb&=-|\hchb|^2-\frac 1 2 \big(\trchb^2-\atrchb^2\big)+2\div\xib-2\omb\trchb +  2 \xib\c(\eta+\etab-2\ze),\\
\nab_3\atrchb&=-\trchb\atrchb +2\curl\xib-2\omb\atrchb+2\xib\wedge(-\eta+\etab+2\ze),\\
\nab_3\hchb&=-\trchb\,\hchb+\nab\hot\xib-2\omb\hchb+\xib\hot(\eta+\etab-2\ze)-\aa,\\
\nab_3\trch
&= -\hchb\c\hch -\frac 1 2 \trchb\trch+\frac 1 2 \atrchb\atrch+2\div\eta+2\omb\trch + 2 \big(\xi\c \xib +|\eta|^2\big)+ 2\rho,\\
\nab_3\atrch
&=-\hchb\wedge\hch-\frac 1 2(\atrchb \trch+\trchb\atrch)+2\curl\eta+2\omb\atrch + 2 \xib\wedge\xi-2\dual\rho,\\
\nab_3\hch&=-\frac 1 2 \big( \trch\,\hchb+\trchb\,\hch\big)-\frac 1 2\big(-\dual \hchb \,\atrch+\dual\hch\,\atrchb\big)+\nab\hot\eta+2\omb\hch+\xib\hot\xi+\eta\hot\eta,\\
\nab_4\trchb&= -\hch\c\hchb -\frac 1 2 \trch\trchb+\frac 1 2 \atrch\atrchb+2\div\etab+2 \om \trchb+2\big( \xi\c \xib +|\etab|^2\big)+2\rho,\\
\nab_4\atrchb&=-\hch\wedge\hchb-\frac 1 2(\atrch \trchb+\trch\atrchb)+ 2 \curl \etab + 2 \om \atrchb + 2 \xi\wedge\xib+2 \dual \rho,\\
\nab_4\hchb&=-\frac 1 2 \big( \trchb\,\hch+\trch\,\hchb\big)-\frac 1 2 \big(-\dual \hch \, \atrchb+\dual \hchb\,\atrch\big)
+\nab\hot \etab +2 \om \hchb+ \xi\hot\xib + \etab\hot\etab,\\
\nab_4\trch&=-|\hch|^2-\frac 1 2 \big(\trch^2-\atrch^2\big)+2\div\xi- 2 \om \trch + 2  \xi\c(\etab+\eta+2\ze),\\
\nab_4\atrch&=-\trch\atrch+2\curl\xi-2\om\atrch+ 2\xi\wedge(-\etab+\eta-2\ze),\\
\nab_4\hch&=-\trch\,\hch+\nab\hot\xi-2\om\hch+\xi\hot(\etab+\eta+2\ze)-\a.
\end{align*}
Also,
\begin{align*}
\nab_3 \ze+2\nab\omb&= -\hchb\c(\ze+\eta)-\frac{1}{2}\trchb(\ze+\eta)-\frac{1}{2}\atrchb(\dual\ze+\dual\eta)+ 2 \omb(\ze-\eta)\\
&+\hch\c\xib+\frac{1}{2}\trch\,\xib+\frac{1}{2}\atrch\dual\xib +2\om \xib -\bb,
\\
\nab_4 \ze -2\nab\om&= \hch\c(-\ze+\etab)+\frac{1}{2}\trch(-\ze+\etab)+\frac{1}{2}\atrch(-\dual\ze+\dual\etab)+2 \om(\ze+\etab)\\
& -\hchb\c\xi -\frac{1}{2}\trchb\,\xi-\frac{1}{2}\atrchb\dual\xi-2 \omb \xi -\b,
\\
\nab_3\etab-\nab_4\xib &= -\hchb\c(\etab-\eta) -\frac{1}{2}\trchb(\etab-\eta)+\frac{1}{2}\atrchb(\dual\etab-\dual\eta) -4 \om \xib +\bb, \\
\nab_4\eta-\nab_3\xi &= -\hch\c(\eta-\etab) -\frac{1}{2}\trch(\eta-\etab)+\frac{1}{2}\atrch(\dual\eta-\dual\etab)-4\omb \xi -\b,\\
\end{align*}
and
\begin{align*}
\nab_3\om+\nab_4\omb-4\om\omb-\xi\c\xib-(\eta-\etab)\c\ze+\eta\c\etab=\rho.
\end{align*}
Also,
\begin{align*}
\div\hch +\ze\c\hch&=\frac{1}{2}\nab\trch+\frac{1}{2}\trch\,\ze -\frac{1}{2}\dual\nab\atrch-\frac{1}{2}\atrch\dual\ze -\atrch\dual\eta-\atrchb\dual\xi-\b,\\
\div\hchb -\ze\c\hchb&=\frac{1}{2}\nab\trchb-\frac{1}{2}\trchb\,\ze -\frac{1}{2}\dual\nab\atrchb+\frac{1}{2}\atrchb\dual\ze-\atrchb\dual\etab-\atrch\dual\xib +\bb,
\end{align*}
and\footnote{Note that this equation follows from expanding $\R_{34ab}$.}
\begin{align*}
\curl\ze&=-\frac 1 2\hch\wedge\hchb+\frac 1 4\big(\trch\atrchb-\trchb\atrch\big)+\om \atrchb-\omb\atrch+\dual \rho.
\end{align*}
\end{proposition}
\begin{proof}
See Proposition 2.2.5 in \cite{GKS}.
\end{proof}
\begin{proposition}\label{prop:bianchi} 
The curvature components verify the following equations:
\begin{align*}
\nab_3\a-\nab\hot \b&=-\frac 1 2 \big(\trchb\,\a+\atrchb\dual\a)+4\omb\a+(\ze+4\eta)\hot \b - 3 (\rho\hch +\rhod\dual\hch),\\
\nab_4\b-\div\a &=-2(\trch\,\b-\atrch\dual\b)-2\om\b +\a\c(2 \ze +\etab)+3  (\xi\rho+\dual \xi\rhod),\\
\nab_3 \b-(\nab\rho+\dual\nab\rhod)&=-(\trchb\,\b+\atrchb\dual\b)+2\omb\,\b+2\bb\c \hch+3 (\rho\eta+\rhod\dual \eta)+\a\c\xib, \\
\nab_4 \rho-\div \b&=-\frac 3 2 (\trch\,\rho+\atrch \rhod)+(2\etab+\ze)\c\b-2\xi\c\bb-\frac 1 2 \hchb \c\a,\\
\nab_4 \rhod+\curl\b&=-\frac 3 2 (\trch\,\rhod-\atrch \rho)-(2\etab+\ze)\c\dual \b-2\xi\c\dual \bb+\frac 1 2 \hchb \c\dual \a, \\
\nab_3 \rho+\div\bb&=-\frac 3 2 (\trchb\,\rho -\atrchb \rhod)-(2\eta-\ze) \c\bb+2\xib\c\b-\frac{1}{2}\hch\c\aa,\\
\nab_3 \rhod+\curl\bb&=-\frac 3 2 (\trchb\,\rhod+\atrchb \rho)-(2\eta-\ze) \c\dual \bb-2\xib\c\dual\b-\frac 1 2 \hch\c\dual \aa,\\
\nab_4\bb+\nab\rho-\dual\nab\rhod&=-(\trch\,\bb+\atrch\dual\bb)+ 2\om\,\bb+2\b\c \hchb-3 (\rho\etab-\rhod\dual \etab)-\aa\c\xi,\\
\nab_3\bb +\div\aa &=-2(\trchb\,\bb-\atrchb\dual\bb)-2\omb\bb-\aa\c(-2\ze+\eta)-3(\xib\rho-\dual \xib \rhod),\\
\nab_4\aa+ \nab\hot \bb&=-\frac 1 2 \big(\trch\,\aa+\atrch\dual \aa)+4\om \aa+
 (\ze-4\etab)\hot \bb - 3(\rho\hchb -\rhod\dual\hchb).
\end{align*}
\end{proposition} 
\begin{proof}
See Proposition 2.2.6 in \cite{GKS}.
\end{proof}
\subsection{Proof of Theorem \ref{expansionexist}}\label{secTaylor}
The following fundamental lemma is useful for treating the remainder terms of geometric quantities of a Taylor expansion of $r^{-1}$.
\begin{lem}\label{remainder}
Let $p>0$, $m\geq 0$ and let $U$ be a tensor field defined on a $\KSAF$ spacetime $(\M,\g)$ which satisfies
\begin{align}\label{scbehavior}
    \lim_{C_u,r\to\infty}r^mU=0,\qquad|\nab_4(r^mU)|\les r^{-p-1}.
\end{align}
Then, we have
\begin{align}\label{O1term}
    U=\OO_1(r^{-m-p}).
\end{align}
\end{lem}
\begin{proof}
We have from \eqref{scbehavior}
\begin{align*}
r^mU\les\lim_{C_u,r\to\infty}r^mU+\int_r^\infty\nab_4(s^mU)ds\les\int_r^\infty s^{-p-1}ds\les r^{-p},
\end{align*}
which implies
\begin{align*}
    |U|\les r^{-m-p}.
\end{align*}
Moreover, we have
\begin{align*}
    |\nab_4 U|\leq |r^{-m}\nab_4(r^mU)|+|r^{-m}e_4(r^m)U|\les r^{-m-p-1}.
\end{align*}
This concludes the proof of Lemma \ref{remainder}.
\end{proof}
We are now ready to prove Theorem \ref{expansionexist}.
\begin{proof}[Proof of Theorem \ref{expansionexist}]
The Taylor expansions in \eqref{GagTaylor} follow directly from Lemmas \ref{limitexist} and \ref{remainder} and \eqref{nab4eqimportant}--\eqref{r3nab4omb}. We now focus on \eqref{GabTaylor}.\\ \\
We have from Proposition \ref{prop-nullstr}
\begin{align*}
\nab_4(r\trchb)=-r\hch\c\hchb-\frac{r}{2}\trchc\trchb+\frac{r}{2}\atrch\atrchb-2r\div\ze+2r|\ze|^2+2r\rho,
\end{align*}
which implies from \eqref{GagTaylor} and \eqref{limitidentities}
\begin{align*}
    r^2\nab_4\left(r\trchb+2-\frac{\Xbscr}{r}\right)&=-r^3\hch\c\hchb-\frac{r^3}{2}\trchc\trchb-2r^3\div\ze+2r^3\rho+\Xbscr+O(r^{-1})\\
    &=-\The\c\Thb+\Xscr-2\divo\Zscr+2\Pscr+\Xbscr+\OO_1\left(r^{-\frac{1}{2}-\dec}\right)+\OO_0(r^{-1})\\
    &=\OO_0\left(r^{-\frac{1}{2}-\dec}\right).
\end{align*}
Hence, we deduce from Lemma \ref{remainder} 
\begin{align*}
    \trchb=-\frac{2}{r}+\frac{\Xbscr}{r^2}+\OO_1\left(r^{-\frac{5}{2}-\dec}\right).
\end{align*}
Similarly, we have from Proposition \ref{prop-nullstr}
\begin{align*}
\nab_4(r\atrchb)=-r\hch\wedge\hchb-\frac{r}{2}\left(\atrch \trchb+\trchc\atrchb\right)-2r\curl\ze+2r\dual\rho,
\end{align*}
which implies from \eqref{GagTaylor} and \eqref{limitidentities}
\begin{align*}
    r^2\nab_4\left(r\atrchb-\frac{\aXbscr}{r}\right)&=-\The\wedge\Thb+\aXscr-2\curlo\Zscr+2\dual\Pscr+\aXbscr\\
    &+\OO_1\left(r^{-\frac{1}{2}-\dec}\right)+\OO_0(r^{-1})\\
    &=\OO_0\left(r^{-\frac{1}{2}-\dec}\right).
\end{align*}
Hence, we deduce from Lemma \ref{remainder} 
\begin{align*}
    \atrchb=\frac{\aXbscr}{r^2}+\OO_1\left(r^{-\frac{5}{2}-\dec}\right).
\end{align*}
Next, we have from Proposition \ref{prop-nullstr}
\begin{align*}
    \nab_4\omb=\rho+2\eta\c\ze+\ze\c\ze.
\end{align*}
Combining with Proposition \ref{limitexist}, \eqref{GagTaylor} and \eqref{limitidentities}, we deduce
\begin{align*}
    r^3\nab_4\left(\omb-\frac{\Wbscr}{r^2}\right)&=r^3\rho+2r^3\eta\c\ze+r^3\ze\c\ze+2\Wbscr\\
    &=\Pscr+2\Hscr\Zscr+\OO_1\left(r^{-\frac{1}{2}-\dec}\right)+\OO_0(r^{-1})+2\Wbscr\\
    &=\OO_0\left(r^{-\frac{1}{2}-\dec}\right).
\end{align*}
Applying Lemma \ref{remainder}, we obtain
\begin{align}\label{ombTaylor}
    \omb=\frac{\Wbscr}{r^2}+\OO_1\left(r^{-\frac{5}{2}-\dec}\right).
\end{align}
Next, we have from Proposition \ref{prop-nullstr}
\begin{align*}
\nab_4\hchb=-\frac{1}{2}\big( \trchb\,\hch+\trch\,\hchb\big)-\frac 1 2\big(-\dual \hch \, \atrchb+\dual \hchb\,\atrch\big)-\nab\hot\ze+\ze\hot\ze.
\end{align*}
Thus, we obtain for any $2$-tensor $\Thb_2$, which only depends on $(u,\th^1,\th^2)$,
\begin{align*}
r^2\nab_4\left(r\hchb-\Thb-\frac{\Thb_2}{r}\right)&=-\frac{r^3}{2}\left(\trchb\,\hch+\trchc\,\hchb\right)-\frac{r^3}{2}\big(-\dual \hch \, \atrchb+\dual \hchb\,\atrch\big)\\
&-r^3\nab\hot\ze+r^3\ze\hot\ze+\Thb_2.
\end{align*}
Combining with \eqref{GagTaylor}, we deduce
\begin{align*}
    r^2\nab_4\left(r\hchb-\Thb-\frac{\Thb_2}{r}\right)=\The-\frac{1}{2}\Xscr\Thb-\frac{1}{2}\aXscr\Thb-\nabo\hot\Zscr+\Thb_2+\OO_1\left(r^{-\frac{1}{2}-\dec}\right)+\OO_0(r^{-1}).
\end{align*}
Taking
\begin{align*}
    \Thb_2:=-\The+\frac{1}{2}\Xscr\Thb+\frac{1}{2}{}^*\Thb\,\aXscr+\nabo\hot\Zscr,
\end{align*}
we deduce from Lemma \ref{remainder}
\begin{align*}
    \hchb=\frac{\Thb}{r}+\frac{-\The+\frac{1}{2}\Xscr\Thb+\frac{1}{2}\dual\Thb\,\aXscr+\nabo\hot\Zscr}{r^2}+\OO_1\left(r^{-\frac{5}{2}-\dec}\right).
\end{align*}
Similarly, we have from Proposition \ref{prop-nullstr}
\begin{align*}
\nab_4\eta+\frac{1}{2}\trch\,\eta=-\b-\hch\c(\eta+\ze)-\frac{1}{2}\trch\,\ze+\frac{1}{2}\atrch(\dual\eta+\dual\ze),
\end{align*}
which implies for any $1$--form $\Hscr_2$, which only depends on $(u,\th^1,\th^2)$,
\begin{align*}
    r^2\nab_4\left(r\eta-\Hscr-\frac{\Hscr_2}{r}\right)&=-r^3\b-\frac{r^3}{2}\trchc\,\eta-r^3\hch\c(\eta+\ze)\\
    &-\frac{r^3}{2}\trch\,\ze+\frac{r^3}{2}\atrch(\dual\eta+\dual\ze)+\Hscr_2.
\end{align*}
Combining with \eqref{GagTaylor}, we deduce
\begin{align*}
    r^2\nab_4\left(r\eta-\Hscr-\frac{\Hscr_2}{r}\right)=\OO_0(r^{-\frac{1}{2}-\dec})-\frac{1}{2}\Xscr\Hscr-\The\c\Hscr-\Zscr+\frac{1}{2}\aXscr\dual\Hscr+\Hscr_2.
\end{align*}
Taking
\begin{align*}
    \Hscr_2:=\frac{1}{2}\Xscr\Hscr+\The\c\Hscr+\Zscr-\frac{1}{2}\aXscr\dual\Hscr,
\end{align*}
we deduce from Lemma \ref{remainder}
\begin{align*}
    \eta=\frac{\Hscr}{r}+\frac{\frac{1}{2}\Xscr\c\Hscr+\The\c\Hscr+\Zscr-\frac{1}{2}\aXscr\,\dual\Hscr}{r^2}+\OO_1\left(r^{-\frac{5}{2}-\dec}\right).
\end{align*}
Next, we have from Proposition \ref{prop-nullstr}
\begin{align*}
\nab_4\xib+\frac{1}{2}\trch\,\xib=2\nab\omb+\atrchb(\dual\ze+\dual\eta)-\hch\c\xib-\frac{1}{2}\atrch\dual\xib.
\end{align*}
Hence, we obtain for any $1$--form $\Ybscr_2$, which only depends on $(u,\th^1,\th^2)$,
\begin{align*}
r^2\nab_4\left(r\xib-\Ybscr-\frac{\Ybscr_2}{r}\right)&=-\frac{r^3}{2}\trchc\,\xib+2r^3\nab\omb+r^3\atrchb(\dual\ze+\dual\eta)\\
&-r^3\hch\c\xib-\frac{r^3}{2}\atrch\dual\xib+\Ybscr_2.
\end{align*}
Combining with \eqref{GagTaylor} and \eqref{ombTaylor}, we deduce
\begin{align*}
    r^2\nab_4\left(r\xib-\Ybscr-\frac{\Ybscr_2}{r}\right)=-\frac{1}{2}\Xscr\Ybscr+2\nabo\Wbscr+\aXbscr\dual\Hscr-\The\c\Ybscr-\frac{1}{2}\aXscr\dual\Ybscr+\Ybscr_2.
\end{align*}
Taking
\begin{align*}
    \Ybscr_2:=\frac{1}{2}\Xscr\Ybscr-2\nabo\Wbscr-\aXbscr\dual\Hscr+\The\c\Ybscr+\frac{1}{2}\aXscr\dual\Ybscr,
\end{align*}
we infer from Lemma \ref{remainder}
\begin{align*}
    \xib=\frac{\Ybscr}{r}+\frac{\frac{1}{2}\Xscr\Ybscr-2\nabo\Wbscr-\aXbscr\,\dual\Hscr+\The\c\Ybscr+\frac{1}{2}\aXscr\,\dual\Ybscr}{r^2}+\OO_1\left(r^{-\frac{5}{2}-\dec}\right).
\end{align*}
Next, we have from Proposition \ref{prop:bianchi}
\begin{align*}
    \nab_4(r^2\bb)=-r^2(\nab\rho-\dual\nab\dual\rho)-r^2(\trchc\,\bb+\atrch\dual\bb)+2r^2\b\c\hchb+3r^2(\rho\ze-\rhod\dual\ze),
\end{align*}
which implies for any $1$--form $\Bbscr_3$, which only depends on $(u,\th^1,\th^2)$,
\begin{align*}
    r^2\nab_4\left(r^2\bb-\Bbscr-\frac{\Bbscr_3}{r}\right)&=-r^4(\nab\rho-\dual\nab\rhod)-r^4(\trchc\,\bb+\atrch\,\dual\bb)\\
    &+2r^4\b\c\hchb+3r^4(\rho\ze-\rhod\dual\ze)+\Bbscr_3\\
    &=-\nabo\Pscr+\dual\nabo\,\dual\Pscr-\Xscr\Bbscr-\aXscr\,\dual\Bbscr+\Bbscr_3+\OO_0\left(r^{-\frac{1}{2}-\dec}\right).
\end{align*}
Taking
\begin{align}\label{dfBbscr3}
    \Bbscr_3:=\nabo\Pscr-\dual\nabo\,\dual\Pscr+\Xscr\Bbscr+\aXscr\,\dual\Bbscr,
\end{align}
we deduce from Lemma \ref{remainder}
\begin{align*}
    \bb=\frac{\Bbscr}{r^2}+\frac{\nabo\Pscr-\dual\nabo\,\dual\Pscr+\Xscr\Bbscr+\aXscr\,\dual\Bbscr}{r^3}+\OO_1\left(r^{-\frac{7}{2}-\dec}\right).
\end{align*}
Finally, we have from Proposition \ref{prop:bianchi}
\begin{align*}
\nab_4\aa+\nab\hot\bb=-\frac 1 2 \big(\trch\,\aa+\atrch\dual \aa)+5\ze\hot\bb-3  (\rho\hchb-\rhod\dual\hchb),
\end{align*}
which implies
\begin{align*}
r^2\nab_4(r\aa)+r^3\nab\hot\bb=-\frac{r^2}{2}\left(\trchc(r\aa)+\atrch\dual(r\aa)\right)+5r^3\ze\hot \bb-3r^3(\rho\hchb -\rhod\dual\hchb).
\end{align*}
Taking $r\to\infty$ and applying Proposition \ref{limitexist}, we deduce
\begin{align*}
    \lim_{C_u,r\to\infty}r^2\nab_4(r\aa)=-\nabo\hot\Bbscr-\frac{1}{2}(\Xscr\Abscr+\aXscr\dual\Abscr).
\end{align*}
Combining with Lemma \ref{lhopital}, we obtain the existence of the following limit:
\begin{align*}
    \Abscr_2:=\lim_{C_u,r\to\infty}r(r\aa-\Abscr)=-\lim_{C_u,r\to\infty}r^2\nab_4(r\aa-\Abscr)=\nabo\hot\Bbscr+\frac{1}{2}(\Xscr\Abscr+\aXscr\dual\Abscr).
\end{align*}
Then, for any $1$--form $\Abscr_3$, which only depends on $(u,\th^1,\th^2)$, we have
\begin{align*}
&r^2\nab_4\left(r\aa-\Abscr-\frac{\Abscr_2}{r}-\frac{\Abscr_3}{r^2}\right)\\
=&-r^3\nab\hot\bb-\frac{r^3}{2}\left(\trchc\,\aa+\atrch\dual\aa\right)+5r^3\ze\hot\bb-3r^3(\rho\hchb-\rhod\dual\hchb)+\Abscr_2+\frac{2\Abscr_3}{r}\\
=&-\nabo\hot\Bbscr-\frac{\nabo\hot\Bbscr_3}{r}-\frac{\Xscr}{2}\left(\Abscr+\frac{\Abscr_2}{r}\right)-\frac{\aXscr}{2}\left(\dual\Abscr+\frac{\dual\Abscr_2}{r}\right)+\Abscr_2+\frac{2\Abscr_3}{r}\\
&+\frac{5\Zscr\hot\Bbscr}{r}-\frac{3\Pscr\Thb-3\dual\Pscr\dual\Thb}{r}+\OO_0\left(r^{-\frac{3}{2}-\dec}\right)\\
=&-\frac{\nabo\hot\Bbscr_3}{r}-\frac{\Xscr\Abscr_2+\aXscr\dual\Abscr_2}{2r}+\frac{5\Zscr\hot\Bbscr}{r}-\frac{3\Pscr\Thb-3\dual\Pscr\dual\Thb}{r}+\frac{2\Abscr_3}{r}+\OO_0\left(r^{-\frac{3}{2}-\dec}\right),
\end{align*}
where $\Bbscr_3$ is defined in \eqref{dfBbscr3}. Taking
\begin{equation*}
    \Abscr_3:=\frac{1}{2}\nabo\hot\Bbscr_3+\frac{1}{4}(\Xscr\Abscr_2+\aXscr\dual\Abscr_2)-\frac{5}{2}\Zscr\hot\Bbscr+\frac{3}{2}(\Pscr\Thb+\dual\Pscr\dual\Thb),
\end{equation*}
we deduce from Lemma \ref{remainder}
\begin{align*}
    \aa=\frac{\Abscr}{r}+\frac{\nabo\hot\Bbscr+\frac{1}{2}(\Xscr\Abscr+\aXscr\dual\Abscr)}{r^2}+\frac{\Abscr_3}{r^3}+\OO_1\left(r^{-\frac{7}{2}-\dec}\right).
\end{align*}
This concludes the proof of Theorem \ref{expansionexist}.
\end{proof}
\section{Notations and conventions}\label{summary}
Various notations of geometric quantities and their limits on $\II^+$ have been introduced throughout this paper. The table below (Table \ref{tab:null-coefficients}) summarizes all these notations for the convenience of the reader.

\begin{center}
\begin{threeparttable}[ht]
    \renewcommand{\arraystretch}{1.5}
    \caption{Null Ricci Coefficients and Null Curvature Components}
    \label{tab:null-coefficients} 
    \begin{tabular}{|c|c|c|}
        \hline
        Real Components & Complex Components & Weighted Limit towards $\II^+$ \\
        \hline
        $\hch$ & $\Xh=\hch+i\dual\hch$ & $r^2\hch\to \The$ \\
        \hline
        $(\trch,\atrch)$ & $\tr X=\trch-i\atrch$ & $r^2\left(\trch-\frac{2}{r}, \atrch\right)\to (\Xscr,\aXscr)$ \\
        \hline
        $\hchb$ & $\Xbh=\hchb+i\dual\hchb$ & $r\hchb\to \Thb$\\
        \hline
        $(\trchb,\atrchb)$ & $\tr \Xb=\trchb-i\atrchb$ & $r^2\left(\trchb+\frac{2}{r}, \atrchb\right)\to (\Xbscr,\aXbscr)$ \\
        \hline
        $\ze$ & $Z=\ze+i\dual \ze$ & $r^2\ze\to\Zscr$ \\
        \hline
        $\eta$ & $H=\eta+i\dual \eta$ & $r\eta\to\Hscr$ \\
        \hline
        $\xib$ & $\Xib=\xib+i\dual\xib$ & $r\xib\to\Ybscr$ \\
        \hline
        $(j,j_\pm)$ & $(\Jk,\Jk_\pm)=(j,j_\pm)+i\dual (j,j_\pm)$ & $r(\Jk,\Jk_\pm)\to(\Jscr,\Jscr_\pm)$ \\
        \hline
        $\omb$ &  & $\big(r\omb,r(r\omb-\Wbscrone)\big)\to(\Wbscrone,\Wbscr)$ \\
        \hline
        $\left(\mu,\mub\right)$ & & $r^3\left(\mu,\mub\right)\to (\Mscr,\Mbscr) $\\
        \hline
        $\frac{1}{2}(\mu+\mub)$ & & $\frac{r^3}{2}(\mu+\mub)\to \Mk$ \\
        \hline
        $\b$ & $B=\b+i\dual \b$ & $(r^5d_1\b)_{\ell=1}\to (\Bk,\Bkd)$ \\
        \hline
        $(\rho, \rhod)$ & $P=\rho+i\rhod$ & $r^3(\rho,\rhod)\to (\Pscr,\dual\Pscr)$ \\
        \hline
        $\bb$ & $\Bb=\bb+i\dual \bb$ & $r^2\Bb\to \Bbscr$\\
        \hline
        $\aa$ & $\Ab=\aa+i\dual \aa$ & $r\Ab\to \Abscr$\\
        \hline
    \end{tabular}
\end{threeparttable}
\end{center}
The physical quantities are defined in the LGCM foliation as follows:
\begin{align*}
    (\EE,\PP,\CC,\JJ)=(\Mk_{\ell=0},\Mk_{\ell=1},\Bk,\Bkd),
\end{align*}
which are respectively energy, linear momentum, center of mass and angular momentum.

\vspace{0.1cm}
\small{Sergiu Klainerman: Department of Mathematics, Princeton University, Princeton, NJ, 08544. \\
Email: \textit{seri@math.princeton.edu}\\ \\
Dawei Shen: Department of Mathematics, Columbia University, New York, NY, 10027. \\
Email: \textit{ds4350@columbia.edu}\\ \\
Jingbo Wan: Department of Mathematics, Columbia University, New York, NY, 10027. \\
Email: \textit{jingbowan@math.columbia.edu}}
\end{document}